\documentclass{lmcs}
\pdfoutput=1
\usepackage[utf8]{inputenc}

\usepackage{lastpage}
\lmcsdoi{22}{1}{2}
\lmcsheading{}{\pageref{LastPage}}{}{}%
{Jan.~31,~2024}{Jan.~05,~2026}{}

\usepackage{stmaryrd}

\usepackage{etoolbox}
\usepackage{xparse}

\usepackage{standalone}
\usepackage{booktabs}
\usepackage{microtype}
\usepackage{multicol}
\usepackage{mathtools}
\usepackage{mathpartir}
\usepackage{todonotes}

\usepackage{figures/tikzit}
\usepackage{tikz-cd}

\usepackage{hyperref}

\newtoggle{conf}
\newtoggle{proofs}
\settoggle{conf}{true}
\settoggle{proofs}{false}

\newcommand{\bgcolour}{white}

\newcommand{\mcc}{\mathcal{C}}

\newcommand{\mce}{\mathcal{E}}

\newcommand{\mcr}{\mathcal{R}} 

\NewDocumentCommand\proj{mo}{%
  \pi_{#1}\IfValueT{#2}{\left(#2\right)}
}

\newcommand{\lub}{\sqcup}
\newcommand{\ljoin}{\lub}
\newcommand{\glb}{\sqcap}
\newcommand{\lmeet}{\glb}

\newcommand{\tuple}[2]{#1^{#2}}

\makeatletter
\newsavebox{\@brx}
\newcommand{\llangle}[1][]{\savebox{\@brx}{\(\m@th{#1\langle}\)}%
  \mathopen{\copy\@brx\kern-0.5\wd\@brx\usebox{\@brx}}}
\newcommand{\rrangle}[1][]{\savebox{\@brx}{\(\m@th{#1\rangle}\)}%
  \mathclose{\copy\@brx\kern-0.5\wd\@brx\usebox{\@brx}}}
\makeatother

\NewDocumentCommand\inl{o}{%
    \mathsf{inl}\IfValueT{#1}{(#1)}
}
\NewDocumentCommand\inr{o}{%
    \mathsf{inr}\IfValueT{#1}{(#1)}
}

\newcommand{\nat}{\mathbb{N}}

\newcommand{\listvar}[1]{\overline{#1}}

\newcommand{\vertices}[1]{#1_\star}

\newcommand{\edges}[3]{#1_{#2,#3}}
\newcommand{\sources}[1]{\mathsf{s}_{#1}}
\newcommand{\targets}[1]{\mathsf{t}_{#1}}

\newcommand{\hyp}{\mathbf{Hyp}}
\newcommand{\hypsignature}[1]{\llbracket#1\rrbracket}
\newcommand{\hypsigma}{\hyp_{\signature}}
\newcommand{\hypsigmac}{\hyp_{C,\signature}}

\newcommand{\cspdhyp}{\csp[D]{\hypsigma}}
\newcommand{\cspdchyp}{\csp[D_C]{\hypsigmac}}

\newcommand{\macspdhyp}{\mathsf{MA}\cspdhyp}

\newcommand{\pmcspdhyp}{\mathsf{PM}\csp[D]{\hypsigma}}

\newcommand{\plmcspdhyp}{\mathsf{PLM}\csp[D]{\hypsigma}}

\NewDocumentCommand\termtohyp{om}{%
    \llbracket{\IfValueTF{#1}{#1}{-}}\rrbracket_{#2}
}
\NewDocumentCommand\termtohypsigma{o}{%
    \llbracket{\IfValueTF{#1}{#1}{-}}\rrbracket_{\Sigma}
}
\NewDocumentCommand\termtohypsigmac{o}{%
    \llbracket{\IfValueTF{#1}{#1}{-}}\rrbracket_{\Sigma,C}
}
\NewDocumentCommand\frobtohyp{om}{%
    \left[{\IfValueTF{#1}{#1}{-}}\right]_{#2}
}
\NewDocumentCommand\frobtohypsigma{o}{%
    \left[{\IfValueTF{#1}{#1}{-}}\right]_{\Sigma}
}
\NewDocumentCommand\frobtohypsigmac{o}{%
    \left[{\IfValueTF{#1}{#1}{-}}\right]_{\Sigma,C}
}
\NewDocumentCommand\tracedtosymandfrob{om}{
    \left\lfloor\IfValueTF{#1}{#1}{-}\right\rfloor^\mathbf{T}_{#2}
}
\NewDocumentCommand\tracedtosymandfrobsigma{o}{
    \left\lfloor\IfValueTF{#1}{#1}{-}\right\rfloor^\mathbf{T}_{\Sigma}
}
\NewDocumentCommand\tracedtosymandfrobsigmac{o}{
    \left\lfloor\IfValueTF{#1}{#1}{-}\right\rfloor^\mathbf{T}_{\Sigma, C}
}
\NewDocumentCommand\comonoidtofrob{o}{
    \left\lfloor\IfValueTF{#1}{#1}{-}\right\rfloor^\mathbf{C}
}
\NewDocumentCommand\tracedandcomonoidtofrob{om}{
    \left\lfloor\IfValueTF{#1}{#1}{-}\right\rfloor_{#2}
}
\NewDocumentCommand\tracedandcomonoidtofrobsigma{o}{
    \left\lfloor\IfValueTF{#1}{#1}{-}\right\rfloor_{\Sigma}
}
\NewDocumentCommand\tracedandcomonoidtofrobsigmac{o}{
    \left\lfloor\IfValueTF{#1}{#1}{-}\right\rfloor_{\Sigma, \mcc}
}
\NewDocumentCommand\comonoidtohyp{om}{%
    \IfValueTF{#1}{\frobtohyp[#1]{#2}}{\frobtohyp{#2}}^\mathbf{C}
}
\NewDocumentCommand\comonoidtohypsigma{o}{%
    \IfValueTF{#1}{\frobtohypsigma[#1]}{\frobtohypsigma}^\mathbf{C}
}
\NewDocumentCommand\termandfrobtohyp{om}{%
    \llangle{\IfValueTF{#1}{#1}{-}}\rrangle_{#2}
}
\NewDocumentCommand\termandfrobtohypsigma{o}{%
    \llangle{\IfValueTF{#1}{#1}{-}}\rrangle_{\Sigma}
}
\NewDocumentCommand\termandfrobtohypsigmac{o}{%
    \llangle{\IfValueTF{#1}{#1}{-}}\rrangle_{\Sigma,\mcc}
}
\NewDocumentCommand\foldinterfaces{o}{%
    \ulcorner{\IfValueTF{#1}{#1}{-}}\urcorner
}

\NewDocumentCommand\rewrite{o}{%
    \Rightarrow\IfValueT{#1}{_{#1}}
}
\NewDocumentCommand\grewrite{o}{%
    \rightsquigarrow\IfValueT{#1}{_{#1}}
}
\NewDocumentCommand\trgrewrite{o}{%
    \rightsquigarrow\IfValueT{#1}{_{#1}}
}
\newcommand{\rrule}[2]{\langle #1,#2 \rangle}

\newcommand{\morph}[3]{{{#1} \colon {#2 \to #3}}}
\newcommand{\seq}{{\ \fatsemi\ }}
\NewDocumentCommand\id{o}{%
    \mathsf{id}\IfValueT{#1}{_{#1}}%
}

\newcommand{\tensor}{\otimes}

\newcommand{\swap}[2]{{\sigma_{#1,#2}}}

\DeclareMathOperator{\tr}{Tr}
\NewDocumentCommand\trace{moom}{%
    {\tr^{#1}\IfValueT{#2}{_{#2, #3}}\left(#4\right)}
}

\newcommand{\dual}[1]{#1^*}
\NewDocumentCommand\ccunit{o}{
    \eta\IfValueT{#1}{_{#1}}
}
\NewDocumentCommand\cccounit{o}{
    \varepsilon\IfValueT{#1}{_{#1}}
}

\NewDocumentCommand\ccopy{o}{
    {\Delta\IfValueT{#1}{_{#1}}}
}
\NewDocumentCommand\cdel{o}{
    {\scalebox{1.3}{\(\diamond\)}\IfValueT{#1}{_{#1}}}
}
\NewDocumentCommand\cmerge{o}{
    {\nabla\IfValueT{#1}{_{#1}}}
}
\NewDocumentCommand\cinit{o}{
    {\square\IfValueT{#1}{_{#1}}}
}

\NewDocumentCommand\fix{o}{
    \mathsf{fix}\IfValueT{#1}{\left(#1\right)}
}

\newcommand{\signature}{\Sigma}
\NewDocumentCommand\generators{o}{%
    \Sigma\IfValueT{#1}{_{#1}}
}

\NewDocumentCommand\equations{o}{%
    \mathcal{E}\IfValueT{#1}{_{#1}}
}

\NewDocumentCommand\dom{o}{%
    \mathsf{dom}\IfValueT{#1}{(#1)}
}
\NewDocumentCommand\domw{o}{%
    \overline{\mathsf{dom}}\IfValueT{#1}{(#1)}
}
\NewDocumentCommand\cod{o}{%
    \mathsf{cod}\IfValueT{#1}{(#1)}
}
\NewDocumentCommand\codw{o}{%
    \overline{\mathsf{cod}}\IfValueT{#1}{(#1)}
}
\NewDocumentCommand\eqaxioms{o}{%
    \IfValueTF{#1}{\stackrel{#1}{=}}{=}
}
\NewDocumentCommand\reduction{o}{%
    \IfValueTF{#1}{\stackrel{#1}{\rightsquigarrow}}{\rightsquigarrow}
}

\NewDocumentCommand\spann{momom}{%
    {#1 \IfValueTF{#2}{\xleftarrow{#2}}{\leftarrow} #3 \IfValueTF{#4}{\xrightarrow{#4}}{\rightarrow} #5}
}
\NewDocumentCommand\cospan{momom}{%
    {#1 \IfValueTF{#2}{\xrightarrow{#2}}{\rightarrow} #3 \IfValueTF{#4}{\xleftarrow{#4}}{\leftarrow} #5}
}
\NewDocumentCommand\csp{om}{%
    \mathsf{Csp}\IfValueT{#1}{_{#1}}(#2)
}

\newcommand{\set}{\mathbf{Set}}
\newcommand{\finset}{\mathbf{FinSet}}

\NewDocumentCommand\frob{o}{%
    \mathbf{Frob}\IfValueT{#1}{_{#1}}%
}

\newcommand{\cmon}{\mathbf{CMon}}
\newcommand{\comon}{\mathbf{Comon}}
\newcommand{\ccomon}{\mathbf{CComon}}

\newcommand{\smc}[1]{\mathbf{S}_{#1}}
\newcommand{\smcsigma}{\smc{\Sigma}}
\newcommand{\smcsigmac}{\smc{\Sigma,C}}
\newcommand{\stmc}[1]{\mathbf{T}_{#1}}
\newcommand{\stmcsigma}{\stmc{\Sigma}}

\newcommand{\hypc}[1]{\mathbf{H}_{#1}}
\newcommand{\hypcsigma}{\hypc{\Sigma}}

\NewDocumentCommand\idf{o}{%
    \mathsf{Id}\IfValueT{#1}{_{#1}}%
}
\newcommand{\ifdir}{(\Leftarrow)}
\newcommand{\onlyifdir}{(\Rightarrow)}

\newcommand{\values}{\mathbf{V}}

\newcommand{\ccirc}[1]{\mathbf{CCirc}_{#1}}

\newcommand{\scirc}[1]{\mathbf{SCirc}_{#1}}
\newcommand{\scircsigma}{\scirc{\Sigma}}

\newcommand{\interpretation}{\mathcal{I}}
\NewDocumentCommand\valueinterpretation{o}{%
    \left\llbracket\IfValueTF{#1}{#1}{-}\right\rrbracket^{\mathbf{V}}
}
\NewDocumentCommand\gateinterpretation{o}{%
    \left\llbracket\IfValueTF{#1}{#1}{-}\right\rrbracket
}

\NewDocumentCommand\circuittofunc{om}{%
    \left\llbracket\IfValueTF{#1}{#1}{-}\right\rrbracket^{\mathbf{C}}_{#2}
}
\NewDocumentCommand\circuittostream{om}{%
    \left\llbracket\IfValueTF{#1}{#1}{-}\right\rrbracket^{\mathbf{S}}_{#2}
}
\NewDocumentCommand\circuittostreami{o}{%
    \circuittostream[\IfValueTF{#1}{#1}{-}]{\interpretation}%
}
\NewDocumentCommand\circuittofunci{o}{%
    \circuittofunc[\IfValueTF{#1}{#1}{-}]{\interpretation}%
}
\NewDocumentCommand\circuittomealy{om}{%
    \left[\IfValueTF{#1}{#1}{-}\right]_{#2}%
}
\NewDocumentCommand\circuittomealyi{o}{%
    \circuittomealy[\IfValueTF{#1}{#1}{-}]{\interpretation}%
}

\newcommand{\monoidunitleqn}{\mathsf{M1}}

\newcommand{\monoidassoceqn}{\mathsf{M3}}
\newcommand{\monoidcommeqn}{\mathsf{M4}}

\newcommand{\comonoidunitleqn}{\mathsf{C1}}

\newcommand{\comonoidassoceqn}{\mathsf{C3}}
\newcommand{\comonoidcommeqn}{\mathsf{C4}}

\newcommand{\frobleqn}{\mathsf{F1}}
\newcommand{\frobreqn}{\mathsf{F2}}
\newcommand{\frobspeceqn}{\mathsf{F3}}

\newcommand{\instantfeedbackeqn}{\mathsf{IF}}

\newcommand{\gateeqn}{\mathsf{Prim}_{\interpretation}}
\newcommand{\forkeqn}{\mathsf{Fork}}

\newcommand{\joineqn}{\mathsf{Join}}
\newcommand{\stubeqn}{\mathsf{Stub}}

\newcommand{\mealyeqn}{\mathsf{Mealy}}

\newcommand{\streamingeqn}{\mathsf{Str}}

\DeclareMathOperator{\andgate}{AND}
\DeclareMathOperator{\orgate}{OR}
\DeclareMathOperator{\notgate}{NOT}

\DeclareMathOperator{\norgate}{NOR}

\newcommand{\belnapfalse}{\mathsf{f}}
\newcommand{\belnaptrue}{\mathsf{t}}
\newcommand{\valuetuple}[1]{\tuple{\values}{#1}}

\newcommand{\belnap}{\mathsf{B}}

\NewDocumentCommand\belnapvalueinterpretation{o}{%
    \valueinterpretation\IfValueT{#1}{[#1]}_{\belnap}
}
\NewDocumentCommand\belnapgateinterpretation{o}{%
    \gateinterpretation\IfValueT{#1}{[#1]}_{\belnap}
}

\newcommand{\equationdisplay}[3]{
    #1
    =
    #2
    \;
    (#3)
}

\input{figures/circuits.tikzdefs}

\tikzstyle{none}=[anchor=center]
\tikzstyle{vertex}=[anchor=center, fill=black, draw=black, shape=circle, minimum size=2.5mm, tikzit category=hypergraph, inner sep=0]
\tikzstyle{red outline vertex}=[anchor=center, fill=black, draw=palered, shape=circle, minimum size=2.5mm, tikzit category=hypergraph, inner sep=0, line width={\stringwidth}]
\tikzstyle{red vertex}=[anchor=center, fill=palered, draw=palered, shape=circle, minimum size=2.5mm, tikzit category=hypergraph, inner sep=0, line width={\stringwidth}]
\tikzstyle{edge}=[anchor=center, fill=neutral, draw=black, shape=rectangle, font={\boxsize}, tikzit category=hypergraph, minimum width=8mm, minimum height=8mm, rounded corners=3mm, very thick, anchor=center]
\tikzstyle{red outline edge}=[anchor=center, fill=neutral, draw=palered, shape=rectangle, font={\boxsize}, tikzit category=hypergraph, minimum width=8mm, minimum height=8mm, rounded corners=3mm, very thick, anchor=center]
\tikzstyle{edge subgraph}=[fill=neutral, draw=black, shape=rectangle, font={\boxsize}, tikzit category=hypergraph, minimum width=8mm, minimum height=8mm, rounded corners=3mm, very thick, dashed]
\tikzstyle{port}=[minimum size=2.5mm, fill=none, draw=none, shape=circle, tikzit draw={rgb,255: red,154; green,154; blue,154}, anchor=center, tikzit fill=neutral]
\tikzstyle{product}=[fill=neutral, draw=black, shape=circle, scale=0.66, tikzit category=string diagram]
\tikzstyle{type}=[fill=none, draw=none, shape=circle, font={\large}, tikzit fill={rgb,255: red,37; green,193; blue,141}, inner sep=0, anchor=center, tikzit category=string diagram]
\tikzstyle{tiny box white}=[{\corners}, font={\boxsize}, fill=neutral, draw=black, shape=rectangle, tikzit category=string diagram, minimum width={\tinywidth}, minimum height={\tinywidth}, anchor=center, line width={\stringwidth}, tikzit fill=neutral, inner sep={\innersep}]
\tikzstyle{tiny box comb}=[{\corners}, font={\boxsize}, fill=comb, draw=black, shape=rectangle, tikzit category=string diagram, minimum width={\tinywidth}, minimum height={\tinywidth}, anchor=center, line width={\stringwidth}, tikzit fill=comb, inner sep={\innersep}]
\tikzstyle{tiny box seq}=[{\corners}, font={\boxsize}, fill=seq, draw=black, shape=rectangle, tikzit category=string diagram, minimum width={\tinywidth}, minimum height={\tinywidth}, anchor=center, line width={\stringwidth}, inner sep={\innersep}]
\tikzstyle{tiny signal seq}=[{\corners}, font={\boxsize}, shape=signal, signal to=west, signal pointer angle=110, fill=seq, draw=black, tikzit category=string diagram, minimum width=6mm, minimum height=5mm, anchor=center, line width={\stringwidth}, inner sep={\innersep}]
\tikzstyle{small box white}=[{\corners}, font={\boxsize}, fill=neutral, draw=black, shape=rectangle, minimum height={\smallwidth}, minimum width={\tinywidth}, tikzit category=string diagram, anchor=center, line width={\stringwidth}, inner sep={\innersep}]
\tikzstyle{small square box white}=[{\corners}, font={\boxsize}, fill=neutral, draw=black, shape=rectangle, minimum height={\smallwidth}, minimum width={\smallwidth}, tikzit category=string diagram, anchor=center, line width={\stringwidth}, inner sep={\innersep}]
\tikzstyle{medium box}=[rounded corners, {\corners}, font={\boxsize}, fill=neutral, draw=black, shape=rectangle, tikzit category=string diagram, minimum height={\mediumwidth}, minimum width={\smallwidth}, anchor=center, line width={\stringwidth}, tikzit fill=neutral, inner sep={\innersep}]
\tikzstyle{medium box white}=[{\corners}, font={\boxsize}, fill=neutral, draw=black, shape=rectangle, tikzit category=string diagram, minimum height={\mediumwidth}, minimum width={\smallwidth}, anchor=center, line width={\stringwidth}, tikzit fill=comb, inner sep={\innersep}]
\tikzstyle{medium box comb}=[{\corners}, font={\boxsize}, fill=comb, draw=black, shape=rectangle, tikzit category=string diagram, minimum height={\mediumwidth}, minimum width={\smallwidth}, anchor=center, line width={\stringwidth}, tikzit fill=comb, inner sep={\innersep}]
\tikzstyle{medium square box white}=[rounded corners, {\corners}, font={\boxsize}, fill=neutral, draw=black, shape=rectangle, tikzit category=string diagram, minimum height={\mediumwidth}, minimum width={\mediumwidth}, line width={\stringwidth}, tikzit fill=neutral, inner sep={\innersep}]
\tikzstyle{medium square box comb}=[rounded corners, {\corners}, font={\boxsize}, fill=comb, draw=black, shape=rectangle, tikzit category=string diagram, minimum height={\mediumwidth}, minimum width={\mediumwidth}, line width={\stringwidth}, tikzit fill=comb, inner sep={\innersep}]
\tikzstyle{medium square box seq}=[rounded corners, {\corners}, draw=black, font={\boxsize}, fill=seq, shape=rectangle, tikzit category=string diagram, minimum height={\mediumwidth}, minimum width={\mediumwidth}, line width={\stringwidth}, tikzit fill=seq, inner sep={\innersep}]
\tikzstyle{medium square rounded box white}=[rounded corners, font={\boxsize}, fill=neutral, draw=black, shape=rectangle, tikzit category=string diagram, minimum height={\mediumwidth}, minimum width={\mediumwidth}, line width={\stringwidth}, tikzit fill=neutral, inner sep={\innersep}]
\tikzstyle{medium square rounded box comb}=[rounded corners, font={\boxsize}, fill=comb, draw=black, shape=rectangle, tikzit category=string diagram, minimum height={\mediumwidth}, minimum width={\mediumwidth}, line width={\stringwidth}, tikzit fill=comb, inner sep={\innersep}]
\tikzstyle{medium square rounded box seq}=[rounded corners, draw=black, font={\boxsize}, fill=seq, shape=rectangle, tikzit category=string diagram, minimum height={\mediumwidth}, minimum width={\mediumwidth}, line width={\stringwidth}, tikzit fill=seq, inner sep={\innersep}]
\tikzstyle{large box}=[{\corners}, font={\boxsize}, fill=neutral, draw=black, shape=rectangle, tikzit category=string diagram, minimum height=15mm, minimum width=10mm, anchor=center, line width={\stringwidth}, tikzit fill=neutral, inner sep={\innersep}]
\tikzstyle{large square box white}=[{\corners}, font={\boxsize}, fill=neutral, draw=black, shape=rectangle, tikzit category=string diagram, minimum width=15mm, minimum height=15mm, line width={\stringwidth}, tikzit fill=neutral, inner sep={\innersep}]
\tikzstyle{large square box comb}=[{\corners}, font={\boxsize}, fill=comb, draw=black, shape=rectangle, tikzit category=string diagram, minimum width=12mm, minimum height=12mm, line width={\stringwidth}, tikzit fill=comb, inner sep={\innersep}]
\tikzstyle{huge box}=[{\corners}, fill=neutral, draw=black, shape=rectangle, minimum height=40mm, minimum width=15mm, anchor=center, tikzit category=string diagram, line width={\stringwidth}, tikzit fill=neutral, inner sep={\innersep}]
\tikzstyle{output-node}=[fill=neutral, draw=black, shape=circle, anchor=west, inner sep=0.5]
\tikzstyle{delay}=[fill=neutral, draw=black, font={\boxsize}, shape=signal, tikzit category=string diagram, line width={\stringwidth}, minimum height={\tinywidth}, minimum width={\tinywidth}, inner sep=0.5mm, outer sep=0mm, minimum height=1em, minimum width=1em]
\tikzstyle{unit delay}=[fill=neutral, draw=unit, font={\boxsize}, shape=signal, tikzit category=string diagram, line width={\stringwidth}, minimum height={\tinywidth}, minimum width={\tinywidth}, inner sep=0.5mm, outer sep=0mm, minimum height=1em, minimum width=1em]
\tikzstyle{register}=[fill=seq, draw=black, font={\boxsize}, shape=signal, tikzit category=string diagram, line width={\stringwidth}, minimum width=1.5em, align=center, anchor=center, inner xsep=1mm, inner ysep=1mm, outer xsep=0mm]
\tikzstyle{waveform}=[fill=seq, draw=black, font={\boxsize}, shape=signal, signal from=west, signal to=east, tikzit category=string diagram, line width={\stringwidth}, minimum height=2em, minimum width=1.5em, align=center, anchor=center, inner xsep=1mm, inner ysep=-1mm, outer xsep=0mm]
\tikzstyle{bproduct}=[fill=black, draw=black, shape=circle, scale=0.5, tikzit category=string diagram]
\tikzstyle{red product}=[fill=palered, draw=palered, shape=circle, scale=0.5, tikzit category=string diagram]
\tikzstyle{wproduct}=[fill=neutral, draw=black, shape=circle, scale=0.5, line width=0.75, tikzit category=string diagram]
\tikzstyle{gproduct}=[fill=unit, draw=unit, shape=circle, scale=0.5, tikzit category=string diagram]
\tikzstyle{bport}=[minimum size=2.5mm, fill=none, draw=none, shape=circle, tikzit draw={rgb,255: red,154; green,154; blue,154}, anchor=center, tikzit fill=neutral, tikzit category=string diagram]
\tikzstyle{mux}=[fill=neutral, draw=black, shape=trapezium, tikzit category=circuits, rotate=-90, minimum height=1em, line width={\stringwidth}, scale=1.25]
\tikzstyle{fork}=[fill=black, draw=black, shape=circle, tikzit category=circuits, scale=0.25]
\tikzstyle{box}=[fill=neutral, draw=black, shape=rectangle, tikzit category=circuits]
\tikzstyle{dangling}=[fill=none, draw=none, shape=circle, anchor=east, scale=0.01, tikzit category=string diagram, tikzit fill={rgb,255: red,162; green,76; blue,77}]
\tikzstyle{label}=[fill=none, draw=none, shape=circle, align=center, inner sep=0, outer sep=0]
\tikzstyle{small label}=[scale=1.25, fill=none, draw=none, shape=circle, outer sep=0, inner sep=0, anchor=center]
\tikzstyle{wire label left}=[scale=1.25, fill=none, draw=none, shape=rectangle, outer sep=0, inner sep=0, anchor=east]
\tikzstyle{wire label right}=[scale=1.25, fill=none, draw=none, shape=rectangle, outer sep=0, inner sep=0, anchor=west]
\tikzstyle{wire label mid}=[scale=1.25, fill={\bgcolour}, shape=rectangle, outer sep=0, inner sep=0, minimum size=0]
\tikzstyle{tile}=[fill=neutral, draw=black, shape=rectangle, tikzit category=string diagram, {\corners}, minimum height=5mm, minimum width=5mm]
\tikzstyle{commuting label}=[fill=none, draw=none, shape=circle, scale=0.5]
\tikzstyle{and}=[fill=neutral, draw=black, shape=and gate, line width={\stringwidth}, and gate, scale=1.75, tikzit category=circuits]
\tikzstyle{or}=[fill=neutral, draw=black, or gate, scale=1.75, line width={\stringwidth}, tikzit category=circuits]
\tikzstyle{red or}=[fill=neutral, draw=palered, or gate, scale=1.75, line width={\stringwidth}, tikzit category=circuits]
\tikzstyle{not}=[fill=neutral, draw=black, not gate, scale=1.5, line width={\stringwidth}, tikzit category=circuits]
\tikzstyle{nor}=[fill=neutral, draw=black, nor gate, scale=1.75, line width={\stringwidth}, tikzit category=circuits]
\tikzstyle{nand}=[fill=neutral, draw=black, nand gate, scale=1.75, line width={\stringwidth}, tikzit category=circuits]
\tikzstyle{xor}=[fill=neutral, draw=black, xor gate, scale=1.75, line width={\stringwidth}, tikzit category=circuits]
\tikzstyle{xnor}=[fill=neutral, draw=black, xnor gate, scale=1.75, line width={\stringwidth}, tikzit category=circuits]
\tikzstyle{west}=[fill=none, draw=none, shape=circle, anchor=west]
\tikzstyle{short bundler}=[fill=neutral, draw=black, shape=rounded rectangle, minimum width=2.5em, rounded rectangle arc length=180, rotate=90, line width={\stringwidth}]
\tikzstyle{bundler}=[fill=neutral, draw=black, shape=rounded rectangle, minimum width=3em, rounded rectangle arc length=180, rotate=90, line width={\stringwidth}]
\tikzstyle{long bundler}=[fill=neutral, draw=black, shape=rounded rectangle, minimum width=5em, rounded rectangle arc length=180, rotate=90, line width={\stringwidth}]
\tikzstyle{box vertex}=[fill=white, draw=black, shape=circle, line width={\stringwidth}]

\tikzstyle{interface}=[-, fill={rgb,255: red,238; green,238; blue,255}, dashed, draw={rgb,255: red,170; green,170; blue,225}, ultra thick]
\tikzstyle{graph}=[-, fill={rgb,255: red,238; green,238; blue,238}, draw={rgb,255: red,191; green,191; blue,191}, dashed, ultra thick, anchor=center]
\tikzstyle{tentacle}=[-, very thick]
\tikzstyle{red tentacle}=[-, draw=palered, very thick]
\tikzstyle{wire}=[-, tikzit category=string diagram, line width={\stringwidth}]
\tikzstyle{red wire}=[-, draw=palered, tikzit category=string diagram, line width={\stringwidth}]
\tikzstyle{empty}=[-, densely dashed, {\corners}, dash pattern=on 0pt off 1.25pt on 1.25pt, line width={\stringwidth}]
\tikzstyle{transition}=[->]
\tikzstyle{output}=[->, decorate, decoration=zigzag]
\tikzstyle{gate}=[-, fill=neutral]
\tikzstyle{thicc}=[-, line width=1]
\tikzstyle{strikethrough}=[-, decoration={markings, mark=at position 0.5 with {
                \draw [blue, thick, - ] (-0.15,-0.15);
            }}, postaction=decorate]
\tikzstyle{traced}=[-, densely dashed, draw=gray]
\tikzstyle{arrow}=[<-, line width={\stringwidth}]
\tikzstyle{arrow up}=[->, line width={\stringwidth}]
\tikzstyle{dashed arrow}=[<-, dashed, line width={\stringwidth}]
\tikzstyle{unit wire}=[-, fill=none, draw=unit, line width={\stringwidth}]
\tikzstyle{interfacearrow}=[->, very thick, dashed]
\tikzstyle{tile none}=[-, draw=none, {\corners}, fill=none, line width={\stringwidth}]
\tikzstyle{tile white}=[-, {\corners}, fill=neutral, line width={\stringwidth}]
\tikzstyle{tile comb}=[-, {\corners}, fill=comb, line width={\stringwidth}]
\tikzstyle{tile seq}=[-, {\corners}, fill=seq, line width={\stringwidth}]
\tikzstyle{commute}=[->]
\tikzstyle{curved rectangle}=[-, rounded corners]
\tikzstyle{hasse}=[-]
\tikzstyle{wiredash}=[-, line width={\stringwidth}]
\tikzstyle{rewrite}=[-, dashed, line width={\stringwidth}]
\tikzstyle{juxtaposition}=[-, draw={rgb,255: red,128; green,128; blue,128}, densely dashed, line width={\stringwidth}]
\tikzstyle{functor box}=[-, thick, draw={rgb,255: red,83; green,83; blue,83}]
\tikzstyle{boundary box}=[-, draw=none, tikzit draw={rgb,255: red,255; green,0; blue,4}]
\tikzstyle{hom image}=[-, draw={rgb,255: red,191; green,191; blue,191}, fill={rgb,255: red,238; green,238; blue,238}, ultra thick, dashed]
\tikzstyle{tile boundary}=[-, draw={rgb,255: red,128; green,128; blue,128}, dashed]
\tikzstyle{lafont fork}=[-, fill=black]

\begin{document}
\title[Rewriting Modulo Traced Comonoid Structure]{Rewriting Modulo Traced Comonoid Structure\rsuper*}
\titlecomment{{\lsuper*}%
    This is a revised and extended version of~\cite{ghica2023rewriting}.}
\author[D.~R.~Ghica]{Dan R.\ Ghica\lmcsorcid{0000-0002-4003-8893}}
\author[G.~Kaye]{George Kaye\lmcsorcid{0000-0002-0515-4055}}
\keywords{
    symmetric traced monoidal categories,
    string diagrams,
    graph rewriting,
    comonoid structure,
    double pushout rewriting
}

\address{University of Birmingham, UK}
\email{d.r.ghica@bham.ac.uk, georgejkaye@gmail.com}
\begin{abstract}
    In this paper we adapt previous work on rewriting string diagrams using
    hypergraphs to the case where the underlying category has a
    \emph{traced comonoid structure}, in which wires can be forked and
    the outputs of a morphism can be connected to its input.
    Such a structure is particularly interesting because any traced
    Cartesian (dataflow) category has an underlying traced comonoid
    structure.
    We show that certain subclasses of hypergraphs are fully complete
    for traced comonoid categories: that is to say, every term in such a
    category has a unique corresponding hypergraph up to isomorphism, and
    from every hypergraph with the desired properties, a unique term in the
    category can be retrieved up to the axioms of traced comonoid
    categories.
    We also show how the framework of double pushout rewriting (DPO) can be
    adapted for traced comonoid categories by characterising the valid
    pushout complements for rewriting in our setting.
    We conclude by presenting a case study in the form of recent work on
    an equational theory for \emph{sequential circuits}: circuits built from
    primitive logic gates with delay and feedback.
    The graph rewriting framework allows for the definition of an
    \emph{operational semantics} for sequential circuits.
\end{abstract}
\maketitle

\section{Introduction}

String diagrams constitute a useful and elegant conceptual bridge between term
rewriting and graph rewriting.
Since their introduction in the 90s~\cite{joyal1991geometry,joyal1996traced}, their
use has exploded recently both for use in diverse fields such as
cyclic lambda calculi~\cite{hasegawa1997recursion}, fixpoint
operators~\cite{hasegawa2003uniformity}, quantum
protocols~\cite{abramsky2004categorical}, signal flow
diagrams~\cite{bonchi2014categorical,bonchi2015full},
linear algebra~\cite{bonchi2017interacting,zanasi2015interacting,bonchi2019graphical,boisseau2022graphical},
finite state automata~\cite{piedeleu2021string}, dynamical
systems~\cite{baez2015categories,fong2016categorical}, electrical and electronic
circuits~\cite{boisseau2022string,ghica2024fully}, and automatic
differentiation~\cite{alvarez-picallo2023functorial}.
Although not the first use of graphical notation by any
means~\cite{penrose1971applications}, string diagrams are notable because they
are \emph{more} than just a visual aid; they are are a sound and complete
representation of categorical terms in their own
right~\cite{joyal1991geometry,kissinger2014abstract}.
This means one can reason solely in string diagrams without fear that they may
somehow be working with malformed terms.
String diagrams have been adapted to accommodate all sorts of categorical
structure; the survey~\cite{selinger2011survey} is a suitable entry point to
the literature.

While string diagrams have proven to be immensely useful for equational
reasoning with terms in symmetric monoidal categories, they still have their
shortcomings in that they are difficult for computers to manipulate compared
to combinatorial graph-based structures.
This can be remedied by mapping string diagram terms into certain categories of
graphs and performing
\emph{double pushout (DPO) rewriting}~\cite{ehrig1976parallelism}, a
categorical framework in which rewrites are performed by using morphisms
between graphs.
Recent work on this subject has used \emph{framed point graphs}~\cite{kissinger2012pictures,dixon2013opengraphs}
and more recently \emph{hypergraphs}~\cite{bonchi2022string,bonchi2022stringa,bonchi2022stringb}.

Vanilla string diagram terms with no extra structure are fairly basic: all wires
are \emph{progressive} in that they are always flowing in one direction across
the page.
Put simply, the output of a box can only ever connect to the input of
another box.
To model more complicated processes, much work on string diagram rewriting has
considered string diagrams in the presence of extra structure.
A common setting involves the addition of a comonoid structure (`fork') and
a monoid structure (`join') along with special equations to yield what is known
as a \emph{Frobenius structure}; among other properties Frobenius terms form a
\emph{compact closed category}.
One effect of this structure is that wires in Frobenius string diagrams are
`bidirectional': one can connect a wire from the output of a box to another
output of a box.
This can be contrasted with settings equipped with a \emph{trace}, in which
wires can be bent backwards but must be unbent before reaching their endpoints;
much like the ordinary symmetric monoidal setting, outputs can only ever be
connected to inputs.

This paper is concerned with bringing string diagram rewriting techniques
to a particular combination of the structures detailed
above: \emph{traced categories with a comonoid structure} but without a monoid
structure.
This is motivated by a particular class of categories known as
\emph{dataflow categories}~\cite{cazanescu1990new,cazanescu1994feedback,hasegawa1997recursion},
in which the comonoid structure is the Cartesian product; one application of
string diagrammatic reasoning that exhibits a dataflow structure is the
semantics of digital circuits~\cite{ghica2024fully}.
The gap between the kind of semantic models which use an underlying
compact closed structure and those which use a traced monoidal structure is
significant: the former have a \emph{relational} nature with subtle causality
(e.g. quantum or electrical circuits) whereas the latter are \emph{functional}
with clear input-output causality (e.g. digital or logical circuits), so it is
not surprising that the underlying rewrite frameworks should differ.

Reasoning with these categories is technically challenging, as it falls in a gap
between compact closed structures constructible via Frobenius and symmetric
monoidal categories without trace.
For example, it is well known that if the Cartesian product exists in a compact
closed category then it is degenerate and identified with the coproduct.
Even without invoking copying, we will see how trying to perform rewriting in a
traced category with a comonoid structure can also lead to inconsistencies.
This is a firm indication that a bespoke rewriting framework needs to be
constructed to fill this particular situation.

Traced categories were considered by the aforementioned work on framed point
graphs, but this requires rewriting modulo so-called \emph{wire homeomorphisms}.
This style of rewriting is awkward and is increasingly considered as obsolete as
compared to the more recent work on rewriting with hypergraphs; our work seeks
to extend the latter work to apply to the traced comonoid case, by building on
the results on hypergraph string diagram rewriting modulo Frobenius
structure~\cite{bonchi2022string}, symmetric monoidal
structure~\cite{bonchi2022stringa}, and settings equipped
with a monoid structure~\cite{fritz2023free,milosavljevic2023string}.

\subsubsection*{Contributions}

This paper makes two distinct technical contributions.
The first is to show that one subclass of cospans of hypergraphs (`partial
monogamous') is fully complete for traced terms (\autoref{cor:stmc-graph-iso}),
and another class (`partial left-monogamous') is fully complete for traced
comonoid terms (\autoref{thm:comonoid-fully-complete}).
What this means is that every string diagram in a traced setting corresponds to
a unique partially monogamous cospan of hypergraphs up to isomorphism, and
every partially monogamous cospan of hypergraphs corresponds to a unique string
diagram, such that the process of mapping between the two interpretations are
inverses.

The challenge in this step is not so much in proving the correctness of the
construction but in defining precisely what these combinatorial structures
should be.
In particular, the extremal point of tracing the identity: \(
\trace{}{
    \iltikzfig{strings/category/identity}[colour=white]
}
=
\iltikzfig{strings/traced/trace-id}
\), corresponding graphically to a closed loop, provides a litmus test.
The way this is resolved must be robust enough to handle the addition of the
comonoid structure, in which one can `trace a forking wire':
\(
\trace{}{
    \iltikzfig{strings/structure/comonoid/copy}[colour=white]
}
=
\iltikzfig{strings/structure/comonoid/trace-fork}
\).

The second contribution is concerned with the well-definedness of graph
rewriting with partial monogamous and partial left-monogamous cospans of
hypergraphs when using DPO rewriting.
A graph rewrite in DPO rewriting is fully determined by the choice of a
\emph{pushout complement} for a rewrite rule and instance of said rule in a
larger graph; the pushout complement is the context of a rewrite step.
For a given rule and graph, there may be multiple such pushout complements, but
not all of these may represent a valid rewrite in a given string diagram
setting.
When rewriting with Frobenius structure, every pushout complement is
valid~\cite{bonchi2022string} whereas when rewriting with symmetric monoidal
structure exactly one pushout complement is valid~\cite{bonchi2022stringa};
for the traced case some pushout complements are valid and some are not.
Our contribution here is to characterise the valid pushout complements as
`traced boundary complements' (\autoref{def:traced-boundary-complement}) for the
traced setting and as `traced left-boundary complements'
(\autoref{def:traced-left-boundary-complement}) for the traced comonoid setting.
Subsequently we show that DPO rewriting using traced boundary complements is
well-defined for
cospans of partial monogamous cospans of hypergraphs
(\autoref{thm:traced-rewriting}), and DPO rewriting using traced left-boundary
complements is well-defined for cospans of partial left-monogamous cospans of
hypergraphs (\autoref{thm:traced-comonoid-rewriting}).

This is best illustrated with an example in which there is a pushout
complement that is valid in a Frobenius setting because it uses the monoid
structure, but it is not valid neither in a traced, nor even in a traced
comonoid setting.
Imagine we have a rule \(\rrule{
    \iltikzfig{graphs/dpo/non-valid/rule-lhs}
}{
    \iltikzfig{graphs/dpo/non-valid/rule-rhs}
}\) and a term \(
\iltikzfig{graphs/dpo/non-valid/term}
\), and rewrite it as follows.
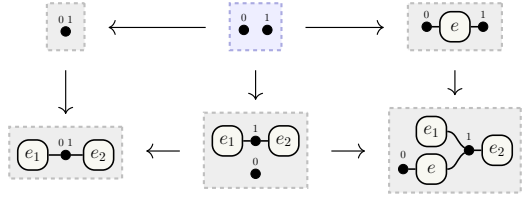
\begin{center}
    
    \begin{tikzcd}[column sep = small, row sep = small]
        \iltikzfig{graphs/dpo/non-valid/l}
        \arrow{d}
        &
        \arrow{l}
        \iltikzfig{graphs/dpo/non-valid/k}
        \arrow{d}
        \arrow{r}
        &
        \iltikzfig{graphs/dpo/non-valid/r}
        \arrow{d}
        \\
        \iltikzfig{graphs/dpo/non-valid/g}
        &
        \arrow{l}
        \iltikzfig{graphs/dpo/non-valid/c}
        \arrow{r}
        &
        \iltikzfig{graphs/dpo/non-valid/h}
    \end{tikzcd}

\end{center}
This corresponds to the term rewrite \(
\iltikzfig{graphs/dpo/non-valid/term}
=
\iltikzfig{graphs/dpo/non-valid/term-rewriting}
=
\iltikzfig{graphs/dpo/non-valid/term-rewritten}
\), which holds in a Frobenius setting but not a setting without a commutative
monoid structure.
On the other hand, the rewriting system for symmetric monoidal
categories~\cite{bonchi2022stringa} is too restrictive as it enforces that any
matching must be mono: this prevents matchings such as \(
\iltikzfig{graphs/dpo/matchings/trace-rule}
\) in \(
\iltikzfig{graphs/dpo/matchings/trace-match}
\).
Here again the challenge is precisely identifying the concept of traced boundary
complement mathematically.
The solution, although not immediately obvious, is not complicated, again
requiring a generalisation from monogamy to partial monogamy and partial
left-monogamy.

Towards the end of this paper, we provide two extended case studies on how
string diagram rewriting modulo traced comonoid structure can be applied.
The first is a discussion on generic rewriting in settings where the comonoid
structure is a Cartesian product; the second is an overview of how graph
rewriting can be applied to the categorical theory of digital circuits
presented in~\cite{ghica2024fully}, resulting in an automated operational
semantics for sequential circuits.

This paper is an extended version of a paper of the same title published
in the proceedings of FSCD 2023~\cite{ghica2023rewriting}.
The results presented are the same, but this version contains a refined
narrative, full proofs for all results, and expanded examples.

\section{Monoidal theories and hypergraphs}

When modelling a system using monoidal categories, its components and
properties are specified using a \emph{monoidal theory}.
A class of SMCs particularly interesting to us is that of
\emph{PROPs} (`categories of \emph{PRO}ducts and
\emph{P}ermutations')~\cite[Chap.\ V.24]{maclane1965categorical},
which have natural numbers as objects and addition as
tensor product on objects.

\begin{defi}[Symmetric monoidal theory]
    A \emph{(single-sorted) symmetric monoidal theory} (SMT) is a tuple \(
    (\generators,\equations)
    \) where \(\generators\) is a set of \emph{generators} in which each
    generator \(\phi \in \generators\) has an associated \emph{arity}
    \(\dom[\phi] \in \nat\) and \emph{coarity} \(\cod[\phi] \in \nat\),
    and \(\equations\) is a set of equations.
    Given a SMT \((\generators,\equations)\), let \(
    \smc{\generators}
    \) be the strict symmetric monoidal category freely generated over
    \(\generators\) and let \(\smc{\generators,\equations}\) be
    \(\smc{\generators}\) quotiented by the equations in \(\equations\).
\end{defi}

\begin{rem}
    One can also define a \emph{multi-sorted} SMT, in which wires can be of
    multiple colours.
    In this paper we will only consider the single-sorted case, but the results
    generalise easily using the results
    of~\cite{bonchi2022string,bonchi2022stringa}.
\end{rem}

While one could reason in \(\smc{\generators}\) using the
one-dimensional categorical term language, it is more intuitive to reason with
\emph{string diagrams}~\cite{joyal1991geometry,selinger2011survey}, which
represent \emph{equivalence classes} of terms up to the axioms of SMCs.
In the language of string diagrams, a generator \(\morph{\phi}{m}{n}\) is drawn
as a box \(
\iltikzfig{strings/category/generator}[box=\phi,colour=white,dom=m,cod=n]
\), the identity \(\id[n]\) as \(
\iltikzfig{strings/category/identity}[colour=white,obj=n]
\), and the symmetry \(\swap{m}{n}\) as \(
\iltikzfig{strings/symmetric/symmetry}[colour=white,obj1=m,obj2=n]
\).
Composite terms will be illustrated as wider boxes \(
\iltikzfig{strings/category/f}[box=f,colour=white,dom=m,cod=n]
\) to distinguish them from generators.
Diagrammatic order composition \(
\iltikzfig{strings/category/f}[box=f,colour=white,dom=m,cod=n]
\seq
\iltikzfig{strings/category/f}[box=g,colour=white,dom=n,cod=p]
\) is defined as horizontal juxtaposition \(
\iltikzfig{strings/category/composition}[box1=f,box2=g,colour=white,dom=m,cod=p]
\) and tensor \(
\iltikzfig{strings/category/f}[box=f,colour=white,dom=m,cod=n]
\tensor
\iltikzfig{strings/category/f}[box=g,colour=white,dom=p,cod=q]
\) as vertical juxtaposition \(
\iltikzfig{strings/monoidal/tensor}[box1=f,box2=g,colour=white,dom1=m,cod1=n,dom2=p,cod2=q]
\).

\begin{figure*}
    \(
    \iltikzfig{strings/category/identity-l-lhs}[box=f,colour=white,dom=m,cod=n]
    =
    \iltikzfig{strings/category/f}[box=f,colour=white,dom=m,cod=n]
    \quad
    \iltikzfig{strings/category/identity-r-lhs}[box=f,colour=white,dom=m,cod=n]
    =
    \iltikzfig{strings/category/f}[box=f,colour=white,dom=m,cod=n]
    \)
    \\[1em]
    \(
    \iltikzfig{strings/category/associativity-lhs}[box1=f,box2=g,box3=h,colour=white,dom=m,cod=p]
    =
    \iltikzfig{strings/category/associativity-rhs}[box1=f,box2=g,box3=h,colour=white,dom=m,cod=p]
    \)
    \\[1em]
    \(
    \iltikzfig{strings/monoidal/unit-l-lhs}[box=f,colour=white,dom=m,cod=n]
    =
    \iltikzfig{strings/category/f}[box=f,colour=white,dom=m,cod=n]
    \quad
    \iltikzfig{strings/monoidal/unit-r-lhs}[box=f,colour=white,dom=m,cod=n]
    =
    \iltikzfig{strings/category/f}[box=f,colour=white,dom=m,cod=n]
    \quad
    \iltikzfig{strings/monoidal/associativity-lhs}[box1=f,box2=g,box3=h,colour=white,dom1=m,cod1=n,dom2=p,cod2=q,dom3=r,cod3=s]
    =
    \iltikzfig{strings/monoidal/associativity-rhs}[box1=f,box2=g,box3=h,colour=white,dom1=m,cod1=n,dom2=p,cod2=q,dom3=r,cod3=s]
    \)
    \\[1em]
    \(
    \iltikzfig{strings/monoidal/identity-tensor-lhs}[colour=white,obj1=m,obj2=n]
    =
    \iltikzfig{strings/category/identity}[colour=white,obj=m+n]
    \quad
    \iltikzfig{strings/monoidal/interchange-lhs}[box1=f,box2=g,box3=h,box4=k,colour=white,dom1=m,dom2=p,cod1=n,cod2=q]
    =
    \iltikzfig{strings/monoidal/interchange-rhs}[box1=f,box2=g,box3=h,box4=k,colour=white,dom1=m,dom2=p,cod1=n,cod2=q]
    \)
    \\[1em]
    \(
    \iltikzfig{strings/symmetric/naturality-lhs}[box1=f,box2=g,colour=white,dom1=m,dom2=p,cod1=n,cod2=q]
    =
    \iltikzfig{strings/symmetric/naturality-rhs}[box1=f,box2=g,colour=white,dom1=m,dom2=p,cod1=n,cod2=q]
    \quad
    \iltikzfig{strings/symmetric/hexagon-lhs}[colour=white,obj1=m,obj2=n,obj3=p]
    =
    \iltikzfig{strings/symmetric/symmetry}[colour=white,obj1=m,obj2=n+p]
    \)
    \\[1em]
    \(
    \iltikzfig{strings/symmetric/unit-l-lhs}[colour=white,unit=0,obj=n]
    =
    \iltikzfig{strings/category/identity}[colour=white,obj=n]
    \quad
    \iltikzfig{strings/symmetric/unit-r-lhs}[colour=white,obj=m,unit=0]
    =
    \iltikzfig{strings/category/identity}[colour=white,obj=m]
    \quad
    \iltikzfig{strings/symmetric/self-inverse-lhs}[colour=white,obj1=m,obj2=n]
    =
    \iltikzfig{strings/symmetric/self-inverse-rhs}[colour=white,obj1=m,obj2=n]
    \)
    \caption{
        The equations of symmetric monoidal categories,
        expressed string diagrammatically in the language of PROPs
    }
    \label{fig:smc-axioms}
\end{figure*}

\begin{rem}
    String diagrams are not restricted to PROPs, and can be used to provide a
    graphical notation for any flavour of symmetric monoidal category.
    In this paper we will stick to the PROP case, motivated by our use of
    monoidal theories containing generators with natural numbers as their domain
    and codomain.
\end{rem}

\begin{rem}
    Although we have drawn a generator
    \(\iltikzfig{strings/category/generator}[box=\phi,colour=white,dom=m,cod=n]\)
    as a box with a single wire as the input and output, this is actually
    syntactic sugar for drawing \(m\) and \(n\) individual wires respectively;
    this avoids cluttering up diagrams with lots of parallel wires.
    A way of turning this syntactic sugar into a formal syntactic construct is
    discussed in~\cite{wilson2023string}.

    On a related note, we may omit the labels from wires when clear from
    context; in the absence of such context, a wire with no label can be taken
    to mean a wire for the object \(1\).
\end{rem}

The power of string diagrams comes from how they `absorb' the equations of
SMCs, as shown in \autoref{fig:smc-axioms}.
This is not merely a convenient notation; string diagrams for symmetric monoidal
categories are a mathematically rigorous language in their own right.

\begin{thmC}[{\cite[Thm. 2.3]{joyal1991geometry}}]
    Given two terms \(f,g \in \smcsigma\), \(f = g\) by axioms of SMCs if and only
    if their string diagrams are isomorphic.
\end{thmC}

String diagrams clearly illustrate the differences between the
\emph{syntactic} category \(\smc{\generators}\) and the \emph{semantic} category
\(\smc{\generators,\equations}\).
In the former, only `structural' equalities of the axioms of SMCs hold: moving
boxes around while retaining connectivity.
In the latter, more equations hold so terms with completely different boxes and
connectivity can be equal.

\begin{exa}\label{ex:frobenius}
    The monoidal theory of
    \emph{special commutative Frobenius algebras} is defined as \(
    (\generators[\frob], \equations[\frob])
    \) where \(
    \generators[\frob] := \{
    \iltikzfig{strings/structure/monoid/merge}[colour=white],
    \iltikzfig{strings/structure/monoid/init}[colour=white],
    \iltikzfig{strings/structure/comonoid/copy}[colour=white],
    \iltikzfig{strings/structure/comonoid/discard}[colour=white]
    \}
    \) and the equations of \(\equations[\frob]\) are listed in
    Figures~\ref{fig:monoid-equations}, \ref{fig:comonoid-equations},
    and~\ref{fig:frobenius-equations}.
    We write \(\frob := \smc{\generators[\frob], \equations[\frob]}\).

    Terms in \(\frob\) are all the ways of composing the generators of
    \(\generators[\frob]\) in sequence and parallel.
    For example, the following are all terms in \(\frob\):
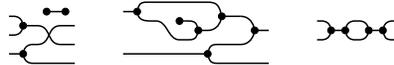
    \[
        \iltikzfig{strings/structure/frobenius/example-1}
        \quad
        \iltikzfig{strings/structure/frobenius/example-2}
        \quad
        \iltikzfig{strings/structure/frobenius/example-3}
    \]
    Using the equations of \(\equations[\frob]\), we can show that the latter
    two terms are equal in \(\frob\):
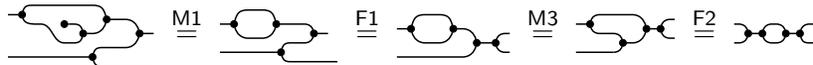
    \begin{gather*}
        \iltikzfig{strings/structure/frobenius/example-2}
        \eqaxioms[\monoidunitleqn]
        \iltikzfig{strings/structure/frobenius/example-equational/step-1}
        \eqaxioms[\frobleqn]
        \iltikzfig{strings/structure/frobenius/example-equational/step-2}
        \eqaxioms[\monoidassoceqn]
        \iltikzfig{strings/structure/frobenius/example-equational/step-3}
        \eqaxioms[\frobreqn]
        \iltikzfig{strings/structure/frobenius/example-3}
    \end{gather*}
    Since the diagrammatic notation takes care of the axioms of SMCs, we only
    need to worry about the equations of the monoidal theory.
\end{exa}

\begin{figure}
    \centering
    \(\equationdisplay{
        \iltikzfig{strings/structure/monoid/unitality-l-lhs}
    }{
        \iltikzfig{strings/structure/monoid/unitality-l-rhs}
    }{
        \monoidunitleqn
    }\)
    \quad
    \(\equationdisplay{
        \iltikzfig{strings/structure/monoid/associativity-lhs}
    }{
        \iltikzfig{strings/structure/monoid/associativity-rhs}
    }{
        \monoidassoceqn
    }\)
    \quad
    \(\equationdisplay{
        \iltikzfig{strings/structure/monoid/commutativity-lhs}
    }{
        \iltikzfig{strings/structure/monoid/commutativity-rhs}
    }{
        \monoidcommeqn
    }\)
    \caption{Equations \(\equations[\cmon]\) of a \emph{commutative monoid}}
    \label{fig:monoid-equations}
\end{figure}

\begin{figure}
    \centering
    \(\equationdisplay{
        \iltikzfig{strings/structure/comonoid/unitality-l-lhs}
    }{
        \iltikzfig{strings/structure/comonoid/unitality-l-rhs}
    }{
        \comonoidunitleqn
    }\)
    \quad
    \(\equationdisplay{
        \iltikzfig{strings/structure/comonoid/associativity-lhs}
    }{
        \iltikzfig{strings/structure/comonoid/associativity-rhs}
    }{
        \comonoidassoceqn
    }\)
    \quad
    \(\equationdisplay{
        \iltikzfig{strings/structure/comonoid/commutativity-lhs}
    }{
        \iltikzfig{strings/structure/comonoid/commutativity-rhs}
    }{
        \comonoidcommeqn
    }\)
    \caption{Equations \(\equations[\ccomon]\) of a \emph{commutative comonoid}}
    \label{fig:comonoid-equations}
\end{figure}

\begin{figure}[t]
    \centering
    \(\equationdisplay{
        \iltikzfig{strings/structure/frobenius/frobenius-l}
    }{
        \iltikzfig{strings/structure/bialgebra/merge-copy-lhs}
    }{
        \frobleqn
    }\)
    \quad
    \(\equationdisplay{
        \iltikzfig{strings/structure/frobenius/frobenius-r}
    }{
        \iltikzfig{strings/structure/bialgebra/merge-copy-lhs}
    }{
        \frobreqn
    }\)
    \quad
    \(\equationdisplay{
        \iltikzfig{strings/structure/frobenius/copy-merge-lhs}
    }{
        \iltikzfig{strings/structure/frobenius/copy-merge-rhs}
    }{
        \frobspeceqn
    }\)
    \caption{
        Equations \(\equations[\frob]\) of a
        \emph{special commutative Frobenius algebra}, in addition to those in
        Figures~\ref{fig:monoid-equations} and~\ref{fig:comonoid-equations}
    }
    \label{fig:frobenius-equations}
\end{figure}

\subsection{Categories of hypergraphs}

Reasoning equationally using string diagrams is certainly attractive
as a pen-and-paper method, but for larger systems it quickly becomes intractible
to do this by hand.
Instead, it is desirable to perform equational reasoning \emph{computationally}.
Unfortunately, string diagrams as topological objects are not particularly
suited for this purpose; it is more suitable to use a combinatorial
representation.
Fortunately, this has been well studied recently, first with
\emph{string graphs}~\cite{dixon2013opengraphs,kissinger2012pictures}
and later with \emph{hypergraphs}~\cite{bonchi2022string,bonchi2022stringa,bonchi2022stringb},
a generalisation of regular graphs in which edges can be the source or target of
an arbitrary number of nodes.
In this paper we are concerned with the latter.

Hypergraphs are formally defined as objects in a functor category.

\begin{defi}[Hypergraph]
    Let \(\mathbf{X}\) be the category containing objects \((k, l)\) for
    \(k, l \in \nat\) and one additional object \(\star\).
    For each \((k, l)\) there are \(k + l\) morphisms \((k, l) \to \star\).
    Let \(\hyp\) be the functor category \([\mathbf{X},\set]\).
\end{defi}

An object in \(\hyp\) maps \(\star\) to a set of nodes, and each pair
\((k,l)\) to a set of hyperedges with \(k\) sources and \(l\) targets.
Given a hypergraph \(F \in \hyp\), we write \(\vertices{F}\) for its set of
vertices and \(\edges{F}{k}{l}\) for the set of edges with \(k\) sources and
\(l\) targets.
Subsequently each functor induces functions
\(\morph{\sources{i}}{\edges{F}{k}{l}}{\vertices{F}}\) for \(i < k\) and
\(\morph{\targets{j}}{\edges{F}{k}{l}}{\vertices{F}}\) for \(j < l\).

\begin{exa}\label{ex:hypergraph}
    We define the hypergraph \(F\) as follows:
    \begin{gather*}
        \vertices{F} := \{v_0,v_1,v_2,v_3,v_4,v_5\}
        \quad
        \edges{F}{2}{1} := \{e_0,e_2\}
        \quad
        \edges{F}{1}{2} := \{e_1\}
        \\
        \sources{0}(e_0) := v_0
        \quad
        \sources{1}(e_0) := v_2
        \quad
        \sources{0}(e_1) := v_3
        \quad
        \sources{0}(e_2) := v_3
        \quad
        \sources{1}(e_2) := v_2
        \\
        \targets{0}(e_0) := v_3
        \quad
        \targets{0}(e_1) := v_5
        \quad
        \targets{1}(e_1) := v_4
        \quad
        \targets{0}(e_2) := v_4
    \end{gather*}

    Much like regular graphs, it is easier to comprehend hypergraphs
    graphically.
    nodes are drawn as black dots and hyperedges as boxes.
    Tentacles from edges to nodes identify the (ordered) sources and targets.
    The hypergraph \(F\) is drawn as follows:
    \[
        \iltikzfig{graphs/blank-example}
    \]
    Note that the nodes themselves do not have any notion of ordering or
    directionality.
\end{exa}

The graphical notation for hypergraphs already looks very similar to that of
string diagrams.
However, the boxes are lacking labels for generators in a signature; to remedy
this we shall translate signatures themselves into hypergraphs.

\begin{defi}[Hypergraph signature~{\cite[Sec.\ 3.1]{bonchi2022string}}]
    For a given monoidal signature \(\signature\), its corresponding
    \emph{hypergraph signature} \(\hypsignature{\Sigma}\) is the hypergraph with
    a single node \(v\) and edges \(
    e_\phi \in \edges{\hypsignature{\Sigma}}{\dom[\phi]}{\cod[\phi]}
    \) for each \(\phi \in \Sigma\).
    For a hyperedge \(e_\phi\), \(i < \dom[\phi]\) and \(j < \cod[\phi]\), \(
    \sources{i}(e_\phi) = \targets{j}(e_\phi) = v
    \).
\end{defi}

A labelling is then a morphism from a hypergraph to a signature.
A morphism of hypergraphs \(\morph{f}{F}{G} \in \hyp\) consists of functions
\(\vertices{f}\) and \(\edges{f}{k}{l}\) for each \(k,l \in \nat\) preserving
sources and targets in the obvious way.

\begin{exa}\label{ex:labelled-hypergraph}
    Let \(\Sigma = \{\morph{\phi}{2}{1}, \morph{\psi}{1}{2}\}\) be a monoidal
    signature.
    The corresponding monoidal signature \(\hypsignature{\Sigma}\) is
    \begin{gather*}
        \vertices{\hypsignature{\Sigma}} := \{v_0\}
        \quad
        \edges{\hypsignature{\Sigma}}{2}{1} := \{e_\phi\}
        \quad
        \edges{\hypsignature{\Sigma}}{1}{2} := \{e_\psi\}
        \\
        \sources{0}(e_\phi) := v_0
        \quad
        \sources{1}(e_\phi) := v_0
        \quad
        \sources{0}(e_\psi) := v_0
        \quad
        \targets{0}(e_\phi) := v_0
        \quad
        \targets{0}(e_\psi) := v_0
        \quad
        \targets{1}(e_\psi) := v_0
    \end{gather*}
    and is drawn as follows:
    \[
        \iltikzfig{graphs/signature-example}
    \]
    Recall the hypergraph \(F\) from \autoref{ex:hypergraph}; one labelling
    \(\morph{\Gamma}{F}{\hypsignature{\Sigma}}\) could be defined as
    \begin{gather*}
        \vertices{\Gamma}(-) := v_0
        \quad
        \edges{\Gamma}{2}{1}(e_0) := e_\phi
        \quad
        \edges{\Gamma}{1}{2}(e_1) := e_\psi
        \quad
        \edges{\Gamma}{2}{1}(e_2) := e_\phi
    \end{gather*}
    We simply call the morphism \(\Gamma\) a \emph{labelled hypergraph} and
    draw it in the same manner as a regular hypergraph but with labelled edges.
    \[
        \iltikzfig{graphs/blank-example-labelled}
    \]
    Note that if there are multiple generators with the same arity and coarity
    in a signature, there may well be multiple valid labellings of a hypergraph.
\end{exa}

A category of labelled hypergraphs is defined using another piece of
categorical machinery.

\begin{defi}[Slice category~{\cite[pg.\ 36]{lawvere1963functorial}}]
    For a category \(\mathbf{C}\) and an object \(C \in \mathbf{C}\), the
    \emph{slice category} \(\mathbf{C} / C\) has objects the
    morphisms of \(\mathbf{C}\) with target \(C\), and has a morphism \(
    (\morph{f}{X}{C}) \to (\morph{f^\prime}{X^\prime}{C})
    \) if there is a morphism \(\morph{g}{X}{X^\prime} \in \mathbf{C}\) such
    that \(f^\prime \circ g = f\).
\end{defi}

\begin{defi}[Labelled hypergraph~{\cite[Sec.\ 3.1]{bonchi2022string}}]
    For a monoidal signature \(\Sigma\), let the category \(\hypsigma\) be
    defined as the slice category \(\hyp / \hypsignature{\Sigma}\).
\end{defi}

There is another difference between hypergraphs and string diagrams.
While (labelled) hypergraphs may have dangling nodes, they do not have
\emph{interfaces} specifying the order of inputs and outputs.
These can be provided using \emph{cospans}.

\begin{defiC}[{\cite[Def.\ 2.10]{bonchi2022string}}]\label{def:cospans}
    For a category \(\mathbf{C}\) with finite colimits, a \emph{cospan} from
    \(X \to Y\) is a pair of arrows \(X \to A \leftarrow Y\).
    A \emph{cospan morphism} \(
    (\cospan{X}[f]{A}[g]{Y}) \to (\cospan{X}[h]{B}[k]{Y})
    \) is a morphism \(\morph{\alpha}{A}{B} \in \mathbf{C}\) such that the
    following diagram commutes:
    \begin{center}
        
    \begin{tikzcd}[row sep = small, column sep = small]
        &
        A
        \arrow{dd}{\alpha}
        &
        \\
        X
        \arrow[bend left = 30]{ur}{f}
        \arrow[bend right = 30, swap]{dr}{h}
        &
        &
        Y
        \arrow[bend right = 30, swap]{ul}{g}
        \arrow[bend left = 30]{dl}{k}
        \\
        &
        B
        &
    \end{tikzcd}

    \end{center}

    Two cospans \(\cospan{X}{A}{Y}\) and \(\cospan{X}{B}{Y}\) are
    \emph{isomorphic} if there exists a morphism of cospans as above where
    \(\alpha\) is an isomorphism.
    Composition is by pushout:

    
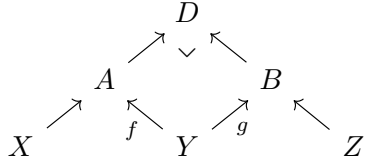
\begin{center}
        
\begin{tikzcd}[row sep = small, column sep = small]
    &
    &
    D
    \arrow[pos=0.125, draw=none]{dd}[description]{\hspace{-0.25em}\rotatebox{225}{\scalebox{1.5}{\(\urcorner\)}}}
    &
    &
    \\
    &
    A
    \arrow{ur}
    &
    &
    B
    \arrow{ul}
    &
    \\
    X
    \arrow{ur}{f}
    &
    &
    Y
    \arrow{ul}{g}
    \arrow[swap]{ur}{h}
    &
    &
    Z
    \arrow{ul}{k}
\end{tikzcd}

    \end{center}

    The identity is \(X \xrightarrow{\id[X]} X \xleftarrow{\id[X]} X\).
    The category of cospans over \(\mathbf{C}\), denoted \(\csp{\mathbf{C}}\),
    has as objects the objects of \(\mathbf{C}\) and as morphisms the
    isomorphism classes of cospans.
    This category has monoidal product given by the coproduct in \(\mathbf{C}\)
    and has monoidal unit given by the initial object \(0 \in \mathbf{C}\).
\end{defiC}

The interfaces of a hypergraph can be specified as cospans by having the `legs'
of the cospan pick nodes in the graph at the apex.

\begin{defi}[Discrete hypergraph]
    A hypergraph is called \emph{discrete} if it has no edges.
\end{defi}

A discrete hypergraph \(F\) with \(|\vertices{F}| = n\) is written as
\(n\) when clear from context.
Morphisms from discrete hypergraphs to another hypergraph pick out the nodes
in the `inputs' and `outputs' of the latter.
To establish a category of cospans in which the legs are discrete hypergraphs,
and moreover to establish an \emph{ordering} on these discrete hypergraphs, we
move from using plain cospans of hypergraphs to cospans in which the legs are in
the image of some functor.

\begin{thmC}[{\cite[Thm.\ 3.6]{bonchi2022string}}]\label{thm:cospan-functor-category}
    Let \(\mathbb{X}\) be a PROP whose monoidal product is a coproduct,
    \(\mathbf{C}\) a category with finite colimits, and \(
    \morph{F}{\mathbb{X}}{\mathbf{C}}
    \) a coproduct-preserving functor.
    Then there exists a PROP \(\csp[F]{\mathbf{C}}\) whose arrows \(m \to n\)
    are isomorphism classes of \(\mathbf{C}\) cospans \(\cospan{Fm}{C}{Fn}\).
\end{thmC}
\begin{proof}
    Composition is by pushout.
    For cospans \(\cospan{Fm}{C}{Fn}\) and \(\cospan{Fp}{C}{Fq}\), their
    coproduct is given by \(\cospan{Fm + Fp}{C + D}{Fn + Fq}\):
    \(F(m + p) \cong Fm + Fp\) and \(F(n + q) \cong Fn + Fq\) because \(F\)
    preserves coproducts.
    Symmetries in \(\mathbb{X}\) are determined by the universal property of
    the coproduct; they are inherited by \(\csp[F]{\mathbf{C}}\) because \(F\)
    preserves coproducts.
\end{proof}

\begin{thmC}[{\cite[Thm.\ 3.8]{bonchi2022string}}]\label{thm:cospan-homomorphism}
    Let \(\mathbb{X}\) be a PROP whose monoidal product is a coproduct,
    \(\mathbf{C}\) a category with finite colimits, and
    \(\morph{F}{\mathbb{X}}{\mathbf{C}}\) a colimit-preserving functor.
    Then there is a homomorphism of PROPs \(
    \morph{\tilde{F}}{\csp{\mathbb{X}}}{\csp[F]{\mathbf{C}}}
    \) that sends \(\cospan{m}[f]{X}[g]{n}\) to \(\cospan{Fm}[Ff]{FX}[Fg]{Fn}\).
    If \(F\) is full and faithful, then \(\tilde{F}\) is faithful.
\end{thmC}
\begin{proof}
    Since \(F\) preserves finite colimits, it preserves composition (pushout)
    and monoidal product (coproduct); symmetries are clearly preserved.
    To show that \(\tilde{F}\) is faithful when \(F\) is full and faithful,
    suppose that \(
    \tilde{F}(\cospan{m}[f]{X}[g]{n})
    =
    \tilde{F}(\cospan{m}[f^\prime]{X}[g^\prime]{n})
    \).
    This gives us the following commutative diagram in \(\mathbf{C}\):
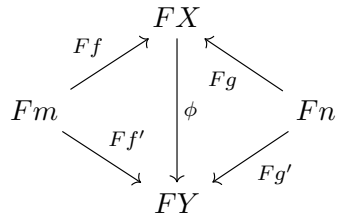
    \begin{center}
        \begin{tikzcd}
            & FX \arrow{dd}{\phi} & \\
            Fm \arrow{ur}{Ff} \arrow{dr}{Ff^\prime} & &
            Fn \arrow{ul}{Fg} \arrow{dl}{Fg^\prime} \\
            & FY &
        \end{tikzcd}
    \end{center}
    where \(\phi\) is an isomorphism as objects in \(\csp[F]{\mathbf{C}}\) are
    isomorphism classes of cospans.
    As \(F\) is full, there exists \(\morph{\psi}{X}{Y}\) such that
    \(F\psi = \phi\).
    As \(F\) is faithful, \(\psi\) is an isomorphism; this means
    \(\cospan{m}[f]{X}[g]{n}\) and \(\cospan{m}[f^\prime]{X}[g^\prime]{n}\) are
    equal in \(\csp{\mathbb{X}}\), so \(\tilde{F}\) is faithful.
\end{proof}

In our setting, this functor will be from the category of finite sets.

\begin{defi}
    Let \(\finset\) be the PROP with morphisms \(m \to n\) the functions
    between finite sets \([m] \to [n]\).
\end{defi}

\begin{defiC}[{\cite[Sec.\ 3.2]{bonchi2022string}}]\label{def:discrete-functor}
    Let \(\morph{D}{\finset}{\hypsigma}\) be a faithful, coproduct-preserving
    functor that sends each object \([n] \in \finset\) to the discrete hypergraph
    \(n \in \hypsigma\) and each morphism to the induced homomorphism of
    discrete hypergraphs.
\end{defiC}

\(D\) (for `discrete') maps a finite set
\([n] := \{0, 1, \cdots, n-1\}\) to the discrete hypergraph with \(n\) nodes.
The functor induces an underlying function between sets
\([n] \to \vertices{n}\); this establishes an ordering on the interfaces.

By combining the above theorems and definitions we obtain the category
\(\cspdhyp\) with objects \emph{discrete cospans of hypergraphs}.
Since the legs of each cospan are discrete hypergraphs with \(n\) nodes, this
is another PROP.

\begin{exa}
    Recall the labelled hypergraph \(F\) from \autoref{ex:labelled-hypergraph}.
    We assign interfaces to it as the cospan \(\cospan{3}[f]{F}[g]{3}\), where
    \begin{gather*}
        f(D0) = v_0 \quad f(D1) = v_1 \quad
        g(D0) = v_4 \quad g(D1) = v_1 \quad g(D2) = v_4
    \end{gather*}
    Interfaces of the hypergraph \(F\) are drawn to its left and right, with
    numbers illustrating the action of the cospan maps.
    For clarity we number the outputs after the inputs.
    \[
        \iltikzfig{graphs/example-interfaces}
    \]
    Composition in \(\cspdhyp\) is by pushout; effectively the nodes in the
    output of the first cospan are `glued together' with the inputs of the
    second.
    \begin{gather*}
        \iltikzfig{graphs/example-interfaces}
        \seq
        \iltikzfig{graphs/example-2-interfaces}
        \\[1em]=
        \iltikzfig{graphs/example-composition}
    \end{gather*}
    Tensor in \(\cspdhyp\) is by direct product; putting cospans on top of each
    other.
    \begin{gather*}
        \iltikzfig{graphs/example-interfaces}
        \tensor
        \iltikzfig{graphs/example-2-interfaces}
        \\[1em]=
        \iltikzfig{graphs/example-tensor}
    \end{gather*}
\end{exa}

\subsection{Frobenius terms as hypergraphs}

The main result of \cite{bonchi2022string} is that the category of
`Frobenius terms' \(\smcsigma + \frob\) is in correspondence with \(\cspdhyp\):
every morphism in the former (modulo the Frobenius equations in
Figures~\ref{fig:monoid-equations}, \ref{fig:comonoid-equations},
and~\ref{fig:frobenius-equations}) corresponds to exactly one isomorphism class
of cospans in \(\cspdhyp\), and vice versa.
To show this, we make use of another type of category known as a
\emph{hypergraph category}

\begin{defi}[Hypergraph category~{\cite[Def.\ 2.12]{fong2019hypergraph}}]
    A \emph{hypergraph category} is a symmetric monoidal category in which
    each object \(X\) has a special commutative Frobenius structure in the sense
    of \autoref{ex:frobenius} satisfying the equations in
    \autoref{fig:hypergraph-category}.
\end{defi}

The prototypical hypergraph category over a signature is the
freely generated symmetric monoidal category \(\smcsigma\) augmented with the
Frobenius generators and equations in \(\frob\).

\begin{cor}
    \(\smcsigma + \frob\) is a hypergraph category.
\end{cor}

The term `hypergraph category' should not be confused with the
category of hypergraphs \(\hyp\), which is not itself a hypergraph category.
However, the category of \emph{cospans} of hypergraphs \emph{is} such a category.

\begin{figure}
    \centering
    \iltikzfig{strings/structure/hypergraph/monoid-resp-lhs}[obj1=x,obj2=y,tensor=+]
    \(=\)
    \iltikzfig{strings/structure/hypergraph/monoid-resp-rhs}[obj1=x,obj2=y]
    \quad
    \iltikzfig{strings/structure/hypergraph/comonoid-resp-lhs}[obj1=x,obj2=y,tensor=+]
    \(=\)
    \iltikzfig{strings/structure/hypergraph/comonoid-resp-rhs}[obj1=x,obj2=y]
    \\[1.5em]
    \iltikzfig{strings/structure/hypergraph/unit-resp-lhs}[obj1=x,obj2=y,tensor=+]
    \(=\)
    \iltikzfig{strings/structure/hypergraph/unit-resp-rhs}[obj1=x,obj2=y]
    \quad
    \iltikzfig{strings/structure/hypergraph/counit-resp-lhs}[obj1=x,obj2=y,tensor=+]
    \(=\)
    \iltikzfig{strings/structure/hypergraph/counit-resp-rhs}[obj1=x,obj2=y]
    \caption{
        Equations \(\equations[\mathbf{Hyp}]\) of a
        \emph{hypergraph category}, as well as those in
        Figures~\ref{fig:monoid-equations}, \ref{fig:comonoid-equations},
        and~\ref{fig:frobenius-equations}.
    }
    \label{fig:hypergraph-category}
\end{figure}

\begin{prop}\label{prop:frobenius-map}
    \(\cspdhyp\) is a hypergraph category.
\end{prop}
\begin{proof}
    A Frobenius structure can be defined on \(\cspdhyp\) for each
    \(n \in \nat\)
    as follows:
    \begin{gather*}
        \iltikzfig{strings/structure/monoid/merge}[colour=white,obj=n]
        :=
        \cospan{n + n}{n}{n}
        \quad
        \iltikzfig{strings/structure/monoid/init}[colour=white,obj=n]
        :=
        \cospan{0}{n}{n}
        \\
        \iltikzfig{strings/structure/comonoid/copy}[colour=white,obj=n]
        :=
        \cospan{n}{n}{n+n}
        \quad
        \iltikzfig{strings/structure/comonoid/discard}[colour=white,obj=n]
        :=
        \cospan{n}{n}{0}
        \qedhere
    \end{gather*}
\end{proof}

We make use of the components of the above definition in order to define a
PROP morphism from \(\smcsigma + \frob\) into \(\cspdhyp\).
Since \(\smc{\Sigma}\) is freely generated, these PROP morphisms can be defined
solely on generators.

\begin{defiC}[{\cite[Sec.\ 3.2]{bonchi2022string}}]\label{def:hyp-morphisms}
    Let \(\morph{\termtohyp{\Sigma}}{\smc{\Sigma}}{\cspdhyp}\) be a PROP
    morphism defined on generators as \(
    \termtohyp[\iltikzfig{strings/category/generator}[box=\phi,colour=white,dom=m,cod=n]]{\Sigma}
    :=
    \cospan{m}{
        \iltikzfig{graphs/terms/generator}
    }{n}.
    \)
\end{defiC}

\begin{defi}
    Let \(\morph{\frobtohyp{\Sigma}}{\frob}{\cspdhyp}\) be a PROP morphism
    defined as in \autoref{prop:frobenius-map}.
\end{defi}

\begin{defi}
    Let the PROP morphism \(
    \morph{\termandfrobtohypsigma}{\smc{\Sigma} + \frob}{\cspdhyp}
    \) be the copairing of \(\termtohyp{\Sigma}\) and \(\frobtohyp{\Sigma}\).
\end{defi}

The PROP morphism \(\termandfrobtohypsigma\) maps Frobenius terms to cospans
of hypergraphs.
To show that \(\cspdhyp\) is suitable for reasoning about Frobenius terms, this
mapping must be full and faithful; i.e.\ it must be a one-to-one mapping between
Frobenius terms and isomorphism classes of cospans.

\begin{thmC}[{\cite[Thm.\ 4.1]{bonchi2022string}}]\label{thm:isomorphism-smcfrob-cospans}
    There is an isomorphism of PROPs \(\smcsigma + \frob \cong \cspdhyp\).
\end{thmC}
\begin{proof}[Proof (Sketch)]
    Since \(\smcsigma + \frob\) is a coproduct in the category of PROPs, this
    can be shown by proving that \(\cspdhyp\) satisfies the universal property
    of the coproduct: given a coloured PROP \(\mathbb{A}\) and PROP morphisms
    \(\smcsigma \to \mathbb{A}\) and \(\frob \to \mathbb{A}\), there exists
    a unique morphism \(\morph{u}{\cspdhyp}{\mathbb{A}}\) as below:
    \begin{center}
        
\begin{tikzcd}[column sep=large]
    \smcsigmac
    \arrow{r}{\termtohypsigma}
    \arrow[swap]{rd}{f}
    &
    \cspdhyp
    \arrow[dotted]{d}{u}
    &
    \frob
    \arrow[swap]{l}{\frobtohypsigma}
    \arrow{ld}{g}
    \\
    &
    \mathbb{A}
    &
\end{tikzcd}

    \end{center}
    All the PROP morphisms involved are identity-on-objects, so all that is
    required to show the existence of \(u\) is to show that any morphism in
    \(\cspdchyp\) can be expressed as a composition of components either in the
    image of \(\termtohypsigma\) or \(\frobtohypsigma\).
\end{proof}

Rather than giving the formal construction required for the above proof, it is
more instructive to provide a concrete example.

\begin{exa}\label{ex:frobenius-composite}
    Consider the following term and its cospan interpretation:
    \begin{center}
        \iltikzfig{graphs/frobenius/correspondence/original-term}
        \qquad
        \begin{tikzcd}
            \iltikzfig{graphs/frobenius/correspondence/inputs}
            \arrow{r}
            &
            \iltikzfig{graphs/frobenius/correspondence/example}
            &
            \arrow{l}
            \iltikzfig{graphs/frobenius/correspondence/outputs}
        \end{tikzcd}
    \end{center}
    This cospan can be assembled into the form shown in
    \autoref{fig:frobenius-composite}; by following the vertex maps, one can
    verify that this is indeed isomorphic to the original cospan.
    The outermost components correspond to terms in \(\frob\) and the
    innermost to a term in \(\smcsigma\).
    \[
        \iltikzfig{graphs/frobenius/correspondence/term}
    \]
    This term is equal to the original term by the Frobenius equations.
\end{exa}

\begin{figure}
    \centering
    \begin{tikzcd}[column sep=tiny]
        \scalebox{0.7}{\iltikzfig{graphs/frobenius/correspondence/l-inputs}}
        \arrow{r}
        &
        \scalebox{0.7}{\iltikzfig{graphs/frobenius/correspondence/l}}
        &
        \arrow{l}
        \scalebox{0.7}{\iltikzfig{graphs/frobenius/correspondence/l-outputs}}
    \end{tikzcd}
    \hspace{-1em}
    \(\seq\)
    \hspace{-1em}
    \begin{tikzcd}[column sep=tiny, row sep=0.1cm]
        \scalebox{0.7}{\iltikzfig{graphs/frobenius/correspondence/e-inputs}}
        \arrow{r}
        &
        \scalebox{0.7}{\iltikzfig{graphs/frobenius/correspondence/n}}
        &
        \arrow{l}
        \scalebox{0.7}{\iltikzfig{graphs/frobenius/correspondence/e-outputs}}
        \\
        \scalebox{0.7}{\iltikzfig{graphs/frobenius/correspondence/e-inputs}}
        \arrow{r}
        &
        \scalebox{0.7}{\iltikzfig{graphs/frobenius/correspondence/e}}
        &
        \arrow{l}
        \scalebox{0.7}{\iltikzfig{graphs/frobenius/correspondence/e-outputs}}
    \end{tikzcd}
    \hspace{-1em}
    \(\seq\)
    \hspace{-1em}
    \begin{tikzcd}[column sep=tiny]
        \scalebox{0.7}{\iltikzfig{graphs/frobenius/correspondence/r-inputs}}
        \arrow{r}
        &
        \scalebox{0.7}{\iltikzfig{graphs/frobenius/correspondence/r}}
        &
        \arrow{l}
        \scalebox{0.7}{\iltikzfig{graphs/frobenius/correspondence/r-outputs}}
    \end{tikzcd}
    \caption{
        The cospan of \autoref{ex:frobenius-composite} in the form of
        \autoref{thm:isomorphism-smcfrob-cospans}
    }
    \label{fig:frobenius-composite}
\end{figure}

This result means that any two terms in \(\smcsigma + \frob\) which are equal by
the Frobenius equations can be mapped to isomorphic cospans of hypergraphs.

\begin{exa}
    Recall the following terms in \(\frob\) from \autoref{ex:frobenius}, which
    we showed were equal by the Frobenius equations.
    \[
        \iltikzfig{strings/structure/frobenius/example-2}
        =
        \iltikzfig{strings/structure/frobenius/example-3}
    \]
    By the isomorphism of \autoref{thm:isomorphism-smcfrob-cospans}, these two
    terms map to the same cospan of hypergraphs:
    \[
        \iltikzfig{graphs/example-frobenius-collapse}
    \]
    All of the Frobenius structure collapses into one vertex, much like when we
    considered the correspondence between Frobenius terms and finite sets.
\end{exa}

\section{Hypergraphs for traced categories}

In the previous section we summarised the results of \cite{bonchi2022string}
showing that every cospan of hypergraphs in \(\cspdhyp\) corresponds to
a single Frobenius term.
In this paper we are concerned with \emph{traced} terms, a more restrictive
setting: we lack the ability to fork and join wires, and enforce that outputs of
boxes may only be connected to inputs of boxes.

\begin{defi}[Symmetric traced monoidal category {\cite[Sec.\ 2]{joyal1996traced}}]
    A \emph{symmetric traced monoidal category} (STMC) is a symmetric monoidal
    category \(\mathbf{C}\) equipped with a family of functions \(
    \morph{
        \trace{X}[A][B]{-}
    }{
        \mathbf{C}(X \tensor A, X \tensor B)
    }{
        \mathbf{C}(A, B)
    }
    \) for any objects \(A,B\) and \(X\) satisfying the axioms of STMCs:
    \begin{align*}
        \id[X] \tensor f \seq g \seq \id[X] \tensor h
                                                                & = f \seq \trace{X}[B][C]{g}
                                                                &
                                                                & \text{(Tightening)}
        \\
        \trace{X}[A][B]{f \seq g \tensor \id[B]}
                                                                & =            \trace{Y}[A][B]{g \tensor \id[A] \seq f}
                                                                &
                                                                & \text{(Sliding)}
        \\
        \trace{X}[A][B]{\trace{Y}[X \tensor A][X \tensor B]{f}} & =            \trace{X \tensor Y}[A][B]{f}
                                                                &
                                                                & \text{(Vanishing)}
        \\
        \trace{X}[A][B]{f} \tensor g
                                                                & = \trace{X}[A \tensor C][B \tensor D]{f \tensor g}
                                                                &
                                                                & \text{(Superposing)}
        \\
        \trace{Y}[A][B]{g \tensor \id[A] \seq f}
                                                                & = \id[X]\
                                                                &
                                                                & \text{(Yanking)}
    \end{align*}
\end{defi}

\begin{figure}
    \centering
    \iltikzfig{strings/traced/naturality-lhs}[dom=m,cod=q,trace=x, box3=h]
    \(=\)
    \iltikzfig{strings/traced/naturality-rhs}[dom=m,cod=q,trace=x, box3=h]
    \\[1.5em]
    \iltikzfig{strings/traced/sliding-lhs}[dom=m,cod=n,trace=x]
    \(=\)
    \iltikzfig{strings/traced/sliding-rhs}[dom=m,cod=n,trace=y]
    \\[1.5em]
    \iltikzfig{strings/traced/vanishing-lhs}[dom=m,cod=n,trace1=x,trace2=y]
    \(=\)
    \iltikzfig{strings/traced/vanishing-prop-rhs}[dom=m,cod=n,trace1=x,trace2=y]
    \\[1.5em]
    \iltikzfig{strings/traced/superposing-lhs}[dom1=m,dom2=p,cod1=n,cod2=q,trace=x]
    \(=\)
    \iltikzfig{strings/traced/superposing-rhs}[dom1=m,dom2=p,cod1=n,cod2=q,trace=x]
    \qquad
    \iltikzfig{strings/traced/yanking-lhs}[obj=x]
    \(=\)
    \iltikzfig{strings/traced/yanking-rhs}[obj=x]
    \caption{
        Equations that hold in any \emph{symmetric traced monoidal category},
        expressed string diagrammatically in the language of PROPs
    }
    \label{fig:stmc-axioms}
\end{figure}

The trace is represented string diagrammatically by joining output
wires to input wires:
\[
    \trace{x}[m][n]{\iltikzfig{strings/traced/trace-lhs}[box=f,colour=white,trace=x,dom=m,cod=n]}
    \stackrel{\text{def}}{=}
    \iltikzfig{strings/traced/trace-rhs}[box=f,colour=white,dom=m,cod=n]
\]
When drawn in this manner the equations of STMCs can be elegantly expressed in
the language of PROPs as shown in \autoref{fig:stmc-axioms}.
Much like regular STMCs, this notation is sound and complete.

\begin{defi}
    For a monoidal signature \(\generators\), let \(\stmcsigma\) be the STMC
    freely generated over \(\generators\).
    For a set of equations \(\equations\), let \(\stmc{\generators,\equations}\)
    be defined as \(\stmcsigma / \equations\).
\end{defi}

\begin{thmC}[{\cite[Cor.\ 6.14]{kissinger2014abstract}}]
    Given two terms \(f,g \in \stmcsigma\), \(f = g\) by axioms of STMCs if and
    only if their string diagrams are isomorphic.
\end{thmC}

\subsection{Compact closed categories}

For \(\cspdhyp\) to be suitable for reasoning about traced terms, it must
necessarily have a trace.
Fortunately, the links between the trace and Frobenius structure is very
well-known, and can be expressed using the notion of another type of category.

\begin{defi}[Compact closed category~{\cite[Sec.\ 1]{kelly1980coherence}}]
    A \emph{compact closed category} (CCC) is a symmetric monoidal category in
    which every object \(X\) has a \emph{dual} \(\dual{X}\) equipped with
    morphisms called the \emph{unit} \(\morph{\eta}{I}{\dual{X} \tensor X}\) and
    the \emph{counit} \(\morph{\varepsilon}{X \tensor \dual{X}}{I}\), satisfying
    the `snake equations':
    \(\varepsilon \tensor \id[\dual{X}] \seq \id[X] \tensor \eta = \id[X]\)
    and
    \(\id[X] \tensor \eta \seq \varepsilon \tensor \id[\dual{X}] = \id[\dual{X}]\)
\end{defi}

\begin{figure}
    \centering
    \iltikzfig{strings/compact-closed/snake-l}[obj=x]
    \(=\)
    \iltikzfig{strings/compact-closed/snake-c}[obj=x]
    \quad
    \iltikzfig{strings/compact-closed/snake-r}[obj=x]
    \(=\)
    \iltikzfig{strings/compact-closed/snake-c}[obj=x]
    \caption{
        Equations that hold in any \emph{compact closed category},
        expressed string diagrammatically in the language of PROPs
    }
    \label{fig:ccc-axioms}
\end{figure}

When thinking string diagrammatically, a dual to an object can be thought of as
a wire `running backwards'.
However, in this paper we are only concerned with categories in which objects
are \emph{self-dual}: any object \(X\) is equal to \(\dual{X}\).
For the self-dual case, the unit and counit are drawn string diagrammatically as
\(\iltikzfig{strings/compact-closed/cup-self-dual}[colour=white,obj=n]
\) (`cup') and \(
\iltikzfig{strings/compact-closed/cap-self-dual}[colour=white,obj=n]
\) (`cap') respectively.
When drawing the snake equations string using the cup and cap as in
\autoref{fig:ccc-axioms}, the reasoning behind their name becomes apparent.

This graphical notation is suggestive of links between traced and compact closed
categories, and this is no accident.

\begin{propC}[{\cite[Prop.\ 3.1]{joyal1996traced}}]
    \label{prop:canonical-trace}
    Any compact closed category has a trace.
\end{propC}

This trace is called the \emph{canonical trace}; for the self-dual case, it
is constructed as follows:
\[\trace{x}[m][n]{\iltikzfig{strings/category/f-2-2}[box=f,colour=white,dom1=x,dom2=m,cod1=x,cod2=n]}
    \coloneqq
    \iltikzfig{strings/compact-closed/canonical-trace-self-dual}
\]

So to use the Frobenius results in the traced setting, we just need to show that
\(\cspdhyp\) is compact closed; this is also a well-known result.

\begin{lemC}[{\cite[Prop.\ 2.8]{rosebrugh2005generic}}]
    Every hypergraph category is self-dual compact closed.
\end{lemC}
\begin{proof}
    The cup is defined as \(
    \iltikzfig{strings/compact-closed/cup-self-dual}[colour=white,obj=n]
    :=
    \iltikzfig{strings/structure/frobenius/cup}[obj=n]
    \) and the cap as \(
    \iltikzfig{strings/compact-closed/cap-self-dual}[colour=white,obj=n]
    :=
    \iltikzfig{strings/structure/frobenius/cap}[obj=n]
    \).
    The snake equations follow by applying the Frobenius equation and unitality:
    \begin{gather*}
        \iltikzfig{strings/structure/frobenius/snake-1-0}
        =
        \iltikzfig{strings/structure/frobenius/snake-1-1}
        =
        \iltikzfig{strings/structure/frobenius/snake-1-2}
        \qquad
        \iltikzfig{strings/structure/frobenius/snake-2-0}
        =
        \iltikzfig{strings/structure/frobenius/snake-2-1}
        =
        \iltikzfig{strings/structure/frobenius/snake-2-2}
        \qedhere
    \end{gather*}
\end{proof}

\begin{cor}
    \(\cspdhyp\) is compact closed.
\end{cor}

\begin{cor}
    \(\cspdhyp\) has a trace.
\end{cor}

A STMC freely generated over a signature faithfully embeds into a CCC
generated over the same signature, mapping the trace in the former to the
canonical trace in the latter.

\begin{lem}\label{lem:stmc-subcat-hypc}
    \(\stmcsigma\) is a subcategory of \(\smcsigma + \frob\).
\end{lem}
\begin{proof}
    Since \(\hypsigma\) is compact closed, it has a (canonical) trace.
    For \(\stmcsigma\) to be a subcategory of \(\hypcsigma\), every morphism
    of the former must also be a morphism on the latter.
    Since the two categories are freely generated (with the trace constructed
    through the Frobenius generators in the latter), all that remains is to
    check that every morphism in \(\stmcsigma\) is a unique morphism in
    \(\hypcsigma\), i.e.\ the equations of \(\frob\) do not merge any together.
    This is trivial since the equations do not apply to the construction of the
    canonical trace.
\end{proof}

\begin{defi}
    Let \(\morph{\tracedtosymandfrobsigma}{\stmcsigma}{\smcsigma + \frob}\) be
    the inclusion functor of \autoref{lem:stmc-subcat-hypc}.
\end{defi}

\begin{cor}
    \(\tracedtosymandfrobsigma\) is faithful.
\end{cor}
\begin{proof}
    \(\tracedtosymandfrobsigma\) is an inclusion functor.
\end{proof}

Crucially, this mapping is not \emph{full}: there are terms in a CCC that are
not terms in a STMC, such as \(\iltikzfig{strings/traced/invalid}\).
In order to identify a category of cospans of hypergraphs for traced terms we
must restrict the cospans of hypergraphs in \(\cspdhyp\) in some way.

\subsection{Properties of hypergraph cospans}

In~\cite{bonchi2022stringa}, it is shown that terms in a (non-traced)
symmetric monoidal category are interpreted via a faithful functor into a
sub-PROP of \(\cspdhyp\) named \(\macspdhyp\), containing only the `monogamous
acyclic` cospans.
We will examine these conditions of monogamicity and acyclicity, and show how
these can be adapted to the traced case.

Considered informally, \emph{monogamy} means that every node has exactly one
`in' and `out' connection, be it to an edge or an interface.
This corresponds to the fact that wires in symmetric monoidal categories cannot
arbitrarily fork or join.

\begin{defi}
    For a hypergraph \(F \in \hyp\), the \emph{degree} of a node
    \(v \in \vertices{F}\) is a tuple \((i,o)\) where \(i\) is the number of
    hyperedges with with \(v\) as a target, and \(o\) is the number of
    hyperedges with \(v\) as a source.
\end{defi}

\begin{defi}[Monogamy~{\cite[Def.\ 13]{bonchi2022stringa}}]
    For a cospan of hypergraphs \(\cospan{m}[f]{F}[g]{n}\) in
    \(\cspdhyp\), let \(\mathsf{in}(F)\) be the image of \(f\) and let
    \(\mathsf{out}(F)\) be the image of \(g\).
    The cospan \(\cospan{m}[f]{F}[g]{n}\) is \emph{monogamous} if
    \(f\) and \(g\) are mono and, for all nodes \(v\), the degree of \(v\) is:
    \begin{center}
        \begin{tabular}{rlcrl}
            \((0,0)\)
             &
            if \(v \in \mathsf{in}(F) \wedge v \in \mathsf{out}(F)\)
             &
            \quad
             &
            \((0,1)\)
             &
            if \(v \in \mathsf{in}(F)\)
            \\
            \((1,0)\)
             &
            if \(v \in \mathsf{out}(F)\)
             &
            \quad
             &
            \((1,1)\)
             &
            otherwise
        \end{tabular}
    \end{center}
\end{defi}

The second condition on cospans in \(\macspdhyp\) is that all hypergraphs are
\emph{acyclic}, as in symmetric monoidal terms all wires must flow from left to
right.

\begin{defi}[Predecessor~{\cite[Def.\ 18]{bonchi2022stringa}}]
    A hyperedge \(e\) is a \emph{predecessor} of another hyperedge \(e^\prime\)
    if there exists a node \(v\) such that \(v\) is a target of \(e\) and a
    source of \(e^\prime\).
\end{defi}

\begin{defi}[Path~{\cite[Def.\ 19]{bonchi2022stringa}}]
    A \emph{path} between two hyperedges \(e\) and \(e^\prime\) is a sequence of
    hyperedges \(e_0, \dots, e_{n-1}\) such that \(e = e_0\),
    \(e^\prime = e_{n-1}\), and for each \(i < n-1\), \(e_i\) is a predecessor
    of \(e_{i+1}\).
    A subgraph \(H\) is the \emph{start} or \emph{end} of a path if it contains
    a node in the sources of \(e\) or the targets of \(e^\prime\) respectively.
\end{defi}

Since nodes are single-element subgraphs, one can also talk about paths from
nodes.

\begin{defi}[Acyclicity~{\cite[Def.\ 20]{bonchi2022stringa}}]
    A hypergraph \(F\) is acyclic if there is no path from a node to itself.
    A cospan \(\cospan{m}{F}{n}\) is acyclic if \(F\) is.
\end{defi}

Cospans of hypergraphs with these properties form a sub-PROP of \(\cspdhyp\).

\begin{lemC}[{\cite[Lems.\ 15-17]{bonchi2022stringa}}]\label{lem:monogamicity-preserved}
    The following statements hold:
    \begin{enumerate}
        \item identities and symmetries are monogamous;
        \item monogamicity is preserved by composition; and
        \item monogamicity is preserved by tensor.
    \end{enumerate}
\end{lemC}
\begin{proof}
    For (1), identities and symmetries are monogamous by definition, as they
    are constructed from discrete hypergraphs, in which every node is an input
    and an output.
    For (2), as pushouts along monos in \(\hypsigma\) are monos, the legs of
    the composition of monogamous cospans must also be mono; moreover as each
    node in the outputs of the first cospan is merged with exactly one node in
    the inputs of the second, only nodes of degree at most \((1,1)\) can be
    created.
    For (3), the degree of each node is the same as in the original two cospans,
    and the coproduct of monos is mono.
\end{proof}

\begin{cor}
    There is a sub-PROP of \(\cspdhyp\), named \(\macspdhyp\), containing only
    the monogamous acyclic cospans of hypergraphs.
\end{cor}

These two conditions are enough to establish a correspondence between
monogamous acyclic cospans of hypergraphs and symmetric monoidal terms.

\begin{thmC}[{\cite[Cor.\ 26]{bonchi2022stringa}}]
    There is an isomorphism \(\smcsigma \cong \macspdhyp\).
\end{thmC}

We want to show a similar result for \emph{traced} terms.
Clearly, to model trace the acyclicity condition must be dropped.
For the most part, monogamy also applies to the traced case: wires cannot
arbitrarily fork and join.
There is one nuance: the trace of the identity.
This is depicted as a closed loop \(
\trace{1}{\iltikzfig{strings/category/identity}[colour=white]}
=
\iltikzfig{strings/traced/trace-id}
\), and one might think that it can be discarded, i.e. \(
\iltikzfig{strings/traced/trace-id}
=
\iltikzfig{strings/monoidal/empty}
\).
This is \emph{not} always the case, such as in the category of finite
dimensional vector spaces~\cite[Sec. 6.1]{hasegawa1997recursion}.

These closed loops must be represented in the hypergraph framework:
there is a natural representation as a lone node disconnected
from either interface.
In fact, this is exactly how the canonical trace applied to an identity is
interpreted in \(\cspdhyp\).
This means we need a weaker version of monogamy, which we dub
\emph{partial monogamy}.

\begin{defi}[Partial monogamy]
    For a cospan \(\cospan{m}[f]{F}[g]{n} \in \cspdhyp\), let \(\mathsf{in}(F)\)
    be the image of \(f\) and let \(\mathsf{out}(F)\) be the image of \(g\).
    The cospan \(\cospan{m}[f]{F}[g]{n}\) is \emph{partial monogamous} if
    \(f\) and \(g\) are mono and, for all nodes
    \(v \in \vertices{F}\), the degree of \(v\) is
    \begin{center}
        \begin{tabular}{rlcrl}
            \((0,0)\)
             &
            if \(v \in \mathsf{in}(F) \wedge v \in \mathsf{out}(F)\)
             &
            \quad
             &
            \((0,1)\)
             &
            if \(v \in \mathsf{in}(F)\)
            \\
            \((1,0)\)
             &
            if \(v \in \mathsf{out}(F)\)
             &
            \quad
             &
            \((0,0)\)
            or \((1,1)\)
             &
            otherwise
        \end{tabular}
    \end{center}
\end{defi}

Intuitively, partial monogamy means that each node has either exactly one `in'
and one `out' connection to an edge or to an interface, or none at all.
\begin{figure}
    \centering
    \begin{gather*}
        \underbrace{
            \iltikzfig{graphs/monogamy/yes-0}
            \iltikzfig{graphs/monogamy/yes-1}
        }_{\text{partial monogamous}}
        \\[1em]
        \underbrace{
            \iltikzfig{graphs/monogamy/no-0}
            \iltikzfig{graphs/monogamy/no-1}
        }_{\text{not partial monogamous}}
    \end{gather*}
    \caption{Examples of cospans that are and are not partial monogamous}
    \label{fig:partial-monogamous-examples}
\end{figure}

\begin{exa}
    Examples of cospans that are and are not partial monogamous are shown
    in \autoref{fig:partial-monogamous-examples}.
\end{exa}

In order to establish a correspondence between cospans of partial monogamous
hypergraphs and traced terms, the former need to be assembled into a sub-PROP of
\(\cspdhyp\).

\begin{lem}
    Identities and symmetries are partial monogamous.
\end{lem}
\begin{proof}
    By \autoref{lem:monogamicity-preserved}, identities and symmetries are
    monogamous, so they must also be partial monogamous.
\end{proof}

\begin{lem}\label{lem:partial-monogamous-composition}
    Given partial monogamous cospans \(\cospan{m}{F}{n}\)
    and \(\cospan{n}{G}{p}\), the composition \(
    (\cospan{m}{F}{n})
    \seq
    (\cospan{n}{G}{p})
    \) is partial monogamous.
\end{lem}
\begin{proof}
    As with the proof for preservation of regular monogamicity by composition,
    the legs of a composed cospan must be mono as pushouts along monos are
    themselves mono.

    To verify the degrees of nodes in the composition respect partial monogamy,
    recall that the only changes in nodes in the original two
    hypergraphs \(F\) and \(G\) when compared to the corresponding nodes in the
    composite is that \(i\)-th node in the outputs of \(F\) is merged with the
    \(i\)-th node in the inputs of \(G\), and their degrees summed.
    The only freedoms permitted by partial monogamy over regular monogamy are:
    \begin{itemize}
        \item there can be cycles, but by definition of a path any node in a
              cycle must have degree \((1, 1)\), so cannot be in the image of
              the interfaces; and
        \item nodes with degree \((0, 0)\) can occur outside of the image of the
              interfaces, so these will be unaffected by composition as well.
    \end{itemize}

    So the nodes than can be altered by composition are precisely those
    permitted by regular monogamy, i.e.\ nodes with degree \((0, 0)\) in the
    inputs of \(F\), and nodes with degree \((1, 0)\); as we already know that
    composition preserves the monogamicity degree conditions for these nodes,
    then it must also preserve the partial monogamy degree conditions.
\end{proof}

\begin{lem}\label{lem:partial-monogamous-tensor}
    Given partial monogamous cospans \(\cospan{m}{F}{n}\)
    and \(\cospan{p}{G}{q}\), the tensor \(
    (\cospan{m}{F}{n})
    \tensor
    (\cospan{n}{G}{p})
    \) is partial monogamous.
\end{lem}
\begin{proof}
    As with composition, tensor preserves monogamicity by
    \autoref{lem:monogamicity-preserved}, and as it does not affect the degree of
    nodes then it preserves partial monogamy as well.
\end{proof}

As partial monogamicity is preserved by both forms of composition, the
partial monogamous cospans themselves form a PROP.

\begin{defi}
    Let \(\pmcspdhyp\) be the sub-PROP of \(\cspdhyp\) containing only the
    partial monogamous cospans of hypergraphs.
\end{defi}

To show that \(\pmcspdhyp\) is a \emph{traced} PROP, we must show that
\(\pmcspdhyp\) has a trace.
Although \(\cspdhyp\) already has a trace, we must make sure that this does not
degenerate for cospans of partial monogamous hypergraphs.

\begin{thm}\label{thm:partial-monogamous-trace}
    The canonical trace is a trace on \(\pmcspdhyp\).
\end{thm}
\begin{proof}
    Consider a partial monogamous cospan \(
    \cospan{x + m}[f + h]{F}[g + k]{x + n}
    \); we must show that its trace \(
    \cospan{m}[h]{F^\prime}[k]{n}
    \) is partial monogamous.
    For each node \(a \in x\), \(f(a)\) and \(g(a)\) are merged together in the
    traced graph, summing their degrees.
    If a node is in the image of \(h\) or \(k\), this is also the case in the
    traced cospan.
    We consider the various cases:
    \begin{itemize}
        \item if \(f(a) = g(a)\), then this node must have degree \((0, 0)\);
              the traced node will still have degree \((0, 0)\) and will no
              longer be in the interface;
        \item if \(f(a)\) is also in the image of \(n \to F\) and \(g(a)\) is
              also in the image of \(m \to F\), then both \(f(a)\) and
              \(g(a)\) have degree \((0, 0)\); the traced node will still
              have degree \((0, 0)\) and be in both interfaces of the traced
              cospan;
        \item if \(f(a)\) is in the image of \(n \to F\), then \(f(a)\)
              has \((0, 0)\) and \(g(a)\) has degree
              \((1,0)\), so the traced node has degree \((1, 0)\) and
              is in the image of \(n \to F^\prime\);
        \item if \(g(a)\) is in the image of \(m \to F\), then the above
              applies in reverse; and
        \item if neither node is in the image of \(m \to F\) and \(n \to F\),
              then the traced node will have degree \((1,1)\) and be in the
              image of no interface.
    \end{itemize}
    In all these cases, partial monogamy is preserved.
\end{proof}

Crucially, while we leave \(\pmcspdhyp\) in order to construct the trace using
the cup and cap, the resulting cospan \emph{is} in \(\pmcspdhyp\).

\subsection{From terms to graphs}

Partial monogamous hypergraphs are the domain in which we will interpret terms
in symmetric monoidal theories equipped with a trace.
A \emph{(traced) PROP morphism} is a strict (traced) symmetric monoidal functor
between PROPs.
For \(\pmcspdhyp\) to be suitable for reasoning with a traced category
\(\stmc{\Sigma}\) for some given signature, there must be a
\emph{full and faithful} PROP morphism \(\stmc{\Sigma} \to \pmcspdhyp\).

We exploit the interplay between compact closed and traced categories in
order to reuse the existing PROP morphisms from~\cite{bonchi2022string} for the
traced case.

To map from traced terms to cospans of hypergraphs, we translate them into
\(\smcsigma + \frob\) using the inclusion functor \(\tracedtosymandfrobsigma\),
then use the previously defined PROP morphism \(\termandfrobtohypsigma\) to
map into \(\cspdhyp\).
A diagram showing the interaction of these PROP morphisms can be seen in
\autoref{fig:roadmap}.

We need to show that \(\termandfrobtohypsigma \circ \tracedtosymandfrobsigma\) is
full and faithful when restricted to \(\pmcspdhyp\).
Showing that it is faithful is simple: \(\tracedtosymandfrobsigma\) is faithful
as it is an inclusion, and the components of \(\termandfrobtohypsigma\) are
known to be faithful by the following result.

\begin{propC}[{\cite[Cor.\ 4.3]{bonchi2022string}}]
    \(\termtohypsigma\) is faithful.
\end{propC}
\begin{proof}[Proof (Sketch)]
    This follows by combining the so-called 3-for-2
    property~\cite[Thm.\ 3.3]{macdonald2009amalgamations} with the fact that
    \(\cspdhyp \cong \smcsigma + \frob\) is a pushout in the category of
    small SMCs.
\end{proof}

\begin{prop}
    \(\frobtohypsigma\) is faithful.
\end{prop}
\begin{proof}
    By definition of \(\frobtohypsigma\).
\end{proof}

The more interesting proof is that of the fullness of
\(\termandfrobtohypsigma \circ \tracedtosymandfrobsigma\).
Following the strategy of \cite{bonchi2022string}, we will
assemble an arbitrary partial monogamous cospan of hypergraphs into a form
constructed of components mapped from \(\stmcsigma\).

\begin{thmC}[{\cite[Thm.\ 25]{bonchi2022stringa}}]\label{thm:monog-acyclic-full}
    A cospan \(\cospan{m}{F}{n}\) is in the image of
    \(\termtohypsigma\) if and only if \(\cospan{m}{F}{n}\)
    is monogamous acyclic.
\end{thmC}
\begin{proof}[Proof (Sketch)]
    The only if direction follows by induction on morphisms on
    \(\smcsigma\).
    The if direction is by induction on the hyperedges in \(F\); if
    there are no hyperedges then \(\cospan{m}{F}{n}\) is in the image of a
    morphism containing only identities and symmetries, and for the inductive
    step each edge corresponds to the image of a generator in some larger term.
\end{proof}

\begin{thm}\label{thm:termtohyp-image}
    A cospan \(\cospan{m}{F}{n}\) is in the image of \(
    \termandfrobtohypsigma \circ \tracedtosymandfrobsigma\) if
    and only if it is partial monogamous.
\end{thm}
\begin{proof}
    For the \(\onlyifdir\) direction, the generators of \(\stmcsigma\) are
    mapped to monogamous cospans by \(
    \termandfrobtohypsigma \circ \tracedtosymandfrobsigma
    \), and partial monogamy is preserved by composition
    (\autoref{lem:partial-monogamous-composition}),
    tensor (\autoref{lem:partial-monogamous-tensor}),
    and trace
    (\autoref{thm:partial-monogamous-trace}).

    For the \(\ifdir\) direction, we show that any partial monogamous cospan \(
    \cospan{m}[f]{F}[g]{n}
    \) is in the image of \(
    \termandfrobtohypsigma \circ \tracedtosymandfrobsigma
    \) by constructing an isomorphic trace of cospans, in which
    each component under the trace is in the image of \(\termtohypsigma\).
    The components of the new cospan are as follows:
    \begin{itemize}
        \item let \(L\) be the discrete hypergraph containing nodes with
              degree
              \((0,0)\) that are not in the image of \(f\) or \(g\);
        \item let \(S\) and \(T\) be the discrete hypergraphs containing
              the source and target nodes of hyperedges in \(F\)
              respectively, with the ordering determined by some order
              \(e_1,e_2,\cdots,e_n\) on the edges in \(F\).
        \item let \(E\) be the hypergraph containing the nodes of
              \(S + T\) and the hyperedges of \(F\), in which the sources of
              edges of \(E\) are the nodes in \(S\) that are the sources of
              the edge in \(F\), and the targets of an edge are the nodes
              in \(T\) that are the targets of the edge in \(F\); and
        \item let \(V\) be the discrete hypergraph containing all the
              nodes of \(F\); and
    \end{itemize}
    These parts can be composed to form the following composite:
    \begin{gather*}
        \cospan{L + T + m}[\id + \id + f]{L + V}[\id + \id + g]{L + S + n}
        \,\seq\,
        \cospan{L + S + n}[\id]{L + E + n}[\id]{L + T + n}
    \end{gather*}
    We take the trace of \(L + T\) over this composite to obtain a
    cospan isomorphic to the original.
    The components of the composite under the trace are all monogamous acyclic
    so are in the image of \(\termtohypsigma\) by
    \autoref{thm:monog-acyclic-full}; this means there is a term
    \(f \in \smcsigma\) such that \(\termtohypsigma[f]\) is isomorphic to the
    original composite.
    The trace of \(f\) is in \(\stmcsigma\), so the trace of the composite is in
    the image of
    \(\termandfrobtohypsigma \circ \tracedtosymandfrobsigma\).
\end{proof}

The large composite cospan may appear rather inpenetrable at first glance;
essentially, it is constructed by stacking up the edges in the cospan
\(\cospan{\tilde{m}}{\tilde{E}}{\tilde{n}}\), and joining up the targets to
the appropriate sources by tracing them around and shuffling them to the correct
source.
The graph \(L\) contains any identity loops.

\begin{exa}\label{ex:trace-composite}
    Consider the following term and its cospan interpretation:
    \begin{center}
        \iltikzfig{graphs/trace/correspondence/original-term}
        \quad
        \begin{tikzcd}[column sep=tiny]
            \iltikzfig{graphs/trace/correspondence/inputs}
            \arrow{r}
            &
            \iltikzfig{graphs/trace/correspondence/example}
            &
            \arrow{l}
            \iltikzfig{graphs/trace/correspondence/outputs}
        \end{tikzcd}
    \end{center}
    We assemble the latter into the composite cospan of
    \autoref{thm:termtohyp-image} as shown in \autoref{fig:trace-composite}.
    Both of the components under the trace correspond to terms in \(\smcsigma\),
    so applying the trace to this produces a term in \(\stmcsigma\):
    \begin{center}
        \iltikzfig{graphs/trace/correspondence/term}
    \end{center}
    This term is equal to the original by string diagrammatic deformations.
\end{exa}

\begin{figure}
    \centering
    \(\trace{4}{
        \begin{tikzcd}[column sep=tiny, ampersand replacement=\&]
            \scalebox{0.7}{\(\iltikzfig{graphs/trace/correspondence/s-inputs}\)}
            \arrow{r}
            \&
            \scalebox{0.7}{\(\iltikzfig{graphs/trace/correspondence/s}\)}
            \&
            \arrow{l}
            \scalebox{0.7}{\(\iltikzfig{graphs/trace/correspondence/s-outputs}\)}
        \end{tikzcd}
        \seq
        \begin{tikzcd}[row sep=tiny,column sep=tiny, ampersand replacement=\&]
            \scalebox{0.7}{\(\iltikzfig{graphs/trace/correspondence/l-inputs}\)}
            \arrow{r}
            \&
            \scalebox{0.7}{\(\iltikzfig{graphs/trace/correspondence/l}\)}
            \&
            \arrow{l}
            \scalebox{0.7}{\(\iltikzfig{graphs/trace/correspondence/l-outputs}\)}
            \\
            \scalebox{0.7}{\(\iltikzfig{graphs/trace/correspondence/e-inputs}\)}
            \arrow{r}
            \&
            \scalebox{0.7}{\(\iltikzfig{graphs/trace/correspondence/e}\)}
            \&
            \arrow{l}
            \scalebox{0.7}{\(\iltikzfig{graphs/trace/correspondence/e-outputs}\)}
            \\
            \scalebox{0.7}{\(\iltikzfig{graphs/trace/correspondence/l-inputs}\)}
            \arrow{r}
            \&
            \scalebox{0.7}{\(\iltikzfig{graphs/trace/correspondence/l}\)}
            \&
            \arrow{l}
            \scalebox{0.7}{\(\iltikzfig{graphs/trace/correspondence/l-outputs}\)}
        \end{tikzcd}
    }\)
    \caption{
        The cospan of \autoref{ex:trace-composite} in the form of
        \autoref{thm:termtohyp-image}
    }
    \label{fig:trace-composite}
\end{figure}

This brings us to the first major result of this paper: that terms in a freely
generated symmetric traced monoidal category are in a one-to-one correspondence
with isomorphism classes of partial monogamous cospans of hypergraphs.

\begin{cor}\label{cor:stmc-graph-iso}
    \(\stmc{\Sigma} \cong \pmcspdhyp\).
\end{cor}

This means we can now freely translate between traced string
diagram terms and partial monogamous cospans of hypergraphs.
Crucially, terms that are equal by the axioms of STMCs are mapped to isomorphic
cospans; to verify if two traced terms \(
\iltikzfig{strings/category/f}[box=f]
\) and \(
\iltikzfig{strings/category/f}[box=g]
\) are equal in \(\stmcsigma\), we can check if their hypergraph interpretations
are isomorphic.

\begin{exa}
    The partial monogamous cospans from
    \autoref{fig:partial-monogamous-examples} are shown below with their
    corresponding terms in \(\stmcsigma\).
    \[
        \iltikzfig{graphs/partial-monogamy/yes-0}
        \quad
        \raisebox{-0.5em}{\iltikzfig{graphs/terms/term-0}}
        \qquad
        \iltikzfig{graphs/partial-monogamy/yes-1}
        \quad
        \iltikzfig{graphs/terms/term-1}
    \]
    Note that the position of the `closed loop' could be moved anywhere in the
    term, but the hypergraph interpretation would be unchanged.
\end{exa}

This result is also important for rewriting traced terms modulo some equational
theory, in which the graph interpretations themselves will have to be rewritten.
\autoref{cor:stmc-graph-iso} will be instrumental in showing that any such
rewriting system is sound and complete, that is to say a graph rewrite \(
\termandfrobtohypsigma[\tracedtosymandfrobsigma[
        \iltikzfig{strings/category/f}[box=f]
    ]]
\Rightarrow
\termandfrobtohypsigma[\tracedtosymandfrobsigma[
        \iltikzfig{strings/category/f}[box=g]
    ]]
\) is valid if and only if \(
\iltikzfig{strings/category/f}[box=f]
=
\iltikzfig{strings/category/f}[box=g]
\) under the equational theory.

\section{Hypergraphs for traced commutative comonoid categories}

\begin{figure}
    \centering
    \begin{tikzcd}[column sep=0.6em]
    \smcsigma
    \arrow[hookrightarrow]{r}
    \arrow{dd}{\termtohypsigma}
    &
    \stmcsigma
    \arrow[hookrightarrow]{r}
    \arrow[swap]{ddd}{\termandfrobtohypsigma \circ \tracedtosymandfrobsigma}
    &
    \stmcsigma + \ccomon
    \arrow[hookrightarrow]{r}
    \arrow[bend right=60, swap, start anchor={[xshift=-0.75em]}, end anchor={[xshift=-0.75em]}]{dd}{\termandfrobtohypsigma \circ \tracedtosymandfrobsigma}
    &
    \smcsigma + \frob
    \arrow[bend left=75, swap, start anchor={[yshift=0.6em]}, end anchor={[yshift=-1em]}]{ddd}{\termandfrobtohypsigma}
    \\
    &
    &
    \ccomon
    \arrow[hookrightarrow]{u}
    \arrow{d}{\frobtohypsigma \circ \comonoidtofrob}
    \arrow{r}{\comonoidtofrob}
    &
    \frob
    \arrow{dd}{\frobtohypsigma}
    \arrow[hookrightarrow]{u}
    \\
    \macspdhyp
    \arrow[bend right, hookrightarrow, end anchor={[yshift=0.6em]}]{dr}
    &
    &
    \plmcspdhyp
    \arrow[bend right, hookrightarrow, end anchor={[yshift=0.6em]}]{dr}
    &
    &
    \\
    &
    \pmcspdhyp
    \arrow[bend right, hookrightarrow, start anchor={[yshift=0.5em]}]{ur}
    &
    &
    \cspdhyp
\end{tikzcd}
    \caption{Interactions between categories of terms and hypergraphs}
    \label{fig:roadmap}
\end{figure}

The trace is but one kind of structure we are interested in adding to
symmetric monoidal terms.
We also want to be able to \emph{fork} and \emph{eliminate} wires by adding a
\emph{(commutative) comonoid structure}; categories equipped
with such a structure are also known as \emph{gs-monoidal}
(\emph{garbage-sharing}) categories~\cite{fritz2023free}.

\begin{defi}
    Let \((\generators[\ccomon], \equations[\ccomon])\) be the symmetric
    monoidal theory of \emph{commutative comonoids}, with \(\Sigma_\ccomon := \{
    \iltikzfig{strings/structure/comonoid/copy}[colour=white],
    \iltikzfig{strings/structure/comonoid/discard}[colour=white]
    \}\) and \(\mce_\ccomon\) defined as in \autoref{fig:comonoid-equations}.
    We write \(\ccomon := \smc{\generators[\ccomon], \equations[\ccomon]}\).
\end{defi}

From now on, we write `comonoid' to mean `commutative comonoid'.
There has already been work using hypergraphs for PROPs with a (co)monoid
structure~\cite{fritz2023free,milosavljevic2023string} but these consider
\emph{acyclic} hypergraphs: we must ensure that removing the acyclicity
condition does not lead to any degeneracies.

\begin{defi}[Partial left-monogamy]
    For a cospan \(\cospan{m}[f]{H}[g]{n}\), we say it is
    \emph{partial left-monogamous} if \(f\) is mono and, for all nodes
    \(v \in H_\star\), the degree of \(v\) is
    \[
        \begin{cases}
            (0, m)                    & \text{if } v \text{ is in the image of } f \\
            (0, m) \text{ or } (1, m) & \text{otherwise}
        \end{cases}
    \] for some \(m \in \nat\).
\end{defi}

Partial left-monogamy is a weakening of partial monogamy that allows nodes
to have multiple `out' connections, which represent the use of the comonoid
structure to fork wires.

\begin{figure}
    \centering
    \begin{gather*}
        \underbrace{
            \iltikzfig{graphs/partial-monogamy/yes-0}
            \iltikzfig{graphs/partial-monogamy/yes-1}
        }_{\text{partial left-monogamous}}
        \\[1em]
        \underbrace{
            \iltikzfig{graphs/partial-monogamy/no-0}
            \iltikzfig{graphs/partial-monogamy/no-1}
        }_{\text{not partial left-monogamous}}
    \end{gather*}
    \caption{Examples of cospans that are and are not partial left-monogamous}
    \label{fig:partial-left-monogamous-examples}
\end{figure}

\begin{exa}
    Examples of cospans that are and are not partial left-monogamous are shown
    in \autoref{fig:partial-left-monogamous-examples}.
\end{exa}

\begin{rem}
    As with the nodes not in the interfaces with degree \((0, 0)\) in the
    vanilla traced case, the nodes not in the interface with degree
    \((0, m)\) allow for the interpretation of terms such as \(
    \trace{}{\iltikzfig{strings/structure/comonoid/copy}[colour=white]}
    \).
\end{rem}

\begin{lem}\label{lem:partial-monogamous-id-sym}
    Identities and symmetries are partial left-monogamous.
\end{lem}
\begin{proof}
    Again by \autoref{lem:monogamicity-preserved}, identities and symmetries are
    monogamous so they are also partial left-monogamous.
\end{proof}

\begin{lem}
    Given partial left-monogamous cospans
    \(\cospan{m}{F}{n}\) and
    \(\cospan{n}{G}{p}\), the composition \(
    (\cospan{m}{F}{n})
    \seq
    (\cospan{n}{G}{p})
    \) is partial left-monogamous.
\end{lem}
\begin{proof}
    We only need to check the in-degree of nodes, as the out-degree can be
    arbitrary.
    Only the nodes in the image of \(n \to G\) have their in-degree modified;
    they will gain the in-tentacles of the corresponding nodes in the image
    of \(n \to F\).
    Initially the nodes in \(n \to G\) must have in-degree \(0\) by partial
    monogamy.
    They can only gain at most one in-tentacle from nodes in \(n \to F\)
    as each of these nodes has in-degree \(0\) or \(1\) and \(n \to G\) is
    mono.
    So the composite graph is partial left-monogamous.
\end{proof}

\begin{lem}
    Given partial left-monogamous cospans \(\cospan{m}{F}{n}\)
    and \(\cospan{p}{G}{q}\), the tensor \(
    (\cospan{m}{F}{n})
    \tensor
    (\cospan{n}{G}{p})
    \) is partial left-monogamous.
\end{lem}
\begin{proof}
    The elements of the original graphs are unaffected so
    the degrees are unchanged.
\end{proof}

\begin{defi}
    Let \(\plmcspdhyp\) be the sub-PROP of \(\cspdhyp\) containing only the
    partial left-monogamous cospans of hypergraphs.
\end{defi}

\begin{prop}
    The canonical trace is a trace on \(\plmcspdhyp\).
\end{prop}
\begin{proof}
    We must show that for any set of nodes in the image
    of \(x + n \to K\) merged by the canonical trace, at most one of them can
    have in-degree \(1\).
    But this must be the case because any image in the image of
    \(x + m \to K\) must have in-degree \(0\), and \(x + m \to K\) is
    moreover mono so it cannot merge nodes in the image of
    \(x + n \to K\).
\end{proof}

This category can be equipped with a comonoid structure.

\begin{defi}
    Let \(
    \morph{
        \comonoidtofrob
    }{
        \ccomon
    }{
        \frob
    }
    \) be the obvious embedding of \(\ccomon\) into \(\frob\), and let \(
    \morph{
        \tracedandcomonoidtofrob{\Sigma}
    }{
        \stmc{\Sigma} + \comon
    }{
        \smc{\Sigma} + \frob
    }
    \) be the copairing of \(\tracedtosymandfrob{\Sigma}\) and
    \(\comonoidtofrob\).
\end{defi}

As before, these PROP morphisms are summarised in \autoref{fig:roadmap}.
To show that partial left-monogamy is the correct notion to characterise terms
in a traced comonoid setting, it is necessary to ensure that the image of these
PROP morphisms lands in \(\plmcspdhyp\).

\begin{lem}
    The image of \(\frobtohyp{\Sigma} \circ \comonoidtofrob\) is in
    \(\plmcspdhyp\).
\end{lem}
\begin{proof}
    By definition.
\end{proof}

\begin{cor}
    The image of \(
    \termandfrobtohypsigma \circ \tracedandcomonoidtofrob{\Sigma}
    \) is in \(\plmcspdhyp\).
\end{cor}

To show the correspondence between \(\stmcsigma + \ccomon\) and
\(\plmcspdhyp\), we use a similar strategy to the one of
\autoref{thm:termtohyp-image}.

\begin{lem}\label{lem:discrete-mono}
    Given a discrete hypergraph \(X \in \hypsigma\), any cospan
    \(\cospan{m}[f]{X}{n}\) with \(f\) mono is in the image of
    \(\frobtohypsigma \circ \comonoidtofrob\).
\end{lem}

This leads to a version of \autoref{cor:stmc-graph-iso} for traced terms
additionally equipped with a comonoid structure: every such term is in
one-to-one correspondence with isomorphism classes of partial left-monogamous
cospans of hypergraphs.

\begin{thm}\label{thm:comonoid-fully-complete}
    \(\stmcsigma + \ccomon \cong \plmcspdhyp\).
\end{thm}
\begin{proof}
    Since \(\termandfrobtohypsigma\) and \(\comonoidtohypsigma\) are faithful,
    it suffices to show that a cospan \(\cospan{m}{F}{n}\) in
    \(\plmcspdhyp\) can be decomposed into a traced cospan in which every
    component under the trace is in the image of either
    \(\termandfrobtohypsigma\) or \(\frobtohypsigma \circ \comonoidtofrob\).
    This is achieved by taking the construction of \autoref{thm:termtohyp-image}
    and allowing the first component to be partial left-monogamous; by
    \autoref{lem:discrete-mono} this is in the image of
    \(\frobtohypsigma \circ \comonoidtofrob\).
    The remaining components remain in the image of \(\termtohypsigma\).
    Subsequently, the entire traced cospan must be in the image of \(
    \termandfrobtohypsigma \circ [\tracedtosymandfrob, \comonoidtofrob]
    \).
\end{proof}

Much like with the result for traced terms and partial monogamous cospans of
hypergraphs, this result is key for both performing rewriting on traced
comonoid terms and for showing the soundness and completeness of such a
rewriting system.

\section{Graph rewriting}

We have now shown that we can reason up to the axioms of symmetric traced
categories with a comonoid structure using hypergraphs: string diagrams equal by
topological deformations are translated into isomorphic cospans of hypergraphs.
Already this is a useful tool to have for reasoning with string diagrams, but
ultimately this only allows us to `move boxes and wires about' while preserving
the connectivity between them; the boxes themselves cannot be altered.

Reasoning with string diagrams becomes more interesting when performed modulo
a \emph{monoidal theory}, in which we have additional equations between terms.
These equations can be used to replace certain patterns of boxes and wires
with others, actually changing the make-up of a diagram.
Some examples of useful monoidal theories will be examined in
\autoref{sec:case-studies}.

The process of translating one traced string diagram term into another is
defined formally as \emph{term rewriting}.

\begin{defi}[Term rewriting]\label{def:term-rewriting}
    A \emph{rewriting system} \(\mcr\) for a traced PROP \(\stmc{\Sigma}\)
    consists of a set of \emph{rewrite rules} \(
    \rrule{
        \iltikzfig{strings/category/f}[box=l,colour=white,dom=i,cod=j]
    }{
        \iltikzfig{strings/category/f}[box=r,colour=white,dom=i,cod=j]
    }
    \).
    Given terms \(
    \iltikzfig{strings/category/f}[box=g,colour=white,dom=m,cod=n]
    \) and \(
    \iltikzfig{strings/category/f}[box=h,colour=white,dom=m,cod=n]
    \) in \(\stmc{\generators}\) we write \(
    \iltikzfig{strings/category/f}[box=g,colour=white]
    \rewrite[\mcr]
    \iltikzfig{strings/category/f}[box=h,colour=white]
    \) if there exists rewrite rule \((
    \iltikzfig{strings/category/f}[box=l,colour=white,dom=i,cod=j],
    \iltikzfig{strings/category/f}[box=r,colour=white,dom=i,cod=j]
    )\) in \(\mcr\) and \(
    \iltikzfig{strings/category/f-2-2}[box=c,colour=white,dom1=j,dom2=m,cod1=i,cod2=n]
    \) in \(\stmc{\Sigma}\) such that \(
    \iltikzfig{strings/category/f}[box=g,colour=white]
    =
    \iltikzfig{strings/rewriting/rewrite-l}
    \) and \(
    \iltikzfig{strings/category/f}[box=h,colour=white]
    =
    \iltikzfig{strings/rewriting/rewrite-r}
    \) by axioms of STMCs.
\end{defi}

Term rewriting using string diagrams is convenient for pen-and-paper reasoning,
but as we have already mentioned is difficult to automate.
Now armed with hypergraph interpretations of string diagram terms, we can turn
to \emph{graph rewriting}.

Graph rewriting is a deeply studied field with a plethora of techniques.
To tie in with our categorical motivations, we use the framework of
\emph{double pushout rewriting} (DPO rewriting), which has its roots in the
70s~\cite{ehrig1976parallelism}.
Rather than using the original presentation, we use an extension, known as
\emph{double pushout rewriting with interfaces}
(DPOI rewriting)~\cite{bonchi2017confluence}.
This definition is advantageous as confluence of graph rewriting using
interfaces has been shown to be decidable~\cite[Cor.\ 20]{bonchi2017confluence}.

In DPO rewriting, rewrite rules are specified as pairs of morphisms from some
shared interface graph.
While the framework can be applied to all manner of graph-based structures, the
definitions we present will be in terms of hypergraphs.

\begin{defi}[DPO rule]
    Given interfaced hypergraphs \(
    \cospan{i}[a_1]{L}[a_2]{j}
    \) and \(
    \cospan{i}[b_1]{R}[b_2]{j}
    \), their \emph{DPO rule} in \(\hypsigma\) is a span \(
    \spann{L}[[a_1,a_2]]{i+j}[[b_1,b_2]]{R}
    \).
\end{defi}

A DPO rule makes up part of a larger diagram that represents the application of
said rule in a larger context graph.

\begin{defi}[DPO(I) rewriting]\label{def:dpo}
    Let \(\mcr\) be a set of DPO rules.
    Then, for morphisms \(G \leftarrow m+n\) and \(H \leftarrow m+n\) in
    \(\hypsigma\), there is a rewrite \(G \trgrewrite[\mcr] H\) if there
    exist a rule \(
    \spann{L}{i+j}{R} \in \mcr
    \) and cospan \(
    \cospan{i+j}{C}{m+n} \in \hypsigma
    \) such that the diagram in the left of \autoref{fig:dpo} commutes.
\end{defi}

\begin{figure}
    \centering
    \raisebox{1em}{    \begin{tikzcd}[row sep = small, column sep = small]
        L \arrow{d}
        & i + j \arrow{l}\arrow{r}\arrow{d}
        & R \arrow{d} \\
        G \arrow["\urcorner"{anchor=center, pos=0.125}, draw=none]{ur}
        & C \arrow{l}\arrow{r}
        & H \arrow["\ulcorner"{anchor=center, pos=0.125}, draw=none]{ul} \\
        & m + n \arrow{ul}\arrow{u}\arrow{ur}
    \end{tikzcd}}
    \qquad
        \begin{tikzcd}[column sep=large]
        L \arrow[swap]{d}{f}
        &
        i+j
        \arrow[swap]{l}{a := [a_1, a_2]}
        \arrow{d}{c := [c_1, c_2]}
        \\
        G
        \arrow["\urcorner"{anchor=center, pos=0.125}, draw=none]{ur}
        &
        C
        \arrow{l}{g}
        \\
        &
        m+n
        \arrow{ul}{[b_1,b_2]}
        \arrow[swap]{u}{d := [d_1,d_2]}
    \end{tikzcd}
    \caption{The DPO diagram and a pushout complement}
    \label{fig:dpo}
\end{figure}

The first thing to note is that the graphs in the DPO diagram have a
\emph{single} interface \(G \leftarrow m + n\) instead of the cospans \(
\cospan{m}{G}{n}
\) we are used to.
Before performing DPO rewriting in \(\hypsigma\), the interfaces must be
`folded' into one.

\begin{defiC}[{\cite[Sec.\ 4.4]{bonchi2022stringa}}]
    Let \(\morph{\foldinterfaces}{\smc{\Sigma} + \frob}{\smc{\Sigma} + \frob}\)
    be defined as having action \(
    \iltikzfig{strings/category/f}[box=f,colour=white,dom=m,cod=n]
    \mapsto
    \iltikzfig{strings/rewriting/folding}[box=f,colour=white,dom=m,cod=n]
    \).
\end{defiC}

Note that the result of applying \(\foldinterfaces\) is no longer a valid
traced term due to the input wire bending round to the output.
This is not an issue for the purpose of rewriting traced terms, as long as we
`unfold' the interfaces once rewriting is completed.
When viewed through the lens of hypergraphs, the distinction is even less
important.

\begin{propC}[{\cite[Prop.\ 4.8]{bonchi2022string}}]
    For \(\iltikzfig{strings/category/f}[box=f,colour=white,dom=m,cod=n]\)
    in \(\smcsigma + \frob\), if \(
    \termandfrobtohypsigma[\iltikzfig{strings/category/f}[box=f,colour=white]]
    =
    \cospan{m}[i]{F}[o]{n}
    \) then \(
    \foldinterfaces[
        \termandfrobtohypsigma[\iltikzfig{strings/category/f}[box=f,colour=white]]
    ]
    \) is isomorphic to \(
    \cospan{0}[]{F}[i+o]{m+n}
    \).
\end{propC}

In order to apply a given DPO rule \(\spann{L}{i+j}{R}\) in some larger
graph \(\cospan{m}{G}{n}\), a morphism \(L \to G\) must first be identified.
The next step is to `cut out' the components of \(L\) that exist in \(G\),
using what is known as a \emph{pushout complement}.
The pushout complement is a key part of DPO rewriting as it defines the
rewriting context for a given rule in a larger term.

\begin{defi}[Pushout complement]\label{def:pushout-complement}
    Let \(i + j \to L \to G \rightarrow m + n\) be morphisms in \(\hypsigma\);
    their \emph{pushout complement} is an object \(C\)
    with morphisms \(i + j \to C \to G\) such that \(\cospan{L}{G}{C}\) is a
    pushout and the right diagram in \autoref{fig:dpo} commutes.
\end{defi}

Given a rule \(L \leftarrow i+j \to R\) and morphism \(L \to G\), a pushout
complement \(i+j \to C \to G\) is effectively \(C\) is `\(G\) with \(L\) cut out
of it'.

Once a pushout complement is computed, the pushout of
\(\spann{C}{i+j}{R}\) can be performed to obtain the completed rewrite \(H\),
Put simply, the pushout \(\cospan{L}{H}{R}\) pastes \(R\) into \(C\) along the
dangling wires left by cutting \(L\) out of \(H\).

\begin{exa}
    Consider the following DPO rule: \[
        \raisebox{2em}{\(\rrule{
                \iltikzfig{strings/category/generator}[box=e_1]
            }{
                \iltikzfig{strings/category/generator}[box=e_2]
            }\)}
        \qquad
        
\begin{tikzcd}[column sep = small, row sep = small]
    \iltikzfig{graphs/dpo/example/l}
    &
    \arrow{l}
    \iltikzfig{graphs/dpo/example/k}
    \arrow{r}
    &
    \iltikzfig{graphs/dpo/example/r}
\end{tikzcd}

    \]
    Now consider the term \(
    \iltikzfig{graphs/dpo/example/term}
    \); a DPO rewrite is performed as follows:
    \begin{center}
        
\begin{tikzcd}[column sep = small, row sep = small]
    \scalebox{0.8}{\iltikzfig{graphs/dpo/example/l}}
    \arrow{d}
    &
    \arrow{l}
    \scalebox{0.8}{\iltikzfig{graphs/dpo/example/k}}
    \arrow{d}
    \arrow{r}
    &
    \scalebox{0.8}{\iltikzfig{graphs/dpo/example/r}}
    \arrow{d}
    \\
    \scalebox{0.8}{\iltikzfig{graphs/dpo/example/g}}
    &
    \arrow{l}
    \scalebox{0.8}{\iltikzfig{graphs/dpo/example/c}}
    \arrow{r}
    &
    \scalebox{0.8}{\iltikzfig{graphs/dpo/example/h}}
    \\
    &
    \scalebox{0.8}{\iltikzfig{graphs/dpo/example/j}}
    \arrow{ul}
    \arrow{u}
    \arrow{ur}
\end{tikzcd}

    \end{center}
    As one would expect, the resulting hypergraph is the interpretation of term
    \(\iltikzfig{graphs/dpo/example/term-rewritten}\).
\end{exa}

A pushout complement may not exist for a given rule and matching;
there are two conditions that must be satisfied for this to be the case.
The first ensures that all the sources and targets of a hyperedge are present
in a candidate context.

\begin{defi}[No-dangling-hyperedges condition~{\cite[Prop.\ 3.3.4]{corradini1997algebraic}}]
    Given morphisms \(i + j \xrightarrow{a} L \xrightarrow{f} G\) in \(\hypsigma\),
    they satisfy the \emph{no-dangling} condition if, for every hyperedge not in
    the
    image of \(f\), each of its source and target nodes is either not in the
    image of \(f\) or are in the image of \(f \circ a\).
\end{defi}

\begin{exa}
    The following morphisms do not satisfy the no-dangling-hyperedges
    condition.
    \[
        \iltikzfig{graphs/dpo/no-dangling/k}
        \xrightarrow{a}
        \iltikzfig{graphs/dpo/no-dangling/l}
        \xrightarrow{f}
        \iltikzfig{graphs/dpo/no-dangling/g}
    \]
    To obtain the pushout complement we `cut out' any vertices in the
    rightmost graph which are in the image of \(f\) but not the image of
    \(f \circ a\), as the latter are the interfaces of the rule.
    However, if we cut out the vertices labelled \(2\) and \(3\), the edge
    \(e_2\) will be left with `dangling' tentacles connected to no vertices, a
    malformed hypergraph.
    \begin{center}
        \begin{tikzcd}
            \iltikzfig{graphs/dpo/no-dangling/l}
            \arrow{d}{f}
            &
            \iltikzfig{graphs/dpo/no-dangling/k}
            \arrow{l}{a}
            \arrow{d}
            \\
            \iltikzfig{graphs/dpo/no-dangling/g}
            &
            \iltikzfig{graphs/dpo/no-dangling/c}
            \arrow{l}
        \end{tikzcd}
    \end{center}
    This illustrates why every node in the image of \(f\) must also be in the
    image of \(f \circ a\).
\end{exa}

The second condition enforces that merging of vertices is well-defined.

\begin{defi}[No-identification condition~{\cite[Prop.\ 3.3.4]{corradini1997algebraic}}]
    Given morphisms \(i + j \xrightarrow{a} L \xrightarrow{f} G\) in
    \(\hypsigma\), they satisfy the \emph{no-identification} condition if any
    two distinct elements merged by \(f\) are also in the image of \(f \circ a\).
\end{defi}

\begin{exa}
    The following morphisms do not satisfy the no-identification
    condition.
    \[
        \iltikzfig{graphs/dpo/no-identification/k}
        \xrightarrow{a}
        \iltikzfig{graphs/dpo/no-identification/l}
        \xrightarrow{f}
        \iltikzfig{graphs/dpo/no-identification/g}
    \]
    When trying to construct a pushout complement, the edge \(e_2\) will be
    removed.
    However, since vertices \(2\) and \(3\) are not mapped from the rule
    interfaces, there is no reason that a pushout would glue them together so
    that they are merged in the final graph.
    Therefore no pushout complement can exist.
    \begin{center}
        \begin{tikzcd}
            \iltikzfig{graphs/dpo/no-identification/l}
            \arrow{d}{f}
            &
            \iltikzfig{graphs/dpo/no-identification/k}
            \arrow{l}{a}
            \arrow{d}{b}
            \\
            \iltikzfig{graphs/dpo/no-identification/g}
            &
            \iltikzfig{graphs/dpo/no-identification/c}
            \arrow{l}
        \end{tikzcd}
    \end{center}
    Although the above diagram may look reasonable, it is \emph{not} a pushout;
    the pushout can only merge vertices that are in the image of \(b\).
\end{exa}

With these two conditions, we can establish when pushout complements exist for
a pair of hypergraph homomorphisms.
If there is a pushout complement, then there is an opportunity for a rewrite.

\begin{propC}[{\cite[Prop.\ 3.3.4]{corradini1997algebraic}}]
    \label{prop:pushout-complement}
    Morphisms \(i + j \to L \to G\) have at least one pushout complement if
    and only if they satisfy the no-dangling and no-identification conditions.
\end{propC}

\begin{defi}
    Given a partial monogamous cospan \(\cospan{i}{L}{j}\), a morphism
    \(L \to G\) is called a \emph{matching} if it has at least one pushout
    complement.
\end{defi}

In certain settings, known as
\emph{adhesive categories}~\cite{lack2004adhesive}, it is possible to be more
precise about the number of pushout complements for a given matching and rewrite
rule.

\begin{propC}[{\cite[Lem.\ 15]{lack2004adhesive}}]
    In an adhesive category, pushout complements of \(
    i + j \xrightarrow{a} L \to G\)
    are unique if they exist and \(a\) is mono.
\end{propC}

\begin{propC}[{\cite[Prop.\ 3.5]{lack2005adhesive}}]
    \(\hypsigma\) is adhesive.
\end{propC}

A given pushout complement uniquely determines the rewrite performed, so it
might seem advantageous to always have exactly one.
However, when writing modulo traced comonoid structure there are settings where
having multiple pushout complements is beneficial.

\subsection{Rewriting with traced structure}

In the Frobenius case considered in~\cite{bonchi2022string}, all valid
pushout complements correspond to a valid rewrite, but this is not the case for
the traced monoidal case.
In the symmetric monoidal case considered in \cite{bonchi2022stringa},
pushout complements that correspond to a valid rewrite in the non-traced
symmetric monoidal case are identified as \emph{boundary complements}.
As with monogamy, we will weaken this definition to find one suitable for the
traced case.

\begin{defi}[Traced boundary complement]
    \label{def:traced-boundary-complement}
    A pushout complement as in \autoref{def:pushout-complement} is called a
    \emph{traced boundary complement} if \(c_1\) and \(c_2\) are mono and \(
    \cospan{j+m}[[c_2,d_1]]{C}[[d_2,c_1]]{{n+i}}
    \) is a partial monogamous cospan.
\end{defi}

Unlike~\cite{bonchi2022stringa}, we do not enforce that the matching is mono,
as this cuts off potential rewrites in the \emph{traced} setting, such as a
matching inside a loop: \[
    \iltikzfig{graphs/dpo/matchings/trace-l}
    \to
    \iltikzfig{graphs/dpo/matchings/trace-g}
\]

The definition of a traced boundary complement leads to a definition of double
pushout rewriting for traced terms.

\begin{defi}[Traced DPO]
    For morphisms \(G \leftarrow m+n\) and \(H \leftarrow m+n\) in
    \(\hypsigma\), there is a traced rewrite \(G \trgrewrite[\mcr] H\) if there
    exists a rule \(
    \spann{L}{i+j}{G} \in \mcr
    \) and cospan \(
    \cospan{i+j}{C}{n+m} \in \hypsigma
    \) such that the diagram in \autoref{def:dpo} commutes and
    \(i+j \to C \to G\) is a
    traced boundary complement.
\end{defi}

Some intuition on the construction of traced boundary complements may be
required: this will be provided through a lemma and some examples.

\begin{lem}
    Consider the DPO diagram as in \autoref{fig:dpo}, and let
    \(i+j \to C \to G\) be a traced boundary complement.
    Let \(v \in i\) and \(w_0, w_1, \cdots w_k \in j\) such that
    \(f(a_1(v)) = f(a_2(w_0))\), \(f(a_1(v)) = f(a_2(w_1))\) and so on.
    Then either (1) there exists exactly one \(w_l\) not in the image of \(d_1\)
    such that \(c_1(v) = c_2(w_l)\); (2) \(c_1(v)\) is in the image of \(d_1\);
    or (3) \(c_1(v)\) has degree \((1,0)\).
    The same also holds for \(w \in j\), with the interface map as \(d_2\) and
    the degree as \((0,1)\).
\end{lem}
\begin{proof}
    Since \(c_1(v)\) is in the image of \(c\), it must have either degree
    \((0,0)\) or \((1,0)\) by partial monogamy.
    For it to have degree \((0, 0)\), it must either be in the image of one
    of \(c_2\) or \(d_2\).
    In the case of the former, this means that there must be a \(w_l\) such
    that \(c_1(v) = c_2(w_l)\), and only one such node as \(c_2\) is mono,
    so (1) is satisfied.
    In the case of the latter, (2) is immediately satisfied.
    Otherwise, (3) is satisfied.
    The proof for the flipped case is exactly the same.
\end{proof}

Often there can be valid rewrites in the realm of graphs that are non-obvious
in the term language.
This is because we are rewriting modulo \emph{yanking}.

\begin{exa}
    Consider the rule \(
    \rrule{
        \iltikzfig{graphs/dpo/split-loop/rule-lhs}
    }{
        \iltikzfig{graphs/dpo/split-loop/rule-rhs}
    }.
    \)
    The interpretation of this as a DPO rule in a valid traced boundary
    complement is illustrated below.
    
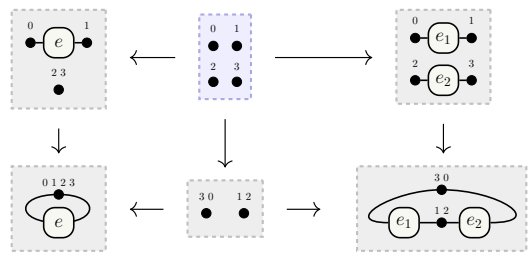
\begin{center}
        
    \begin{tikzcd}[column sep = small, row sep = small]
        \iltikzfig{graphs/dpo/split-loop/l}
        \arrow{d}
        &
        \arrow{l}
        \iltikzfig{graphs/dpo/split-loop/k}
        \arrow{d}
        \arrow{r}
        &
        \iltikzfig{graphs/dpo/split-loop/r}
        \arrow{d}
        \\
        \iltikzfig{graphs/dpo/split-loop/g}
        &
        \arrow{l}
        \iltikzfig{graphs/dpo/split-loop/c}
        \arrow{r}
        &
        \iltikzfig{graphs/dpo/split-loop/h}
    \end{tikzcd}

    \end{center}
    This corresponds to a valid term rewrite:
    \[
        \iltikzfig{graphs/dpo/split-loop/derivation-1}
        =
        \iltikzfig{graphs/dpo/split-loop/derivation-2}
        =
        \iltikzfig{graphs/dpo/split-loop/derivation-3}
        =
        \iltikzfig{graphs/dpo/split-loop/derivation-4}
    \]

    Note that applying yanking is required in the term setting because
    the traced wire is flowing from right to left, whereas applying the rule
    requires all wires flowing left to right.
\end{exa}

Unlike regular boundary complements, traced boundary complements need not be
unique.
However, this is not a problem since all pushout complements can be enumerated
given a rule and matching~\cite{heumuller2011construction}.

\begin{exa}
    Consider the rule \(
    \rrule{
        \iltikzfig{graphs/dpo/non-unique/rule-lhs}
    }{
        \iltikzfig{graphs/dpo/non-unique/rule-rhs}
    }
    \).
    Below are two valid traced boundary complements involving a matching of this
    rule.

    
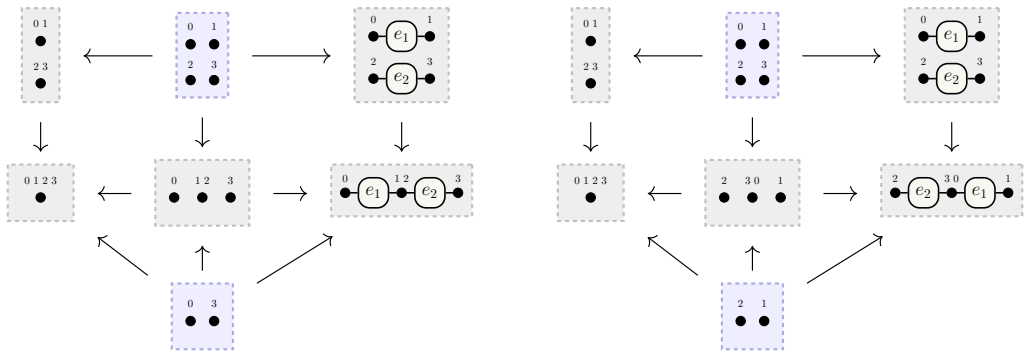
\begin{center}
        
    \begin{tikzcd}[column sep = small, row sep = small]
        \iltikzfig{graphs/dpo/non-unique/l}
        \arrow{d}
        &
        \arrow{l}
        \iltikzfig{graphs/dpo/non-unique/k}
        \arrow{d}
        \arrow{r}
        &
        \iltikzfig{graphs/dpo/non-unique/r}
        \arrow{d}
        \\
        \iltikzfig{graphs/dpo/non-unique/g}
        &
        \arrow{l}
        \iltikzfig{graphs/dpo/non-unique/c-1}
        \arrow{r}
        &
        \iltikzfig{graphs/dpo/non-unique/h-1}
        \\
        &
        \arrow{ul}
        \iltikzfig{graphs/dpo/non-unique/j-1}
        \arrow{u}
        \arrow{ur}
    \end{tikzcd}

        \quad
        
    \begin{tikzcd}[column sep = small, row sep = small]
        \iltikzfig{graphs/dpo/non-unique/l}
        \arrow{d}
        &
        \arrow{l}
        \iltikzfig{graphs/dpo/non-unique/k}
        \arrow{d}
        \arrow{r}
        &
        \iltikzfig{graphs/dpo/non-unique/r}
        \arrow{d}
        \\
        \iltikzfig{graphs/dpo/non-unique/g}
        &
        \arrow{l}
        \iltikzfig{graphs/dpo/non-unique/c-2}
        \arrow{r}
        &
        \iltikzfig{graphs/dpo/non-unique/h-2}
        \\
        &
        \arrow{ul}
        \iltikzfig{graphs/dpo/non-unique/j-2}
        \arrow{u}
        \arrow{ur}
    \end{tikzcd}

    \end{center}

    Once again, these derivations arise through yanking:
    \begin{gather*}
        \iltikzfig{graphs/dpo/non-unique/derivation-1}
        =
        \iltikzfig{graphs/dpo/non-unique/derivation-2}
        =
        \iltikzfig{graphs/dpo/non-unique/derivation-3a}
        =
        \iltikzfig{graphs/dpo/non-unique/derivation-4a}
        =
        \iltikzfig{graphs/dpo/non-unique/derivation-5a}
        \\
        \iltikzfig{graphs/dpo/non-unique/derivation-1}
        =
        \iltikzfig{graphs/dpo/non-unique/derivation-2}
        =
        \iltikzfig{graphs/dpo/non-unique/derivation-3b}
        =
        \iltikzfig{graphs/dpo/non-unique/derivation-4b}
        =
        \iltikzfig{graphs/dpo/non-unique/derivation-5b}
    \end{gather*}
\end{exa}

Rewriting modulo yanking also eliminates another foible of rewriting modulo
(non-traced) symmetric monoidal structure.
In the SMC case, the image of the matching must be `convex': any path between
nodes must also be captured by a matching.
This is not necessary in the traced case.

\begin{exa}
    Consider the rule \(
    \rrule{
        \iltikzfig{graphs/dpo/convex/example-l}
    }{
        \iltikzfig{graphs/dpo/convex/example-r}
    }
    \) and the term \(\iltikzfig{graphs/dpo/convex/example-g}\).
    Although it is not immediately obvious, there is in fact
    a matching of the former in the latter.
    Performing the DPO procedure yields the following:
    
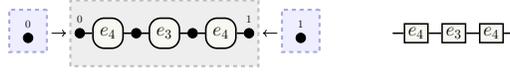
\begin{gather*}
        
\begin{tikzcd}[column sep = small, row sep = small]
    \iltikzfig{graphs/dpo/convex/l}
    \arrow{d}
    &
    \arrow{l}
    \iltikzfig{graphs/dpo/convex/k}
    \arrow{d}
    \arrow{r}
    &
    \iltikzfig{graphs/dpo/convex/r}
    \arrow{d}
    \\
    \iltikzfig{graphs/dpo/convex/g}
    &
    \arrow{l}
    \iltikzfig{graphs/dpo/convex/c}
    \arrow{r}
    &
    \iltikzfig{graphs/dpo/convex/h}
    \\
    &
    \iltikzfig{graphs/dpo/convex/j}
    \arrow{ul}
    \arrow{u}
    \arrow{ur}
\end{tikzcd}

    \end{gather*}
    In a non-traced setting this is an invalid rule, but it is possible with
    yanking.
    \begin{gather*}
        \iltikzfig{graphs/dpo/convex/example-g}
        =
        =
        \iltikzfig{graphs/dpo/convex/rewrite-4}
        \\[1em]=
        \iltikzfig{graphs/dpo/convex/rewrite-5}
        =
        \iltikzfig{graphs/dpo/convex/example-h}
    \end{gather*}
\end{exa}

We are almost ready to show the soundness and completeness of this DPO rewriting
system with respect to term rewriting, but first we prove some lemmas to be used
in the final result.
The first is a decomposition lemma, akin to that used in
\cite[Lem.\ 24]{bonchi2022stringa}.

\begin{lem}[Traced decomposition]\label{lem:traced-decomposition}
    Given partial monogamous cospans \(
    \cospan{m}[d_1]{G}[d_2]{n}
    \) and \(
    \cospan{i}[a_1]{L}[a_2]{j}
    \), and a morphism \(
    L \xrightarrow{f} G
    \) such that \(i+j \rightarrow L \rightarrow G\) satisfies the no-dangling
    and no-identification conditions, then there exists a partial monogamous
    cospan \(
    \cospan{j+m}[[c_2,d_1]]{C}[[c_1,d_2]]{i+n}
    \) such that \(
    \cospan{m}{G}{n}
    \) can be factored as
    \begin{gather}
        \trace{i}{
            \begin{array}{cc}
                \cospan{i}[a_1]{L}[a_2]{j} \\
                \tensor                    \\
                \cospan{m}{m}{m}
            \end{array}
            \seq
            \cospan{j+m}[[c_2,d_1]]{C}[[c_1,d_2]]{i+n}
        }
        \label{gath:decomposition}
    \end{gather}
    where \(
    \cospan{j+m}[c_2,d_1]{C}[c_1,d_2]{i+n}
    \) is a traced boundary complement.
\end{lem}
\begin{proof}
    Let \(
    i + j \xrightarrow{[c_1, c_2]} C \xleftarrow{[d_1, d_2]} m+n
    \) be defined as a traced boundary complement of \(
    i+j \xrightarrow{[a_1,a_2]} L \xrightarrow{f} G
    \), which exists as the no-dangling and no-identification condition is
    satisfied.
    We assign names to the various cospans in play, and reason string
    diagrammatically:
    \begin{align*}
        \iltikzfig{strings/category/f}[box=l,colour=white,dom=i,cod=j] & := \cospan{i}{L}{j}
                                                                       &
        \iltikzfig{strings/category/f-0-2}[box=\hat{l},colour=white,cod1=i,cod2=j]
                                                                       & :=
        \cospan{0}{L}{i+j}
        \\
        \iltikzfig{strings/category/f-2-2}[box=c,colour=white,dom1=j,dom2=m,cod1=i,cod2=n]
                                                                       & :=
        \cospan{j+m}[[c_2, d_1]]{C}[[c_1, d_2]]{i+n}
                                                                       &
        \iltikzfig{strings/category/f-2-2}[box=\hat{c},colour=white,dom1=i,dom2=j,cod1=m,cod2=n]
                                                                       & :=
        \cospan{i+j}[[c_1, c_2]]{C}[[d_1, d_2]]{m+n}
        \\
        \iltikzfig{strings/category/f}[box=g,colour=white,dom=m,cod=n]
                                                                       & :=
        \cospan{m}{G}{n}
                                                                       &
        \iltikzfig{strings/category/f-0-2}[box=\hat{g},colour=white,cod1=m,cod2=n]
                                                                       & :=
        \cospan{0}{G}{m+n}
    \end{align*}
    Note that the cospans in the left column are partial monogamous by definition
    of rewrite rules and traced boundary complements.
    We will show that  \(
    \iltikzfig{strings/category/f}[box=g,colour=white]
    \) can be decomposed into a form using the two cospans \(
    \iltikzfig{strings/category/f}[box=l,colour=white]
    \) and \(
    \iltikzfig{strings/category/f-2-2}[box=c,colour=white]
    \), along with identities.

    By using the compact closed structure of \(\cspdhyp\), we have the following:
    \begin{gather*}
        \iltikzfig{strings/category/f}[box=g,colour=white,dom=m,cod=n]
        =
        \iltikzfig{graphs/dpo/g-bent}
        \qquad
        \iltikzfig{strings/category/f-2-2}[box=\hat{c},colour=white,dom1=i,dom2=j,cod1=m,cod2=n]
        =
        \iltikzfig{graphs/dpo/cprime-as-c}
        \qquad
        \iltikzfig{strings/category/f-0-2}[box=\hat{l},colour=white,cod1=i,cod2=j]
        =
        \iltikzfig{strings/compact-closed/f-bent-input}[box=l,colour=white,cod=i,dom=j]
    \end{gather*}
    Since \(G\) is the pushout of \(
    L \xleftarrow{[a_1, a_2]} i+j \xrightarrow{[c_1, c_2]} C
    \) and pushout is cospan composition, we also have that \(
    \iltikzfig{strings/category/f-0-2}[box=\hat{g},colour=white,cod1=m,cod2=n]
    =
    \iltikzfig{graphs/dpo/lctilde}
    \).
    Putting this all together we can show that
    \begin{gather*}
        \iltikzfig{strings/category/f}[box=g,colour=white,dom=m,cod=n]
        =
        \iltikzfig{graphs/dpo/g-bent}
        =
        \iltikzfig{graphs/dpo/l-c-bent}
        \\[1em]
        =
        \iltikzfig{graphs/dpo/l-c-bent-1}
        =
        \iltikzfig{graphs/dpo/lc-bent-2}
        =
        \iltikzfig{strings/rewriting/rewrite-l}[dom=m,cod=n]
        \label{gath:deomposition}
    \end{gather*}
    Since the `loop' is constructed in the same manner as the canonical trace on
    \(\cspdhyp\) (and is therefore identical in the graphical notation), this is a
    term in the form of (\ref{gath:decomposition}).
\end{proof}

The second result relates to how cospans of hypergraphs need to be `bent'
using \(\foldinterfaces\) in order to be used in DPO rewriting.
We write \(
\foldinterfaces[\tracedtosymandfrob[\mcr]{\Sigma}]
\) for the pointwise map \(
(
\iltikzfig{strings/category/f}[box=l,colour=white],
\iltikzfig{strings/category/f}[box=r,colour=white]
)
\mapsto
(
\foldinterfaces[
    \tracedtosymandfrob[
    \iltikzfig{strings/category/f}[box=l,colour=white]
]{\Sigma}
],
\foldinterfaces[
    \tracedtosymandfrob[
    \iltikzfig{strings/category/f}[box=r,colour=white]
]{\Sigma}
]
).
\)

\begin{lem}\label{lem:switch-interfaces}
    Let \(
    \iltikzfig{strings/category/f-2-2}[box=c,colour=white,dom1=m,dom2=n,cod1=p,cod2=q]
    \) be a term in \(\smcsigma + \frob\); if \[
        \termandfrobtohypsigmac[
            \iltikzfig{strings/category/f-2-2}[box=c,colour=white]
        ]
        =
        \cospan{m + n}[[f_1,f_2]]{F}[[g_1,g_2]]{p + q}
    \] then \[
        \termandfrobtohypsigmac[
            \iltikzfig{strings/rewriting/c-folded}
        ]
        =
        \cospan{p + m}[[g_1,f_1]]{F}[[f_2,g_2]]{n + q}.
    \]
\end{lem}
\begin{proof}
    By definition of cups and caps in \(\cspdhyp\).
\end{proof}

We now show the main result of this section: that traced DPO rewriting on
partial monogamous cospans of hypergraphs coincides exactly with term rewriting.
That is to say, a term rewrite \(
\iltikzfig{strings/category/f}[box=f]
\rewrite[\mcr]
\iltikzfig{strings/category/f}[box=g]
\) is valid if and only if the hypergraph interpretation of \(
\iltikzfig{strings/category/f}[box=f]
\) can be rewritten using traced DPO to the hypergraph interpretation of \(
\iltikzfig{strings/category/f}[box=g]
\).

\begin{thm}\label{thm:traced-rewriting}
    Let \(\mcr\) be a rewriting system on \(\stmcsigma\).
    Then,
    \begin{gather*}
        \iltikzfig{strings/category/f}[box=g,colour=white,dom=m,cod=n]
        \rewrite[\mcr]
        \iltikzfig{strings/category/f}[box=h,colour=white,dom=m,cod=n]
        \,\,
        \text{if and only if}
        \,\,
        \termandfrobtohypsigma[
            \foldinterfaces[
                \tracedtosymandfrob[
                    \iltikzfig{strings/category/f}[box=g,colour=white]
                ]{\Sigma}
            ]
        ]
        \grewrite[
            \termandfrobtohypsigma[
                \foldinterfaces[
                    \tracedtosymandfrob[\mcr]{\Sigma}
                ]
            ]
        ]
        \termandfrobtohypsigma[
            \foldinterfaces[
                \tracedtosymandfrob[
                    \iltikzfig{strings/category/f}[box=h,colour=white]
                ]{\Sigma}
            ]
        ].
    \end{gather*}
\end{thm}
\begin{proof}
    First the \((\Rightarrow)\) direction.
    If \(
    \iltikzfig{strings/category/f}[box=g,colour=white]
    \rewrite[\mcr]
    \iltikzfig{strings/category/f}[box=h,colour=white]
    \) then we have \(
    \iltikzfig{strings/category/f}[box=g,colour=white]
    =
    \iltikzfig{strings/rewriting/rewrite-l}
    \) and \(
    \iltikzfig{strings/rewriting/rewrite-r}
    =
    \iltikzfig{strings/category/f}[box=h,colour=white]
    \); we must derive the DPO diagram in \(\hypsigma\).
    First we give names to the following cospans:
    \begin{alignat*}{3}
        \cospan{0}{L}{i + j}
         & \coloneqq
        \termandfrobtohypsigma[
            \foldinterfaces[
                \tracedtosymandfrobsigma[
                    \iltikzfig{strings/category/f}[box=l,colour=white]
                ]
            ]
        ]
         &           & =
        \termandfrobtohypsigma[
            \iltikzfig{strings/rewriting/l-folded}
        ]
        \\
        \cospan{0}{R}{i + j}
         & \coloneqq
        \termandfrobtohypsigma[
            \foldinterfaces[
                \tracedtosymandfrobsigma[
                    \iltikzfig{strings/category/f}[box=r,colour=white]
                ]
            ]
        ]
         &           & =
        \termandfrobtohypsigma[\iltikzfig{strings/rewriting/r-folded}]
        \\
        \cospan{0}{G}{m + n}
         & \coloneqq
        \termandfrobtohypsigma[
            \foldinterfaces[
                \tracedtosymandfrobsigma[
                    \iltikzfig{strings/category/f}[box=g,colour=white]
                ]
            ]
        ]
         &           & =
        \termandfrobtohypsigma[\iltikzfig{strings/rewriting/lc-folded-shifted}]
        \\
        \cospan{0}{H}{m + n}
         & \coloneqq
        \termandfrobtohypsigma[
            \foldinterfaces[
                \tracedtosymandfrobsigma[
                    \iltikzfig{strings/category/f}[box=h,colour=white]
                ]
            ]
        ]
         &           & =
        \termandfrobtohypsigma[\iltikzfig{strings/rewriting/rc-folded-shifted}]
    \end{alignat*}

    Moving into \(\smcsigma + \frob\), we have that \(
    \iltikzfig{strings/rewriting/lc-folded-shifted}
    =
    \iltikzfig{strings/rewriting/lc-folded}
    \); so by functoriality \(
    \termandfrobtohypsigma[
        \foldinterfaces[
            \tracedtosymandfrobsigma[
                \iltikzfig{strings/category/f}[box=g,colour=white]
            ]
        ]
    ]
    =
    \termandfrobtohypsigma[\iltikzfig{strings/rewriting/l-folded}]
    \seq
    \termandfrobtohypsigma[\iltikzfig{strings/rewriting/c-folded}]
    \), i.e.\ \(
    \cospan{0}{G}{m + n} =
    \cospan{0}{L}{i + j}
    \seq
    \cospan{i + j}{C}{m + n}
    \).
    Cospan composition is by pushout, so \(\cospan{L}{G}{C}\) is a pushout.
    Using the same reasoning, \(\cospan{R}{G}{C}\) is also a pushout; this
    gives us the DPO diagram.
    All that remains is to check that the aforementioned pushouts are traced
    boundary complements; this follows by \autoref{lem:switch-interfaces} as \(
    \termandfrobtohypsigma[
        \tracedandcomonoidtofrobsigma[
            \iltikzfig{strings/category/f-2-2}[box=c,colour=white]
        ]
    ]
    \) is partial monogamous.

    Now for the \(\ifdir\) direction: we assume we have a traced DPO
    rewrite, so there exist cospans \(
    \cospan{0}{L}{i + j},
    \cospan{0}{R}{i + j},
    \cospan{i + j}{C}{m + n}
    \) as above such that the DPO diagram commutes and
    \(i + j \to C \to G\) is a traced boundary complement.
    We must show that \(
    \iltikzfig{strings/category/f}[box=g,colour=white]
    =
    \iltikzfig{strings/rewriting/rewrite-l}
    \) and \(
    \iltikzfig{strings/category/f}[box=h,colour=white]
    =
    \iltikzfig{strings/rewriting/rewrite-r}
    \).

    We have that \(
    \cospan{0}{G}{m + n} =
    \cospan{0}{L}{i + j} \seq
    \cospan{i + j}[[c_1,c_2]]{C}[[d_1,d_2]]{m + n}
    \) as cospan composition is by pushout.
    Let \(
    \iltikzfig{strings/category/f-2-2}[box=c^\prime, colour=white,dom1=i,dom2=j,cod1=m,cod2=n]
    \) be the term in \(\smcsigma + \frob\) such that \(
    \termandfrobtohypsigma[
        \iltikzfig{strings/category/f-2-2}[box=c^\prime, colour=white]
    ]
    =
    \cospan{i + j}[[c_1,c_2]]{C}[[d_1,d_2]]{m + n}
    \), which exists as \(\termandfrobtohypsigma\) is full.

    The cospan \(\cospan{j + m}[[c_2,d_1]]{C}[[c_1,d_2]]{i + n}\)
    is partial monogamous because \(i + j \to C \to G\) is a traced
    boundary complement.
    Let \(
    \iltikzfig{strings/category/f-2-2}[box=c, colour=white,dom1=j,dom2=m,cod1=i,cod2=n]
    \)  be the term in \(\smcsigma + \frob\) such that \(
    \termandfrobtohypsigma[
        \iltikzfig{strings/category/f-2-2}[box=c, colour=white]
    ]
    =
    \cospan{j + m}[[c_2,d_1]]{C}[[c_1,d_2]]{i + n}
    \); by \autoref{lem:switch-interfaces}, we have \(
    \termandfrobtohypsigma[
        \iltikzfig{strings/rewriting/c-folded}
    ]
    =
    \cospan{i + j}[[c_1,c_2]]{C}[[d_1,d_2]]{m + n}
    .\)

    So we have that \(
    \termandfrobtohypsigma[
        \foldinterfaces[
            \iltikzfig{strings/category/f}[box=g,colour=white]
        ]
    ]
    =
    \termandfrobtohypsigma[
        \foldinterfaces[
            \iltikzfig{strings/category/f}[box=l,colour=white]
        ]
    ]
    \seq
    \termandfrobtohypsigma[
        \iltikzfig{strings/category/f-2-2}[box=c^\prime, colour=white]
    ]
    \); by fullness we derive that \(
    \iltikzfig{strings/rewriting/g-folded-box}
    =
    \iltikzfig{strings/rewriting/lc}
    =
    \iltikzfig{strings/rewriting/lc-folded}
    =
    \iltikzfig{strings/rewriting/lc-folded-shifted}
    \).
    This means that \(
    \foldinterfaces[
        \iltikzfig{strings/category/f}[box=g,colour=white]
    ]
    =
    \iltikzfig{strings/rewriting/lc-folded-shifted}
    \) so `unfolding' the interface gives us \(
    \iltikzfig{strings/category/f}[box=g,colour=white]
    =
    \iltikzfig{strings/rewriting/rewrite-l}
    \).
    Since \(
    \termandfrobtohypsigma[
        \iltikzfig{strings/category/f-2-2}[box=c, colour=white]
    ]
    \) is partial monogamous, \(
    \iltikzfig{strings/category/f-2-2}[box=c, colour=white]
    \) is in \(\stmcsigma\).
    As the trace in \(\stmcsigma\) is the canonical trace, the entire term is in
    \(\stmcsigma\), completing the proof.
    The same procedure holds for rewriting from the other direction.
\end{proof}

\subsection{Rewriting with traced comonoid structure}

It is straightforward to adapt the results for rewriting with partial monogamous
cospans to those for rewriting with partial \emph{left}-monogamous cospans.
In this case, the definition of traced boundary complement is too restrictive,
so must be weakened to permit more valid rewrites.

\begin{defi}[Traced left-boundary complement]
    \label{def:traced-left-boundary-complement}
    For partial left-monogamous cospans \(
    \cospan{i}[a_1]{L}[a_2]{j}
    \) and \(
    \cospan{n}[b_1]{G}[b_2]{m} \in \hypsigma
    \), a pushout complement as in \autoref{def:traced-boundary-complement}
    is called a \emph{traced left-boundary complement} if \(c_2\)
    is mono and \(
    \cospan{j+m}[[c_2,d_1]]{C}[[c_1,d_2]]{{i+n}}
    \) is a partial left-monogamous cospan.
\end{defi}

\begin{defi}[Traced comonoid DPO]
    For morphisms \(G \leftarrow m+n\) and \(H \leftarrow m+n\) in
    \(\hypsigma\), there is a traced comonoid rewrite \(G \trgrewrite[\mcr] H\)
    if there exists a rule \(
    \spann{L}{i+j}{G} \in \mcr
    \) and cospan \(
    \cospan{i+j}{C}{n+m} \in \hypsigma
    \) such that the diagram in \autoref{def:dpo} commutes and \(i+j \to C \to G\) is a
    traced left-boundary complement.
\end{defi}

\begin{exa}
    As with traced DPO, there may be multiple valid traced comonoid DPO rewrites
    for a given rule and instance in a larger graph.
    Consider the following rule and its interpretation.
    \begin{gather*}
        \rrule{
            \iltikzfig{graphs/dpo/non-unique-comonoid/rule-lhs}
        }{
            \iltikzfig{graphs/dpo/non-unique-comonoid/rule-rhs}
        }
        \qquad
        \raisebox{-3em}
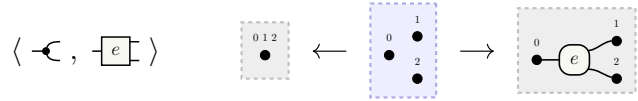
{
    \begin{tikzcd}[column sep = small]
        \iltikzfig{graphs/dpo/non-unique-comonoid/l}
        &
        \arrow{l}
        \iltikzfig{graphs/dpo/non-unique-comonoid/k}
        \arrow{r}
        &
        \iltikzfig{graphs/dpo/non-unique-comonoid/r}
    \end{tikzcd}
}
    \end{gather*}
    Two valid rewrites are as follows:
    
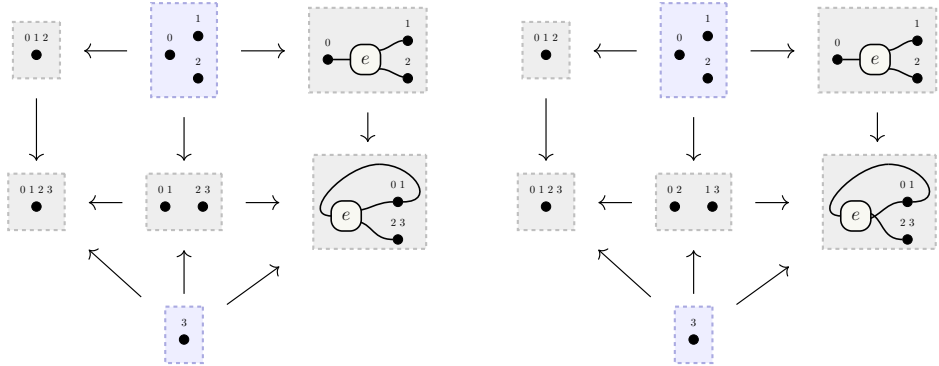
\begin{center}
        
\begin{tikzcd}[column sep = small, row sep = small]
    \iltikzfig{graphs/dpo/non-unique-comonoid/l}
    \arrow{d}
    &
    \arrow{l}
    \iltikzfig{graphs/dpo/non-unique-comonoid/k}
    \arrow{d}
    \arrow{r}
    &
    \iltikzfig{graphs/dpo/non-unique-comonoid/r}
    \arrow{d}
    \\
    \iltikzfig{graphs/dpo/non-unique-comonoid/g}
    &
    \arrow{l}
    \iltikzfig{graphs/dpo/non-unique-comonoid/c1}
    \arrow{r}
    &
    \iltikzfig{graphs/dpo/non-unique-comonoid/h1}
\end{tikzcd}

        \quad
        
\begin{tikzcd}[column sep = small, row sep = small]
    \iltikzfig{graphs/dpo/non-unique-comonoid/l}
    \arrow{d}
    &
    \arrow{l}
    \iltikzfig{graphs/dpo/non-unique-comonoid/k}
    \arrow{d}
    \arrow{r}
    &
    \iltikzfig{graphs/dpo/non-unique-comonoid/r}
    \arrow{d}
    \\
    \iltikzfig{graphs/dpo/non-unique-comonoid/g}
    &
    \arrow{l}
    \iltikzfig{graphs/dpo/non-unique-comonoid/c2}
    \arrow{r}
    &
    \iltikzfig{graphs/dpo/non-unique-comonoid/h2}
\end{tikzcd}

    \end{center}
    The first rewrite is the `obvious' one, but the second also holds by
    cocommutativity:
    \begin{gather*}
        \iltikzfig{graphs/dpo/non-unique-comonoid/rewrite-1}
        =
        \iltikzfig{graphs/dpo/non-unique-comonoid/rewrite-2a}
        \qquad
        \iltikzfig{graphs/dpo/non-unique-comonoid/rewrite-1}
        =
        \iltikzfig{graphs/dpo/non-unique-comonoid/rewrite-2b}
        =
        \iltikzfig{graphs/dpo/non-unique-comonoid/rewrite-3b}
    \end{gather*}
\end{exa}

To show that traced comonoid rewriting is sound and complete with respect to
traced comonoid term rewriting we follow the same procedure as for the traced
setting.

\begin{lem}[Traced comonoid decomposition]\label{lem:traced-comonoid-decomposition}
    Given partial left-monogamous cospans \(
    \cospan{m}[d_1]{G}[d_2]{n}
    \) and \(
    \cospan{i}[a_1]{L}[a_2]{j}
    \), along with a morphism \(
    L \xrightarrow{f} G
    \) such that \(i+j \rightarrow L \rightarrow G\) satisfies the no-dangling
    and no-identification conditions, then there exists a partial
    left-monogamous cospan \(
    \cospan{j+m}[[c_2,d_1]]{C}[[c_1,d_2]]{i+n}
    \) such that \(
    \cospan{m}{G}{n}
    \) can be factored as
    \begin{gather*}
        \trace{i}{
            \begin{array}{cc}
                \cospan{i}[a_1]{L}[a_2]{j} \\
                \tensor                    \\
                \cospan{m}{m}{m}
            \end{array}
            \seq
            \cospan{j+m}[[c_2,d_1]]{C}[[c_1,d_2]]{i+n}
        }
    \end{gather*}
    where \(
    \cospan{j+m}[c_2,d_1]{C}[c_1,d_2]{i+n}
    \) is a traced left-boundary complement.
\end{lem}
\begin{proof}
    As \autoref{lem:traced-decomposition}, but with partial left-monogamous
    cospans.
\end{proof}

Traced comonoid decomposition is then used in exactly the same role as traced
decomposition in the previous section.

\begin{thm}\label{thm:traced-comonoid-rewriting}
    Let \(\mcr\) be a rewriting system on \(\stmcsigma + \ccomon\).
    Then, \[
        \iltikzfig{strings/category/f}[box=g,colour=white]
        \rewrite[\mcr]
        \iltikzfig{strings/category/f}[box=h,colour=white]
        \quad\text{if and only if}\quad
        \termandfrobtohypsigma[
            \foldinterfaces[
                \tracedandcomonoidtofrob[
                    \iltikzfig{strings/category/f}[box=g,colour=white]
                ]{\Sigma}
            ]
        ]
        \grewrite[
            \termandfrobtohypsigma[
                \foldinterfaces[
                    \tracedandcomonoidtofrob[\mcr]{\Sigma}
                ]
            ]
        ]
        \termandfrobtohypsigma[
            \foldinterfaces[
                \tracedandcomonoidtofrob[
                    \iltikzfig{strings/category/f}[box=h,colour=white]
                ]{\Sigma}
            ]
        ].
    \]
\end{thm}
\begin{proof}
    As \autoref{thm:traced-rewriting}, but with traced left-boundary complements
    and traced comonoid decomposition
    (\autoref{lem:traced-comonoid-decomposition}).
\end{proof}

This means that not only can we perform rewriting on traced terms modulo some
equational theory using partial monogamous cospans and traced boundary
complements; we can do the same for traced comonoid terms using partial
left-monogamous cospans and traced left-boundary complements.

\section{Case studies}\label{sec:case-studies}

\subsection{Cartesian structure}

One important class of categories with a traced comonoid structure are
\emph{traced Cartesian}, or \emph{dataflow},
categories~\cite{cazanescu1990new,hasegawa1997recursion}.
These categories are interesting because any traced Cartesian category has a
fixpoint operator~\cite[Thm. 3.1]{hasegawa1997recursion}.

\begin{defi}[Cartesian category~\cite{fox1976coalgebras}]
    A monoidal category is \emph{Cartesian} if its tensor is given by the
    Cartesian product.
\end{defi}

As a result of this, the unit is a terminal object in any Cartesian category,
and any object has a comonoid structure.
Cartesian categories are settings in which morphisms can be \emph{copied} and
\emph{discarded}.
These two operations are more clearly illustrated when viewed through the lens
of a monoidal theory.

\begin{defi}
    For a given PROP \(\stmc{\Sigma_\mathbf{C}}\) with a comonoid
    structure, the traced SMT \((
    \generators[\mathbf{Cart}_\mathbf{C}],
    \equations[\mathbf{Cart}_\mathbf{C}]
    )\) is defined with \(
    \generators[\mathbf{Cart}_\mathbf{C}] := \Sigma_\mathbf{C}
    \) and \(
    \equations[\mathbf{Cart}_\mathbf{C}]
    \) as the equations in \autoref{fig:cartesian-equations}.
\end{defi}

The hypergraph interpretations of these rules for a generator \(e\)
are shown in \autoref{fig:cartesian-graphs}.

\begin{rem}
    The combination of Cartesian equations with the underlying compact closed
    structure of \(\cspdhyp\) may prompt alarm bells, as a compact closed
    category in which the tensor is the Cartesian product is trivial.
    However, it is important to note that \(\cspdhyp\) is \emph{not} subject to
    these equations: it is only a setting for performing graph
    rewrites.
\end{rem}

\begin{figure}
    \centering
    \begin{tabular}{cc}
        \iltikzfig{strings/structure/cartesian/naturality-copy-lhs}[box=f,colour=white,dom=m,cod=n]
        \(=\)
        \iltikzfig{strings/structure/cartesian/naturality-copy-rhs}[box=f,colour=white,dom=m,cod=n]
         &
        \iltikzfig{strings/structure/cartesian/naturality-discard-lhs}[box=f,colour=white,dom=m]
        \(=\)
        \iltikzfig{strings/structure/cartesian/naturality-discard-rhs}[dom=m]
        \\[2em]
    \end{tabular}
    \caption{
        Equations of the monoidal theory \(\mathbf{Cart}_\mathcal{C}\),
        for an arbitrary generator \(f\)
    }
    \label{fig:cartesian-equations}
\end{figure}

\begin{figure}
    \centering
        \begin{tikzcd}[column sep = small]
        \iltikzfig{graphs/dpo/cartesian/copy/l}
        &
        \arrow{l}
        \iltikzfig{graphs/dpo/cartesian/copy/k}
        \arrow{r}
        &
        \iltikzfig{graphs/dpo/cartesian/copy/r}
    \end{tikzcd}
    \qquad
    \raisebox{1em}{    \begin{tikzcd}[column sep = small]
        \iltikzfig{graphs/dpo/cartesian/discard/l}
        &
        \arrow{l}
        \iltikzfig{graphs/dpo/cartesian/discard/k}
        \arrow{r}
        &
        \iltikzfig{graphs/dpo/cartesian/discard/r}
    \end{tikzcd}}
    \caption{
        Interpretations of equations in \(\mathbf{Cart}_\mathcal{C}\) for an
        arbitrary generator \(e\).
    }
    \label{fig:cartesian-graphs}
\end{figure}

Reasoning about fixpoints can be performed using the \emph{unfolding} rule,
which holds in any traced Cartesian category.

\begin{center}
    \iltikzfig{strings/traced/trace-rhs}[box=f,colour=white,dom=m,cod=n,trace=x]
    \(=\)
    \iltikzfig{circuits/examples/reasoning/unfolding/unfolding-1}[box=f,colour=white,dom=m,cod=n,trace=x]
    \(=\)
    \iltikzfig{circuits/examples/reasoning/unfolding/unfolding-2}[box=f,colour=white,dom=m,cod=n,trace=x]
    \(=\)
    \iltikzfig{circuits/examples/reasoning/unfolding/unfolding-3}[box=f,colour=white,dom=m,cod=n,trace=x]
\end{center}

In the syntactic setting, this requires the application of multiple
equations: the two counitality equations followed by the copy equation and
optionally some axioms of STMCs for housekeeping.
However, if we interpret this in the hypergraph setting, the comonoid equations
are absorbed into the notation so only one rewrite rule needs to be applied.

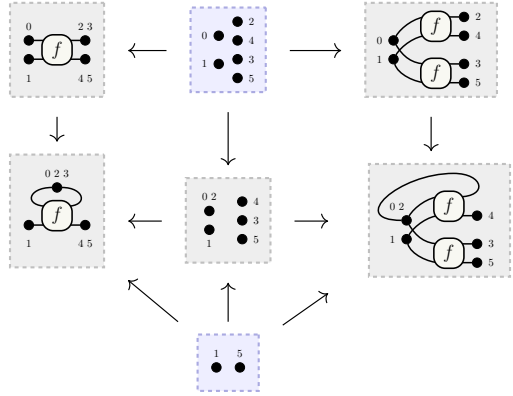
\begin{center}
    
\begin{tikzcd}[column sep = small, row sep = small, bezier bounding box=true]
    \scalebox{0.8}{\iltikzfig{graphs/dpo/unfolding/l}}
    \arrow{d}
    &
    \arrow{l}
    \scalebox{0.8}{\iltikzfig{graphs/dpo/unfolding/k}}
    \arrow{d}
    \arrow{r}
    &
    \scalebox{0.8}{\iltikzfig{graphs/dpo/unfolding/r}}
    \arrow{d}
    \\
    \scalebox{0.8}{\iltikzfig{graphs/dpo/unfolding/g}}
    &
    \arrow{l}
    \scalebox{0.8}{\iltikzfig{graphs/dpo/unfolding/c-1}}
    \arrow{r}
    &
    \scalebox{0.8}{\iltikzfig{graphs/dpo/unfolding/h-1}}
    \\
    &
    \arrow{ul}
    \iltikzfig{graphs/dpo/unfolding/j-1}
    \arrow{u}
    \arrow{ur}
\end{tikzcd}

\end{center}

The dual notion of traced \emph{cocartesian}
categories~\cite{bainbridge1976feedback} are also important in computer science:
a trace in a traced cocartesian category corresponds to \emph{iteration} in the
context of \emph{control flow}.
The details of this section could also be applied to the cocartesian case by
flipping all the directions and working with partial \emph{right}-monogamous
cospans.

However, attempting to combine the product and coproduct approaches for settings
with a \emph{biproduct} would simply yield the category \(\cspdhyp\), a
hypergraph category (\autoref{prop:frobenius-map}) subject to the Frobenius
equations in \autoref{fig:frobenius-equations}.
A category with biproducts is not necessarily subject to such equations, so this
would not be a suitable approach.

\subsection{Digital circuits}
\label{sec:digital-circuits}

As mentioned above, traced Cartesian categories are useful for reasoning in
settings with fixpoint operators.
One such setting is that of \emph{sequential digital circuits} built from
primitive logic gates: in \cite{ghica2024fully}, such circuits are
modelled as morphisms in a STMC.
Here, the trace models a feedback loop, and the comonoid structure represents
forking wires.
We are interested in using graph rewriting to implement an
\emph{operational semantics} for sequential circuits.

\begin{defi}
    Let \(\values\) be the set \(\{\bot,\belnaptrue,\belnapfalse,\top\}\) with
    a lattice structure defined by \(\belnaptrue \ljoin \belnapfalse = \top\)
    and \(\belnaptrue \lmeet \belnapfalse = \bot\).
\end{defi}

The elements of \(\values\) are \emph{values} that flow through wires in a
circuit.
The \(\belnaptrue\) and \(\belnapfalse\) values are the traditional true and
false values, the \(\bot\) value represents \emph{no} information
(a \emph{disconnected wire}) and the \(\top\) value represents \emph{both} true
and false at once (a \emph{short circuit}).

\begin{nota}
    For \(m \in \nat\), we write elements of \(\valuetuple{m}\)
    with an overline, e.g.\ \(
    \listvar{v} \in \valuetuple{3} := \belnaptrue\belnapfalse\belnaptrue
    \).
\end{nota}

Values are one part of the signature for sequential circuits.

\begin{defi}[Gate-level signature]
    Let the set \(\generators[\ccirc{}]\) of
    \emph{combinational gate-level circuit generators} be defined as \(\{
    \iltikzfig{circuits/components/gates/and},
    \iltikzfig{circuits/components/gates/or},
    \iltikzfig{circuits/components/gates/not},
    \iltikzfig{strings/structure/monoid/init}[colour=comb]
    \iltikzfig{strings/structure/comonoid/copy}[colour=comb],
    \iltikzfig{strings/structure/monoid/merge}[colour=comb],
    \iltikzfig{strings/structure/comonoid/discard}[colour=comb],
    \}\) and the set \(\generators[\scirc{}]\) of
    \emph{sequential gate-level circuit generators} be defined as \(\{
    \iltikzfig{circuits/components/values/vs}[val=\belnaptrue],
    \iltikzfig{circuits/components/values/vs}[val=\belnapfalse],
    \iltikzfig{circuits/components/values/vs}[val=\top],
    \iltikzfig{circuits/components/waveforms/delay}
    \}\).
\end{defi}

The \emph{combinational} generators are components that model functions; they
are respectively \(\andgate\), \(\orgate\) and \(\notgate\) gates along with
\emph{structural} constructs for introducing, forking, joining and eliminating
wires.
Combinational circuits are drawn with light blue backgrounds \(
\iltikzfig{strings/category/f}[box=f,colour=comb,dom=m,cod=n]
\); these circuits \emph{always} produce the same outputs given the same
inputs.

\begin{nota}
    The structural generators are defined for wires of width \(1\), but
    versions for arbitrary widths can be easily derived by axioms of STMCs.
    In diagrams, these are drawn the same as their single-bit
    counterparts: \(
    \iltikzfig{strings/structure/monoid/init}[colour=comb,obj=n]
    \), \(
    \iltikzfig{strings/structure/comonoid/copy}[colour=comb,obj=n]
    \), \(
    \iltikzfig{strings/structure/monoid/merge}[colour=comb,obj=n]
    \) and \(
    \iltikzfig{strings/structure/comonoid/discard}[colour=comb,obj=n]
    \).
\end{nota}

\emph{Sequential} generators are components that model state:
\emph{instantaneous values} and a delay of one unit of time.
Sequential circuits are drawn with green backgrounds \(
\iltikzfig{strings/category/f}[box=f,colour=seq,dom=m,cod=n]
\); these are circuits where the outputs may differ depending on past states.

The intended interpretation of the value generators is that they
produce the relevant value on the first cycle of execution, followed by the
disconnected \(\bot\) value after that.
Subsequently, there is no sequential \(\bot\) value generator; it is instead
modelled by the combinational \(
\iltikzfig{strings/structure/monoid/init}[colour=comb,obj=n]
\), as it will \emph{always} emit \(\bot\).

The delay component is the opposite: initially it outputs \(\bot\) but on
subsequent cycles it will output the input to the previous cycle.
What a `cycle' is can differ depending on application; the most obvious
interpretation is a D flipflop in a clocked circuit, but it could also model
the inertial delay on wires.

\begin{nota}
    We write \(
    \iltikzfig{circuits/components/values/vs}[val=v]
    \) for an arbitrary value \(v \in \values\); note that this could also
    include \(\bot\).
    For \(\listvar{v} \in \valuetuple{m}\) we collapse multiple
    value and delay generators into one as \(
    \iltikzfig{circuits/components/values/vs}[width=m]
    \) and \(
    \iltikzfig{circuits/components/waveforms/delay}[width=m]
    \), and write \(
    \iltikzfig{circuits/components/waveforms/register-shorthand}[width=m]
    :=
    \iltikzfig{circuits/components/waveforms/register}[width=m]
    \) for a \emph{register}.
\end{nota}

\begin{exa}[SR Latch~\cite{ghica2024fully}]
    An example of a sequential circuit component is an
    \emph{SR NOR latch}, illustrated in \autoref{fig:latch}; a \(\norgate\) gate
    is constructed as \(
    \iltikzfig{circuits/components/gates/nor}
    :=
    \iltikzfig{circuits/components/gates/nor-construction}
    \).
    SR latches are used to hold state: when the \(\mathsf{S}\) input (`set') is
    pulsed true then the \(\mathsf{Q}\) output is held at true until the
    \(\mathsf{R}\) input (`reset') is true.
    The state is held because there are delays in the gates and the wires; one
    of the feedback loops between the two \(\norgate\) gates will `win'.
    In \(\scircsigma\) this is modelled by using a different number of delay
    generators on the wires between the top and the bottom of the latch, as
    shown in \autoref{fig:latch}.
    The interpretation of this implementation as a cospan of hypergraphs is
    also depicted in \autoref{fig:latch}, where \(\lor\) is the interpretation
    of the \(\orgate\) gate, \(\neg\) is the interpretation of the \(\notgate\)
    gate, and \(\delta\) is the interpretation of the delay.
\end{exa}

\begin{figure}
    \centering
    \iltikzfig{circuits/examples/sr-latch/real-circuit}
    \qquad
    \iltikzfig{circuits/examples/sr-latch/circuit}
    \\[1em]
    \iltikzfig{graphs/examples/sr-latch}
    \caption{
        An SR NOR latch, a possible construction of said latch in
        \(\scircsigma\), and the hypergraph interpretation of this construction
    }
    \label{fig:latch}
\end{figure}

An operational semantics for sequential circuits is defined in terms of
equations showing how a circuit transforms an input value into an output value.
On sticking point is \emph{non-delay-guarded feedback}; this is tackled by
applying the Kleene fixpoint theorem and \emph{iterating}
a circuit multiple times until a fixpoint is reached.

\begin{defi}
    Let the set \(
    \equations[\scirc{}]
    \) of \emph{gate-level circuit equations} be those listed
    in~\autoref{fig:circuit-equations}, where \(
    \iltikzfig{circuits/components/gates/gate}
    \) is one of the logic gates and \(
    \gateinterpretation
    \) maps them to the corresponding truth table.
\end{defi}

\begin{figure*}
    \centering
    \(
    \equationdisplay{
        \iltikzfig{circuits/productivity/trace-delay}[core=f,state=\listvar{v},dom=m,cod=n,delay=y,trace=x,valwidth=z]
    }{
        \iltikzfig{circuits/productivity/mealy-rule}[trace=x,delays=y,values=z]
    }{
        \mealyeqn
    }
    \)
    \\[0.5em]
    \(
    \equationdisplay{
        \iltikzfig{circuits/instant-feedback/equation-lhs}[dom=m,cod=n,trace=x]
    }{
        \iltikzfig{circuits/instant-feedback/concrete-unfolding}[box=f,dom=m,cod=n]
    }{
        \instantfeedbackeqn
    }
    \)
    \\[0.5em]
    \(
    \equationdisplay{
        \iltikzfig{circuits/axioms/generalised-streaming-lhs}[box=f]
    }{
        \iltikzfig{circuits/axioms/generalised-streaming-rhs}[box=f]
    }{
        \streamingeqn
    }
    \)
    \quad
    \(
    \equationdisplay{
        \iltikzfig{circuits/axioms/gate-lhs}[gate=p,input=\listvar{v}]
    }{
        \iltikzfig{circuits/axioms/gate-rhs}[gate=p,input=\listvar{v}]
    }{
        \mathsf{Prim}
    }
    \)
    \\[0.5em]
    \(
    \equationdisplay{
        \iltikzfig{circuits/axioms/fork-lhs}[val=v]
    }{
        \iltikzfig{circuits/axioms/fork-rhs}[val=v]
    }{
        \forkeqn
    }
    \)
    \quad
    \(
    \equationdisplay{
        \iltikzfig{circuits/axioms/join-lhs}[val1=v,val2=w]
    }{
        \iltikzfig{circuits/axioms/join-rhs}[val1=v,val2=w]
    }{
        \joineqn
    }
    \)
    \quad
    \(
    \equationdisplay{
        \iltikzfig{circuits/axioms/stub-lhs}[val=v]
    }{
        \iltikzfig{strings/monoidal/empty}
    }{
        \stubeqn
    }
    \)
    \caption{
        Equations for reducing sequential circuits
    }
    \label{fig:circuit-equations}
\end{figure*}

The aim of the operational semantics is to transform a circuit of the form \(
\iltikzfig{circuits/productivity/productive-goal-lhs}[box=f,input=\listvar{v},dom=m,cod=n]
\) into one of the form \(
\iltikzfig{circuits/productivity/productive-goal-rhs}[box=G,output=\listvar{w},dom=m,cod=n]
\); this shows how a circuit processes inputs, producing a new
internal state and outputs.

First the circuit \(
\iltikzfig{strings/category/f}[box=f,colour=seq,dom=m,cod=n]
\) must be assembled into a special form known as `Mealy form' by
applying the \(\mealyeqn\) rule followed by the \(\instantfeedbackeqn\) rule.
Rewriting modulo comonoid structure is useful here since this procedure can
create multiple forks and stubs, which can clutter up a term; here the
connectivity is much clearer.

Once the circuit is in this form, it has a `combinational core' traced by
delay-guarded feedback.
Values can now be applied to it and the \(\streamingeqn\) rule used to copy
the core into a `now' copy and a `later' copy.
The `now' copy computes the outputs and the transition for the current cycle of
execution.
To determine the exact values, the \(\forkeqn\), \(\joineqn\), \(\stubeqn\)
and \(\mathsf{Prim}\) equations can be applied to reduce the core
to values.

\begin{exa}
    The rewriting procedure discussed above is applied to the SR NOR latch in
    Figures \ref{fig:latch-mealy}-\ref{fig:latch-reduce}.
\end{exa}

\begin{figure}
    \centering
    \iltikzfig{graphs/examples/sr-latch-pre-mealy}
    \caption{
        The result of rewriting the SR NOR latch using the \(\mealyeqn\) rule
    }
    \label{fig:latch-mealy}
\end{figure}

\begin{figure}
    \centering
    \iltikzfig{graphs/examples/sr-latch-unrolled}
    \caption{
        The result of rewriting \autoref{fig:latch-mealy} using the
        \(\instantfeedbackeqn\) rule
    }
    \label{fig:latch-unroll}
\end{figure}

\begin{figure}
    \centering
    \iltikzfig{graphs/examples/sr-latch-applied}
    \caption{
        Applying \autoref{fig:latch-unroll} to the input values
        \(\belnapfalse\belnaptrue\)
    }
    \label{fig:latch-applied}
\end{figure}

\begin{figure}
    \centering
    \iltikzfig{graphs/examples/sr-latch-streamed}
    \caption{
        Applying \(\streamingeqn\) to \autoref{fig:latch-applied}
    }
    \label{fig:latch-streamed}
\end{figure}

\begin{figure}
    \centering
    \scalebox{0.9}{\iltikzfig{graphs/examples/sr-latch-streamed-valprop}}
    \caption{Applying \(\forkeqn\) to \autoref{fig:latch-streamed}}
    \label{fig:latch-forked}
\end{figure}

\begin{figure}
    \scalebox{0.9}{\iltikzfig{graphs/examples/sr-latch-values}}
    \caption{
        Applying \(\forkeqn\), \(\gateeqn\) and \(\stubeqn\)
        to \autoref{fig:latch-forked}
    }
    \label{fig:latch-reduce}
\end{figure}

\section{Conclusion, related and further work}

We have shown how previous work on rewriting string diagrams modulo
Frobenius~\cite{bonchi2022string} and symmetric
monoidal~\cite{bonchi2022stringa} structure using hypergraphs can also be
adapted for rewriting modulo traced comonoid structure using a setting between
the two.

Graphical languages for traced categories have seen many applications, such as
to illustrate cyclic lambda calculi~\cite{hasegawa1997recursion}, or to reason
graphically about programs~\cite{schweimeier1999categorical}.
The presentation of traced categories as \emph{string diagrams} has existed
since the 90s~\cite{joyal1991geometry,joyal1996traced}; a soundness and
completeness theorem for traced string diagrams, folklore for many years
but only proven for certain signatures~\cite{selinger2011survey}, was finally
shown in~\cite{kissinger2014abstract}.
Combinatorial languages predate even this, having existed since at least the 80s
in the guise of
\emph{flowchart schemes}~\cite{stefanescu1990feedback,cazanescu1990new,cazanescu1994feedback}.
These diagrams have also been used to show the completeness of finite dimensional
vector spaces~\cite{hasegawa2008finite} with respect to traced categories and,
when equipped with a dagger, Hilbert spaces~\cite{selinger2012finite}.

We are not just concerned with diagrammatic languages as a standalone concept:
we are interested in performing \emph{graph rewriting} with them to reason about
monoidal theories.
This has been been studied in the context of traced categories before using
\emph{string graphs}~\cite{kissinger2012pictures,dixon2013opengraphs}.
We have instead opted to use the \emph{hypergraph} framework
of~\cite{bonchi2022string,bonchi2022stringa,bonchi2022stringb} instead, as it
allows rewriting modulo \emph{yanking}, is more extensible for rewriting modulo
comonoid structure, and one does not need to awkwardly reason modulo wire
homeomorphisms.

As mentioned during the case studies, there are still elements of the rewriting
framework that are somewhat informal.
One such issue involves defining rewrite spans for arbitrary subgraphs: this is
hard to do at a general level because the edges must be concretely specified in
DPO rewriting.
However, if we performed rewriting with
\emph{hierarchical hypergraphs}~\cite{alvarez-picallo2023functorial}, in
which edges can have hypergraphs as labels, we could `compress' the subgraph
into a single edge that can be rewritten: this is future work.

In regular PROP notation, wires are annotated with numbers in order to avoid
drawing multiple wires in parallel: when interpreted as hypergraphs a node is
created for each wire, and simple diagrams can quickly get very large.
The results of \cite{bonchi2022stringa} also extend to the multi-sorted case, in
which nodes are labelled in addition to wires.
We could use this in combination with the \emph{strictifiers}
of~\cite{wilson2023string}: these are additional generators for transforming
buses of wires into thinner or thicker ones.
This could drastically reduce the number of elements in a hypergraph, which is
ideal from a computational point of view.
Work has already begun on implementing the rewriting system for digital circuits
using these techniques.

\paragraph{Acknowledgements}

Thanks to Chris Barrett for comments on earlier versions of this paper,
and to the anonymous reviewers for their helpful insights.

\bibliographystyle{alphaurl}
\bibliography{refs/refs}

\newcommand{\etalchar}[1]{$^{#1}$}
\begin{thebibliography}{BGK{\etalchar{+}}22b}

\bibitem[AC04]{abramsky2004categorical}
Samson Abramsky and Bob Coecke.
\newblock A categorical semantics of quantum protocols.
\newblock In {\em Proceedings of the 19th {{Annual IEEE Symposium}} on
  {{Logic}} in {{Computer Science}}, 2004.}, pages 415--425, July 2004.
\newblock \href {https://doi.org/10.1109/LICS.2004.1319636}
  {\path{doi:10.1109/LICS.2004.1319636}}.

\bibitem[AGSZ23]{alvarez-picallo2023functorial}
Mario {Alvarez-Picallo}, Dan Ghica, David Sprunger, and Fabio Zanasi.
\newblock Functorial {{String Diagrams}} for {{Reverse-Mode Automatic
  Differentiation}}.
\newblock In Bartek Klin and Elaine Pimentel, editors, {\em 31st {{EACSL Annual
  Conference}} on {{Computer Science Logic}} ({{CSL}} 2023)}, volume 252 of
  {\em Leibniz {{International Proceedings}} in {{Informatics}} ({{LIPIcs}})},
  pages 6:1--6:20, {Dagstuhl, Germany}, 2023. {Schloss Dagstuhl {\textendash}
  Leibniz-Zentrum f{\"u}r Informatik}.
\newblock \href {https://doi.org/10.4230/LIPIcs.CSL.2023.6}
  {\path{doi:10.4230/LIPIcs.CSL.2023.6}}.

\bibitem[Bai76]{bainbridge1976feedback}
E.~S. Bainbridge.
\newblock Feedback and generalized logic.
\newblock {\em Information and Control}, 31(1):75--96, May 1976.
\newblock \href {https://doi.org/10.1016/S0019-9958(76)90390-9}
  {\path{doi:10.1016/S0019-9958(76)90390-9}}.

\bibitem[BE15]{baez2015categories}
John~C. Baez and Jason Erbele.
\newblock Categories in {{Control}}.
\newblock {\em Theory and Applications of Categories}, 30(24):836--881, May
  2015.
\newblock \href {https://doi.org/10.48550/arXiv.1405.6881}
  {\path{doi:10.48550/arXiv.1405.6881}}.

\bibitem[BGK{\etalchar{+}}17]{bonchi2017confluence}
Filippo Bonchi, Fabio Gadducci, Aleks Kissinger, Pawe{\l} Soboci{\'n}ski, and
  Fabio Zanasi.
\newblock Confluence of {{Graph Rewriting}} with {{Interfaces}}.
\newblock In Hongseok Yang, editor, {\em Programming {{Languages}} and
  {{Systems}}}, Lecture {{Notes}} in {{Computer Science}}, pages 141--169,
  {Berlin, Heidelberg}, 2017. {Springer}.
\newblock \href {https://doi.org/10.1007/978-3-662-54434-1_6}
  {\path{doi:10.1007/978-3-662-54434-1_6}}.

\bibitem[BGK{\etalchar{+}}22a]{bonchi2022string}
Filippo Bonchi, Fabio Gadducci, Aleks Kissinger, Pawel Soboci{\'n}ski, and
  Fabio Zanasi.
\newblock String {{Diagram Rewrite Theory I}}: {{Rewriting}} with {{Frobenius
  Structure}}.
\newblock {\em Journal of the ACM}, 69(2):14:1--14:58, March 2022.
\newblock \href {https://doi.org/10.1145/3502719} {\path{doi:10.1145/3502719}}.

\bibitem[BGK{\etalchar{+}}22b]{bonchi2022stringa}
Filippo Bonchi, Fabio Gadducci, Aleks Kissinger, Pawel Soboci{\'n}ski, and
  Fabio Zanasi.
\newblock String diagram rewrite theory {{II}}: {{Rewriting}} with symmetric
  monoidal structure.
\newblock {\em Mathematical Structures in Computer Science}, 32(4):511--541,
  April 2022.
\newblock \href {https://doi.org/10.1017/S0960129522000317}
  {\path{doi:10.1017/S0960129522000317}}.

\bibitem[BGK{\etalchar{+}}22c]{bonchi2022stringb}
Filippo Bonchi, Fabio Gadducci, Aleks Kissinger, Pawe{\l} Soboci{\'n}ski, and
  Fabio Zanasi.
\newblock String diagram rewrite theory {{III}}: {{Confluence}} with and
  without {{Frobenius}}.
\newblock {\em Mathematical Structures in Computer Science}, 32(7):1--41, June
  2022.
\newblock \href {https://doi.org/10.1017/S0960129522000123}
  {\path{doi:10.1017/S0960129522000123}}.

\bibitem[BP22]{boisseau2022graphical}
Guillaume Boisseau and Robin Piedeleu.
\newblock Graphical {{Piecewise-Linear Algebra}}.
\newblock In {\em Foundations of {{Software Science}} and {{Computation
  Structures}}: 25th {{International Conference}}, {{FOSSACS}} 2022, {{Held}}
  as {{Part}} of the {{European Joint Conferences}} on {{Theory}} and
  {{Practice}} of {{Software}}, {{ETAPS}} 2022, {{Munich}}, {{Germany}},
  {{April}} 2{\textendash}7, 2022, {{Proceedings}}}, pages 101--119, {Berlin,
  Heidelberg}, April 2022. {Springer-Verlag}.
\newblock \href {https://doi.org/10.1007/978-3-030-99253-8_6}
  {\path{doi:10.1007/978-3-030-99253-8_6}}.

\bibitem[BPSZ19]{bonchi2019graphical}
Filippo Bonchi, Robin Piedeleu, Pawel Soboci{\'n}ski, and Fabio Zanasi.
\newblock Graphical {{Affine Algebra}}.
\newblock In {\em 2019 34th {{Annual ACM}}/{{IEEE Symposium}} on {{Logic}} in
  {{Computer Science}} ({{LICS}})}, pages 1--12, June 2019.
\newblock \href {https://doi.org/10.1109/LICS.2019.8785877}
  {\path{doi:10.1109/LICS.2019.8785877}}.

\bibitem[BS22]{boisseau2022string}
Guillaume Boisseau and Pawe{\l} Soboci{\'n}ski.
\newblock String {{Diagrammatic Electrical Circuit Theory}}.
\newblock {\em Electronic Proceedings in Theoretical Computer Science},
  372:178--191, November 2022.
\newblock \href {https://arxiv.org/abs/2106.07763} {\path{arXiv:2106.07763}},
  \href {https://doi.org/10.4204/EPTCS.372.13}
  {\path{doi:10.4204/EPTCS.372.13}}.

\bibitem[BSZ14]{bonchi2014categorical}
Filippo Bonchi, Pawe{\l} Soboci{\'n}ski, and Fabio Zanasi.
\newblock A {{Categorical Semantics}} of {{Signal Flow Graphs}}.
\newblock In Paolo Baldan and Daniele Gorla, editors, {\em {{CONCUR}} 2014
  {\textendash} {{Concurrency Theory}}}, Lecture {{Notes}} in {{Computer
  Science}}, pages 435--450, {Berlin, Heidelberg}, 2014. {Springer}.
\newblock \href {https://doi.org/10.1007/978-3-662-44584-6_30}
  {\path{doi:10.1007/978-3-662-44584-6_30}}.

\bibitem[BSZ15]{bonchi2015full}
Filippo Bonchi, Pawel Soboci{\'n}ski, and Fabio Zanasi.
\newblock Full {{Abstraction}} for {{Signal Flow Graphs}}.
\newblock {\em ACM SIGPLAN Notices}, 50(1):515--526, January 2015.
\newblock \href {https://doi.org/10.1145/2775051.2676993}
  {\path{doi:10.1145/2775051.2676993}}.

\bibitem[BSZ17]{bonchi2017interacting}
Filippo Bonchi, Pawe{\l} Soboci{\'n}ski, and Fabio Zanasi.
\newblock Interacting {{Hopf}} algebras.
\newblock {\em Journal of Pure and Applied Algebra}, 221(1):144--184, January
  2017.
\newblock \href {https://doi.org/10.1016/j.jpaa.2016.06.002}
  {\path{doi:10.1016/j.jpaa.2016.06.002}}.

\bibitem[CMR{\etalchar{+}}97]{corradini1997algebraic}
A.~Corradini, U.~Montanari, F.~Rossi, H.~Ehrig, R.~Heckel, and M.~L{\"o}we.
\newblock Algebraic approaches to graph transformation ? part i: Basic concepts
  and double pushout approach.
\newblock In {\em Handbook of {{Graph Grammars}} and {{Computing}} by {{Graph
  Transformation}}}, pages 163--245. {World Scientific}, February 1997.
\newblock \href {https://doi.org/10.1142/9789812384720_0003}
  {\path{doi:10.1142/9789812384720_0003}}.

\bibitem[C{\c S}90]{cazanescu1990new}
Virgil~Emil C{\u a}z{\u a}nescu and Gheorghe {\c S}tef{\u a}nescu.
\newblock Towards a {{New Algebraic Foundation}} of {{Flowchart Scheme
  Theory}}.
\newblock {\em Fundamenta Informaticae}, 13(2):171--210, January 1990.
\newblock \href {https://doi.org/10.3233/FI-1990-13204}
  {\path{doi:10.3233/FI-1990-13204}}.

\bibitem[C{\c S}94]{cazanescu1994feedback}
Virgil~Emil C{\u a}z{\u a}nescu and Gheorghe {\c S}tef{\u a}nescu.
\newblock Feedback, {{Iteration}}, and {{Repetition}}.
\newblock In {\em Mathematical {{Aspects}} of {{Natural}} and {{Formal
  Languages}}}, volume Volume 43 of {\em World {{Scientific Series}} in
  {{Computer Science}}}, pages 43--61. {World Scientific}, October 1994.
\newblock \href {https://doi.org/10.1142/9789814447133_0003}
  {\path{doi:10.1142/9789814447133_0003}}.

\bibitem[DK13]{dixon2013opengraphs}
Lucas Dixon and Aleks Kissinger.
\newblock Open-graphs and monoidal theories.
\newblock {\em Mathematical Structures in Computer Science}, 23(2):308--359,
  April 2013.
\newblock \href {https://doi.org/10.1017/S0960129512000138}
  {\path{doi:10.1017/S0960129512000138}}.

\bibitem[EK76]{ehrig1976parallelism}
Hartmut Ehrig and Hans-J{\"o}rg Kreowski.
\newblock Parallelism of manipulations in multidimensional information
  structures.
\newblock In Antoni Mazurkiewicz, editor, {\em Mathematical {{Foundations}} of
  {{Computer Science}} 1976}, Lecture {{Notes}} in {{Computer Science}}, pages
  284--293, {Berlin, Heidelberg}, 1976. {Springer}.
\newblock \href {https://doi.org/10.1007/3-540-07854-1_188}
  {\path{doi:10.1007/3-540-07854-1_188}}.

\bibitem[FL23]{fritz2023free}
Tobias Fritz and Wendong Liang.
\newblock Free gs-{{Monoidal Categories}} and {{Free Markov Categories}}.
\newblock {\em Applied Categorical Structures}, 31(2):21, April 2023.
\newblock \href {https://doi.org/10.1007/s10485-023-09717-0}
  {\path{doi:10.1007/s10485-023-09717-0}}.

\bibitem[Fox76]{fox1976coalgebras}
Thomas Fox.
\newblock Coalgebras and cartesian categories.
\newblock {\em Communications in Algebra}, 4(7):665--667, January 1976.
\newblock \href {https://doi.org/10.1080/00927877608822127}
  {\path{doi:10.1080/00927877608822127}}.

\bibitem[FS19]{fong2019hypergraph}
Brendan Fong and David~I. Spivak.
\newblock Hypergraph categories.
\newblock {\em Journal of Pure and Applied Algebra}, 223(11):4746--4777,
  November 2019.
\newblock \href {https://doi.org/10.1016/j.jpaa.2019.02.014}
  {\path{doi:10.1016/j.jpaa.2019.02.014}}.

\bibitem[FSR16]{fong2016categorical}
Brendan Fong, Pawe{\l} Soboci{\'n}ski, and Paolo Rapisarda.
\newblock A categorical approach to open and interconnected dynamical systems.
\newblock In {\em Proceedings of the 31st {{Annual ACM}}/{{IEEE Symposium}} on
  {{Logic}} in {{Computer Science}}}, {{LICS}} '16, pages 495--504, {New York,
  NY, USA}, July 2016. {Association for Computing Machinery}.
\newblock \href {https://doi.org/10.1145/2933575.2934556}
  {\path{doi:10.1145/2933575.2934556}}.

\bibitem[GK23]{ghica2023rewriting}
Dan~R. Ghica and George Kaye.
\newblock Rewriting {{Modulo Traced Comonoid Structure}}.
\newblock In Marco Gaboardi and Femke {van Raamsdonk}, editors, {\em 8th
  {{International Conference}} on {{Formal Structures}} for {{Computation}} and
  {{Deduction}} ({{FSCD}} 2023)}, volume 260 of {\em Leibniz {{International
  Proceedings}} in {{Informatics}} ({{LIPIcs}})}, pages 14:1--14:21, {Dagstuhl,
  Germany}, 2023. {Schloss Dagstuhl {\textendash} Leibniz-Zentrum f{\"u}r
  Informatik}.
\newblock \href {https://doi.org/10.4230/LIPIcs.FSCD.2023.14}
  {\path{doi:10.4230/LIPIcs.FSCD.2023.14}}.

\bibitem[GKS24]{ghica2024fully}
Dan~R. Ghica, George Kaye, and David Sprunger.
\newblock A {{Fully Compositional Theory}} of {{Sequential Digital Circuits}}:
  {{Denotational}}, {{Operational}} and {{Algebraic Semantics}}.
\newblock (arXiv:2201.10456), January 2024.
\newblock \href {https://doi.org/10.48550/arXiv.2201.10456}
  {\path{doi:10.48550/arXiv.2201.10456}}.

\bibitem[Has97]{hasegawa1997recursion}
Masahito Hasegawa.
\newblock Recursion from cyclic sharing: {{Traced}} monoidal categories and
  models of cyclic lambda calculi.
\newblock In Philippe {de Groote} and J.~Roger~Hindley, editors, {\em Typed
  {{Lambda Calculi}} and {{Applications}}}, Lecture {{Notes}} in {{Computer
  Science}}, pages 196--213, {Berlin, Heidelberg}, 1997. {Springer}.
\newblock \href {https://doi.org/10.1007/3-540-62688-3_37}
  {\path{doi:10.1007/3-540-62688-3_37}}.

\bibitem[Has03]{hasegawa2003uniformity}
Masahito Hasegawa.
\newblock The {{Uniformity Principle}} on {{Traced Monoidal Categories}}.
\newblock {\em Electronic Notes in Theoretical Computer Science}, 69:137--155,
  February 2003.
\newblock \href {https://doi.org/10.1016/S1571-0661(04)80563-2}
  {\path{doi:10.1016/S1571-0661(04)80563-2}}.

\bibitem[HHP08]{hasegawa2008finite}
Masahito Hasegawa, Martin Hofmann, and Gordon Plotkin.
\newblock Finite {{Dimensional Vector Spaces Are Complete}} for {{Traced
  Symmetric Monoidal Categories}}.
\newblock In Arnon Avron, Nachum Dershowitz, and Alexander Rabinovich, editors,
  {\em Pillars of {{Computer Science}}: {{Essays Dedicated}} to {{Boris}}
  ({{Boaz}}) {{Trakhtenbrot}} on the {{Occasion}} of {{His}} 85th
  {{Birthday}}}, Lecture {{Notes}} in {{Computer Science}}, pages 367--385.
  {Springer}, {Berlin, Heidelberg}, 2008.
\newblock \href {https://doi.org/10.1007/978-3-540-78127-1_20}
  {\path{doi:10.1007/978-3-540-78127-1_20}}.

\bibitem[HJKS11]{heumuller2011construction}
Marvin Heum{\"u}ller, Salil Joshi, Barbara K{\"o}nig, and Jan St{\"u}ckrath.
\newblock Construction of {{Pushout Complements}} in the {{Category}} of
  {{Hypergraphs}}.
\newblock {\em Electronic Communications of the EASST}, 39(0), September 2011.
\newblock \href {https://doi.org/10.14279/tuj.eceasst.39.647}
  {\path{doi:10.14279/tuj.eceasst.39.647}}.

\bibitem[JS91]{joyal1991geometry}
Andr{\'e} Joyal and Ross Street.
\newblock The geometry of tensor calculus, {{I}}.
\newblock {\em Advances in Mathematics}, 88(1):55--112, 1991.
\newblock \href {https://doi.org/10.1016/0001-8708(91)90003-P}
  {\path{doi:10.1016/0001-8708(91)90003-P}}.

\bibitem[JSV96]{joyal1996traced}
Andr{\'e} Joyal, Ross Street, and Dominic Verity.
\newblock Traced monoidal categories.
\newblock {\em Mathematical Proceedings of the Cambridge Philosophical
  Society}, 119(3):447--468, April 1996.
\newblock \href {https://doi.org/10.1017/S0305004100074338}
  {\path{doi:10.1017/S0305004100074338}}.

\bibitem[Kis12]{kissinger2012pictures}
Aleks Kissinger.
\newblock {\em Pictures of {{Processes}}: {{Automated Graph Rewriting}} for
  {{Monoidal Categories}} and {{Applications}} to {{Quantum Computing}}}.
\newblock PhD thesis, University of Oxford, March 2012.
\newblock \href {https://arxiv.org/abs/1203.0202} {\path{arXiv:1203.0202}},
  \href {https://doi.org/10.48550/arXiv.1203.0202}
  {\path{doi:10.48550/arXiv.1203.0202}}.

\bibitem[Kis14]{kissinger2014abstract}
Aleks Kissinger.
\newblock Abstract {{Tensor Systems}} as {{Monoidal Categories}}.
\newblock In Claudia Casadio, Bob Coecke, Michael Moortgat, and Philip Scott,
  editors, {\em Categories and {{Types}} in {{Logic}}, {{Language}}, and
  {{Physics}}: {{Essays Dedicated}} to {{Jim Lambek}} on the {{Occasion}} of
  {{His}} 90th {{Birthday}}}, Lecture {{Notes}} in {{Computer Science}}, pages
  235--252. {Springer}, {Berlin, Heidelberg}, 2014.
\newblock \href {https://doi.org/10.1007/978-3-642-54789-8_13}
  {\path{doi:10.1007/978-3-642-54789-8_13}}.

\bibitem[KL80]{kelly1980coherence}
G.~M. Kelly and M.~L. Laplaza.
\newblock Coherence for compact closed categories.
\newblock {\em Journal of Pure and Applied Algebra}, 19:193--213, December
  1980.
\newblock \href {https://doi.org/10.1016/0022-4049(80)90101-2}
  {\path{doi:10.1016/0022-4049(80)90101-2}}.

\bibitem[Law63]{lawvere1963functorial}
F.~William Lawvere.
\newblock Functorial semantics of algebraic theories.
\newblock {\em Proceedings of the National Academy of Sciences},
  50(5):869--872, November 1963.
\newblock \href {https://doi.org/10.1073/pnas.50.5.869}
  {\path{doi:10.1073/pnas.50.5.869}}.

\bibitem[LS04]{lack2004adhesive}
Stephen Lack and Pawe{\l} Soboci{\'n}ski.
\newblock Adhesive {{Categories}}.
\newblock In Igor Walukiewicz, editor, {\em Foundations of {{Software Science}}
  and {{Computation Structures}}}, Lecture {{Notes}} in {{Computer Science}},
  pages 273--288, {Berlin, Heidelberg}, 2004. {Springer}.
\newblock \href {https://doi.org/10.1007/978-3-540-24727-2_20}
  {\path{doi:10.1007/978-3-540-24727-2_20}}.

\bibitem[LS05]{lack2005adhesive}
Stephen Lack and Pawe{\l} Soboci{\'n}ski.
\newblock Adhesive and quasiadhesive categories.
\newblock {\em RAIRO - Theoretical Informatics and Applications},
  39(3):511--545, July 2005.
\newblock \href {https://doi.org/10.1051/ita:2005028}
  {\path{doi:10.1051/ita:2005028}}.

\bibitem[Mac65]{maclane1965categorical}
Saunders MacLane.
\newblock Categorical algebra.
\newblock {\em Bulletin of the American Mathematical Society}, 71(1):40--106,
  1965.
\newblock \href {https://doi.org/10.1090/S0002-9904-1965-11234-4}
  {\path{doi:10.1090/S0002-9904-1965-11234-4}}.

\bibitem[MPZ23]{milosavljevic2023string}
Aleksandar Milosavljevic, Robin Piedeleu, and Fabio Zanasi.
\newblock String {{Diagram Rewriting Modulo Commutative}} ({{Co}})monoid
  {{Structure}}.
\newblock (arXiv:2204.04274), March 2023.
\newblock \href {https://doi.org/10.48550/arXiv.2204.04274}
  {\path{doi:10.48550/arXiv.2204.04274}}.

\bibitem[MS09]{macdonald2009amalgamations}
John MacDonald and Laura Scull.
\newblock Amalgamations of {{Categories}}.
\newblock {\em Canadian Mathematical Bulletin}, 52(2):273--284, June 2009.
\newblock \href {https://doi.org/10.4153/CMB-2009-030-5}
  {\path{doi:10.4153/CMB-2009-030-5}}.

\bibitem[Pen71]{penrose1971applications}
Roger Penrose.
\newblock Applications of negative dimensional tensors.
\newblock In {\em Combinatorial {{Mathematics}} and {{Its Applications}}},
  pages 221--244. {Academic Press}, 1971.

\bibitem[PZ21]{piedeleu2021string}
Robin Piedeleu and Fabio Zanasi.
\newblock A {{String Diagrammatic Axiomatisation}} of {{Finite-State
  Automata}}.
\newblock {\em Foundations of Software Science and Computation Structures},
  12650:469--489, March 2021.
\newblock \href {https://doi.org/10.1007/978-3-030-71995-1_24}
  {\path{doi:10.1007/978-3-030-71995-1_24}}.

\bibitem[RSW05]{rosebrugh2005generic}
R.~Rosebrugh, N.~Sabadini, and R.~F.~C. Walters.
\newblock Generic commutative separable algebras and cospans of graphs.
\newblock {\em Theory and Applications of Categories}, 15:164--177, 2005.

\bibitem[Sel11]{selinger2011survey}
Peter Selinger.
\newblock A {{Survey}} of {{Graphical Languages}} for {{Monoidal Categories}}.
\newblock In Bob Coecke, editor, {\em New {{Structures}} for {{Physics}}},
  Lecture {{Notes}} in {{Physics}}, pages 289--355. {Springer}, {Berlin,
  Heidelberg}, 2011.
\newblock \href {https://doi.org/10.1007/978-3-642-12821-9_4}
  {\path{doi:10.1007/978-3-642-12821-9_4}}.

\bibitem[Sel12]{selinger2012finite}
Peter Selinger.
\newblock Finite dimensional {{Hilbert}} spaces are complete for dagger compact
  closed categories.
\newblock {\em Logical Methods in Computer Science}, 8(3), August 2012.
\newblock \href {https://doi.org/10.2168/LMCS-8(3:6)2012}
  {\path{doi:10.2168/LMCS-8(3:6)2012}}.

\bibitem[SJ99]{schweimeier1999categorical}
Ralf Schweimeier and Alan Jeffrey.
\newblock A {{Categorical}} and {{Graphical Treatment}} of {{Closure
  Conversion}}.
\newblock {\em Electronic Notes in Theoretical Computer Science}, 20:481--511,
  January 1999.
\newblock \href {https://doi.org/10.1016/S1571-0661(04)80090-2}
  {\path{doi:10.1016/S1571-0661(04)80090-2}}.

\bibitem[WGZ23]{wilson2023string}
Paul Wilson, Dan Ghica, and Fabio Zanasi.
\newblock String {{Diagrams}} for {{Non-Strict Monoidal Categories}}.
\newblock In Bartek Klin and Elaine Pimentel, editors, {\em 31st {{EACSL Annual
  Conference}} on {{Computer Science Logic}} ({{CSL}} 2023)}, volume 252 of
  {\em Leibniz {{International Proceedings}} in {{Informatics}} ({{LIPIcs}})},
  pages 37:1--37:19, {Dagstuhl, Germany}, 2023. {Schloss Dagstuhl {\textendash}
  Leibniz-Zentrum f{\"u}r Informatik}.
\newblock \href {https://doi.org/10.4230/LIPIcs.CSL.2023.37}
  {\path{doi:10.4230/LIPIcs.CSL.2023.37}}.

\bibitem[Zan15]{zanasi2015interacting}
Fabio Zanasi.
\newblock {\em Interacting {{Hopf Algebras}}: The Theory of Linear Systems}.
\newblock PhD thesis, University of Lyon, October 2015.
\newblock \href {https://doi.org/10.48550/arXiv.1805.03032. arXiv: 1805.03032}
  {\path{doi:10.48550/arXiv.1805.03032. arXiv: 1805.03032}}.

\bibitem[Șt90]{stefanescu1990feedback}
Gheorghe Ștef{\u a}nescu.
\newblock Feedback {{Theories}} ({{A Calculus}} for {{Isomorphism Classes}} of
  {{Flowchart Schemes}}).
\newblock {\em Romanian Journal of Pure and Applied Mathematics}, 35(1):73--79,
  1990.

\end{thebibliography}
\end{document}